\documentclass[runningheads]{llncs}

\pdfoutput=1

\input{llncs.sty}
\usepackage[T1]{fontenc}

\begin{document}

\title{Improved Synthesis of Toffoli-Hadamard Circuits}

\author{Matthew Amy \inst{1}\orcidlink{0000-0003-3514-420X} \and
  Andrew N. Glaudell \inst{2}\orcidlink{0000-0001-9824-5804} \and
  Sarah Meng Li \inst{3}\orcidlink{0000-0002-1415-4145} \and Neil
  J. Ross\inst{4}\orcidlink{0000-0003-0941-4333}}

\authorrunning{M. Amy, A. N. Glaudell, S. M. Li, N. J. Ross}

\institute{School of Computing Science, Simon Fraser University\\
\email{matt\_amy@sfu.ca}\and
Photonic Inc.\\
\email{andrewglaudell@gmail.com}\and
Institute for Quantum Computing, University of Waterloo\\
\email{sarah.li@uwaterloo.ca}\and
Department of Mathematics and Statistics, Dalhousie University\\
\email{neil.jr.ross@dal.ca}}

\maketitle

\begin{abstract}
The matrices that can be exactly represented by a circuit over the
Toffoli-Hadamard gate set are the orthogonal matrices of the form
$M/\sqrt{2}{}^k$, where $M$ is an integer matrix and $k$ is a
nonnegative integer. The exact synthesis problem for this gate set is
the problem of constructing a circuit for a given such
matrix. Existing methods produce circuits consisting of
$O(2^n\log(n)k)$ gates, where $n$ is the dimension of the matrix. In
this paper, we provide two improved synthesis methods. First, we show
that a technique introduced by Kliuchnikov in 2013 for Clifford+$T$
circuits can be straightforwardly adapted to Toffoli-Hadamard
circuits, reducing the complexity of the synthesized circuit from
$O(2^n\log(n)k)$ to $O(n^2\log(n)k)$. Then, we present an alternative
synthesis method of similarly improved cost, but whose application is
restricted to circuits on no more than three qubits. Our results also
apply to orthogonal matrices over the dyadic fractions, which
correspond to circuits using the 2-qubit gate $H\otimes H$, rather
than the usual single-qubit Hadamard gate $H$.

\end{abstract}

\section{Introduction}
\label{sec:intro}

Recent experimental progress has made it possible to carry out large
computational tasks on quantum computers faster than on
state-of-the-art classical supercomputers
\cite{arute2019quantum,zhu2022quantum}. However, qubits are incredibly
sensitive to decoherence, which leads to the degradation of quantum
information. Moreover, physical gates implemented on real quantum
devices have poor gate fidelity, so that every additional gate in a
circuit introduces a small error to the computation. To harness the
full power of quantum computing, it is therefore crucial to design
resource-efficient compilation techniques.

Over the past decade, researchers have taken advantage of a
correspondence between quantum circuits and matrices of elements from
algebraic number rings
\cite{AGR2019,fgkm15,GS13,KMM-exact,kliuchnikov2015framework}. This
number-theoretic perspective can reveal important properties of gate
sets and has resulted in several improved synthesis protocols. An
important instance of this correspondence occurs in the study of the
Toffoli-Hadamard gate set $\s{X, CX, CCX, H}$
\cite{AGR2019,gajewski2009analysis}. Circuits over this gate set
correspond exactly to orthogonal matrices of the form
$M/\sqrt{2}{}^k$, where $M$ is an integer matrix and $k$ is a
nonnegative integer. A closely related instance of this correspondence
arises with the gate set $\s{X, CX, CCX, H\otimes H}$, where the
2-qubit gate $H\otimes H$ replaces the usual single-qubit Hadamard
gate $H$. Circuits over this second gate set correspond exactly to
orthogonal matrices over the ring of dyadic fractions $\Z[1/2]$
\cite{AGR2019}. The Toffoli-Hadamard gate set is arguably the simplest
universal gate set for quantum computation
\cite{aharonov2003,shi2002}; the corresponding circuits have been
studied in the context of diagrammatic calculi \cite{vilmart2018zx},
path-sums \cite{vilmart2023completeness}, and quantum logic
\cite{dalla2013toffoli}, and play a critical role in quantum error
correction
\cite{cory1998experimental,fedorov2012implementation,reed2012realization},
fault-tolerant quantum computing
\cite{gottesman1997,paetznick2013universal,yoder2017universal}, and
the quantum Fourier transform \cite{nielsen2001}.

In this paper, we leverage the number-theoretic structure of the
aforementioned circuits to design improved synthesis algorithms. Our
approach is to focus on the matrix groups associated with the gate
sets $\s{X, CX, CCX, H}$ and $\s{X, CX, CCX, H\otimes H}$. For each
group, we use a convenient set of generators and study the
factorization of group elements into products of these
generators. Because each generator can be expressed as a short
circuit, a good solution to this factorization problem yields a good
synthesis algorithm.

Exact synthesis algorithms for Toffoli-Hadamard circuits were
introduced in \cite{gajewski2009analysis} and independently in
\cite{AGR2019}. We refer to the algorithm of \cite{AGR2019}, which we
take as our baseline, as the \emph{local synthesis algorithm} because
it factors the input matrix column by column. This algorithm produces
circuits of $O(2^n\log(n)k)$ gates, where $n$ is the dimension of the
input matrix.

We propose two improved synthesis methods. The first, which we call
the \emph{Householder synthesis algorithm}, is an adaptation to the
Toffoli-Hadamard gate set of the technique introduced by Kliuchnikov
in \cite{kliuchnikov2013synthesis} for the Clifford+$T$ gate set. This
algorithm proceeds by embedding the input matrix in a larger one, and
then expressing this larger matrix as a product of Householder
reflections. The Householder synthesis algorithm produces circuits of
$O(n^2\log(n)k)$ gates. We then introduce a \textit{global synthesis
  algorithm}, inspired by the work of Russell in
\cite{russell2014exact} and of Niemann, Wille, and Drechsler in
\cite{niemann2020advanced}. In contrast to the local synthesis
algorithm, which proceeds one column at a time, the global algorithm
considers the input matrix in its entirety. In its current form, this
last algorithm is restricted to matrices of dimensions $2$, $4$, and
$8$. As a result of this restriction, the dimension of the input
matrix can be dropped in the asymptotic analysis, and the circuits
produced by the global algorithm consist of $O(k)$ gates.

The rest of this paper is organized as follows. In
\cref{sec:esp_on_ln}, we introduce the exact synthesis problem, as
well as the matrices, groups, and rings that will be used throughout
the paper. In \cref{sec:AGR_algo}, we review the local synthesis
algorithm of \cite{AGR2019}. The Householder synthesis algorithm and
the global synthesis algorithm are discussed in
\cref{sec:householder_algo,sec:globalsyn_algo}, respectively. We
conclude in \cref{sec:conclusion}.

\section{The Exact Synthesis Problem}
\label{sec:esp_on_ln}

In this section, we introduce the \emph{exact synthesis problem}. We
start by defining the matrices, groups, and rings that will be used in
the rest of the paper.

\begin{definition}
  \label{def:Ztwo}
  The \emph{ring of dyadic fractions} is defined as $\Ztwo=\left\{
  u/2^k; \ u\in \Z, k\in\N \right\}$.
\end{definition}

\begin{definition}
  \label{def:OZtwo}
  $\OD$ is the \emph{group of $n$-dimensional orthogonal dyadic
  matrices}. It consists of the $n\times n$ orthogonal matrices of the
  form $M/2^k$, where $M$ is an integer matrix and $k$ is a
  nonnegative integer. For brevity, we denote this group by
  $\integral_n$.
\end{definition}

\begin{definition}
  \label{def:OZroottwo}
  $\supin_n$ is the group of \emph{$n$-dimensional orthogonal scaled
  dyadic matrices}. It consists of the $n\times n$ orthogonal matrices
  of the form $M/\sqrt{2}{}^k$, where $M$ is an integer matrix and $k$
  is a nonnegative integer.
\end{definition}

$\integral_n$ is infinite if and only if $n\geq 5$. Moreover,
$\integral_n$ is a subgroup of $\supin_n$. When $n$ is odd, we in fact
have $\integral_n=\supin_n$ \cite[Lemma~5.9]{AGR2019}. When $n$ is
even, $\integral_n$ is a subgroup of $\supin_n$ of index 2. As a
result, it is also the case that $\supin_n$ is infinite if and only if
$n\geq 5$.

\begin{definition}
  \label{def:lde}
  Let $t \in \Ztwo$. A natural number $k$ is a \emph{denominator
  exponent} of $t$ if $2^k t \in \Z$. The least such $k$ is called the
  \emph{least denominator exponent} of $t$, and is denoted by
  $\lde(t)$.
\end{definition}

\begin{definition}
  \label{def:scaledlde}
  Let $t = u/\sqrt{2}^k$, where $u \in \Z$ and $k \in \N$. A natural
  number $k$ is a \emph{scaled denominator exponent} of $t$ if
  $\sqrt{2}{}^kt\in \Z$. The least such $k$ is called the \emph{least
  scaled denominator exponent} of $t$, and is denoted by
  $\lde_{\sqrt{2}}(t)$.
\end{definition}

We extend \cref{def:lde,def:scaledlde} to matrices with appropriate
entries as follows. A natural number $k$ is a \textit{(scaled)
  denominator exponent} of a matrix $M$ if it is a (scaled)
denominator exponent of all of the entries of $M$. Similarly, the
least such $k$ is called the \textit{least (scaled) denominator
  exponent} of $M$. We denote the least denominator exponent of $M$
and the least scaled denominator of $M$ by $\lde(M)$ and
$\lde_{\sqrt{2}}(M)$, respectively.

We now leverage some well-known quantum gates to define generators for
$\integral_n$ and $\supin_n$.

\begin{definition}
  \label{def:basegens}
  The matrices \emph{$(-1)$}, \emph{$X$}, \emph{$CX$}, \emph{$CCX$},
  \emph{$H$}, and \emph{$K$} are defined as follows:
  \[
  \left(-1\right) = [-1],
  \qquad 
  X = \begin{bmatrix}\ 0 \ & \ 1 \ \\ \ 1 \ & \ 0 \ \end{bmatrix}, 
  \qquad
  H =\frac{1}{\sqrt{2}}\begin{bmatrix} \ 1 \ & \ \ 1 \ \\ 1 & -1\end{bmatrix},
  \]
  $CX = \diag(I_2,X)$, $CCX = \diag(I_6,X)$, and $K = H\otimes H$.
\end{definition}

The matrix $(-1)$ is a scalar. The matrices $X$, $CX$, and $CCX$ are
known as the \emph{NOT}, \emph{CNOT}, and \emph{Toffoli} gates,
respectively, while the matrix $H$ is the \emph{Hadamard} gate. In
\cref{def:basegens}, $CX$ and $CCX$ are defined as block matrices,
while $K$ is defined as the twofold tensor product of $H$ with
itself. Below we explicitly write out matrices for $CX$, $CCX$, and
$K$:
\[
CX = \begin{bmatrix}
 \ 1 \ & \ 0 \ & \ 0 \ & \ 0 \ \\
  0 & 1 & 0 & 0 \\
  0 & 0 & 0 & 1 \\
  0 & 0 & 1 & 0
\end{bmatrix},
\qquad
CCX = \begin{bmatrix} 
  \ 1\ & \ 0 \ & \ 0\ & \ 0\ & \ 0\ & \ 0\ & \ 0\ & \ 0\ \\ 
  0 & 1 & 0 & 0 & 0 & 0 & 0 & 0 \\ 
  0 & 0 & 1 & 0 & 0 & 0 & 0 & 0 \\ 
  0 & 0 & 0 & 1 & 0 & 0 & 0 & 0 \\ 
  0 & 0 & 0 & 0 & 1 & 0 & 0 & 0 \\ 
  0 & 0 & 0 & 0 & 0 & 1 & 0 & 0 \\ 
  0 & 0 & 0 & 0 & 0 & 0 & 0 & 1 \\ 
  0 & 0 & 0 & 0 & 0 & 0 & 1 & 0 
\end{bmatrix},
\qquad
K =\frac{1}{2}\begin{bmatrix}
  1 & \ \ 1 & \ \ 1 & \ \ 1 \\
  1 & -1 & \ \ 1 & -1 \\
  1 & \ \ 1 & -1 & -1 \\
  1 & -1 & -1 & \ \ 1
\end{bmatrix}.
\]             

\begin{definition}
  \label{def:onetwolevel}
  Let $M$ be an $m\times m$ matrix, $m\leq n$, and $0 \leq a_0<\ldots
  <a_{m-1}< n$. The \emph{$m$-level matrix of type $M$} is the
  $n\times n$ matrix $M_{[a_0,\ldots,a_{m-1}]}$ defined by
  \[
    M_{{[a_0,\ldots,a_{m-1}]}_{i,j}} = 
  \begin{cases}
      M_{i',j'} \text{ if } i = a_{i'} \text{ and }j=a_{j'}\\
      I_{i,j} \text{ otherwise.}
    \end{cases}
\]
\end{definition}

The dimension $n$ of the matrix $M_{[a_0,\ldots,a_{m-1}]}$ is left
implicit most of the time, as it can often be inferred from the
context. As an example, the 2-level matrix $H_{[0,2]}$ of dimension 4
is
\[
H_{[0,2]} =\frac{1}{\sqrt{2}}\begin{bmatrix} 
1 & \ \ 0 & \ \ 1 & \ \ 0 \\
0 & \sqrt{2} & \ \ 0 & \ \ 0 \\
1 & \ \ 0 &  -1 & \ \ 0 \\
0 & \ \ 0 & \ \ 0 & \sqrt{2}
\end{bmatrix}.
\]

\begin{definition}
  \label{def:gens}  
  The set $\gens_n$ of \emph{$n$-dimensional generators of
  $\integral_n$} is the subset of $\integral_n$ defined as
  \[
  \gens_n =\s{\mone{a},\xx{a,b},\hh{a,b,c,d}~;~ 0\leq a<b<c<d < n}.
  \]
\end{definition}

\begin{definition}
\label{def:supgens} 
  The set $\gensup_n$ of \emph{$n$-dimensional generators of
  $\supin_n$} is the subset of $\supin_n$ defined as
  $\gensup_n=\gens_n$ when $n$ is odd and as
  \[
  \gensup_n = \left\{\mone{a},\xx{a, b}, \hh{a,b,c,d}, I_{n/2} \otimes
  H~;~ 0\leq a < b <c<d< n\right\}
  \]
  when $n$ is even.
\end{definition}

In \cref{def:supgens}, the condition on the parity of $n$ ensures that
$I_{n/2}\otimes H$ is only included when it is meaningful to do so. In
what follows, for brevity, we ignore the subscript in $I_{n/2} \otimes
H$ and simply write $I \otimes H$. It is known that $\gens_n$ and
$\gensup_n$ are indeed generating sets for $\integral_n$ and
$\supin_n$, respectively \cite{AGR2019}.

Circuits over a set $\mathbb{G}$ of quantum gates are constructed from
the elements of $\mathbb{G}$ through composition and tensor
product. Circuits can use ancillary qubits, but these must be
initialized and terminated in the computational basis state
$\ket{0}$. For example, in the diagram below the circuit $C$ uses a
single ancilla.
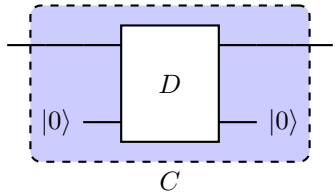
\begin{center}
    \begin{quantikz}
       \qw & \qw \gategroup[2,steps=5,style={dashed,rounded
           corners,fill=blue!20, inner xsep=2pt},background,label
         style={label
           position=below,anchor=north,yshift=-0.2cm}]{{$C$}} &\qw &
       \gate[2]{\quad D \quad} & \qw & \qw & \qw\\ & &
       \lstick{$\ket{0}$} & & \qw\rstick{$\ket{0}$} & &
    \end{quantikz}    
\end{center}
Ancillas provide additional computational space and, as we will see
below, can be useful in reducing the gate count of circuits.

If $C$ is a circuit over some gate set, we write $\interp{C}$ for the
matrix represented by $C$ and we say that $C$ \emph{represents}
$\interp{C}$. If $\mathbb{G}$ is a gate set, we write
$\mathcal{U}(\mathbb{G})$ for the collection of matrices representable
by a circuit over $\mathbb{G}$. That is, $\mathcal{U}(\mathbb{G}) =
\s{\interp{C}~;~ C \mbox{ is a circuit over } \mathbb{G}}$.

\begin{definition}  
  \label{def:exactsynth}
  The \emph{exact synthesis problem} for a gate set $\mathbb{G}$ is
  the following: given $U \in \mathcal{U}(\mathbb{G})$, find a circuit
  $C$ over $\mathbb{G}$ such that $\interp{C}=U$. A constructive
  solution to the exact synthesis problem for $\mathbb{G}$ is known as an
  \emph{exact synthesis algorithm} for $\mathbb{G}$.
\end{definition}

The \emph{Toffoli-Hadamard} gate set consists of the gates $CCX$ and
$H$. Because the Toffoli gate is universal for classical reversible
computation with ancillary bits in the $0$ or $1$ state
\cite{fredkin1982}, one can express both $X$ and $CX$ over this gate
set. As a result, and by a slight abuse of terminology, we refer to
the gate set $\{X, CX, CCX, H\}$ as the Toffoli-Hadamard gate set.

It was shown in \cite{gajewski2009analysis} and later, independently,
in \cite{AGR2019}, that the operators exactly representable by an
$m$-qubit Toffoli-Hadamard circuit are precisely the elements of
$\supin_{2^m}$. The proof of this fact takes the form of an exact
synthesis algorithm. The algorithm of \cite{AGR2019}, following prior
work of \cite{GS13}, proceeds in two steps. First, one shows that,
when $n$ is a power of 2, every operator in $\gensup_n$ can be exactly
represented by a circuit over $\s{X, CX, CCX, H}$. Then, one shows
that every element of $\supin_n$ can be factored as a product of
matrices from $\gensup_n$. Together, these two steps solve the exact
synthesis problem for the Toffoli-Hadamard gate set. By considering
the gate set $\{X, CX, CCX, K\}$, rather than the gate set $\{X, CX,
CCX, H\}$, one obtains circuits that correspond precisely to the
elements of $\integral_{n}$ \cite{AGR2019}. The exact synthesis
problem for this gate set is solved similarly, with the exact
synthesis algorithm using $\gens_n$ rather than $\gensup_n$.

Each element of $\gensup_n$ (resp. $\gens_n$) can be represented by a
circuit containing $O(\log(n))$ gates (so a constant number of gates
when $n$ is fixed). It is therefore the complexity of the
factorization of elements of $\supin_n$ (resp. $\integral_n$) into
elements of $\gensup_n$ (resp. $\gens_n$) that determines the
complexity of the overall synthesis algorithm. For this reason, in the
rest of the paper, we focus on finding improved solutions to this
factorization problem.

\section{The Local Synthesis Algorithm}
\label{sec:AGR_algo}

In this section, we revisit the solution to the exact synthesis
problem for $\s{X, CX, CCX, H}$ (and $\s{X, CX, CCX, K}$) proposed in
\cite{AGR2019}. The algorithm, which we call the \emph{local synthesis
algorithm}, is an analogue of the \emph{Giles-Selinger algorithm}
introduced in \cite{GS13} for the synthesis of Clifford+$T$
circuits. In a nutshell, the local synthesis algorithm proceeds one
column at a time, reducing each column of the input matrix to a basis
vector. This process is repeated until the input matrix is itself
reduced to the identity. The algorithm is local in the sense that the
matrix factorization is carried out column by column and that, at each
step, the algorithm only uses information about the column currently
being reduced. We now briefly recall the main points of
\cite[Section~5.1]{AGR2019} in order to better understand the
functionality of the local synthesis algorithm. We encourage the
reader to consult \cite{AGR2019} for further details.

\begin{lemma}
  \label{lem:twohs}
  Let $v_0,v_1,v_2,v_3$ be odd integers. Then there exists
  $\tau_0,\tau_1,\tau_2,\tau_3\in\Z_2$ such that
  \[
  \hh{0,1,2,3} \mone{0}^{\tau_0}\mone{1}^{\tau_1}
  \mone{2}^{\tau_2}\mone{3}^{\tau_3}
  \begin{bmatrix}
    v_0 \\ v_1 \\ v_2 \\ v_3
  \end{bmatrix} =
  \begin{bmatrix} 
    v_0' \\ v_1' \\ v_2'\\ v_3' 
  \end{bmatrix},
  \]
  where $v_0',v_1',v_2',v_3'$ are even integers.
\end{lemma}

\begin{lemma}
  \label{lem:ind}
  Let $\ket{u} \in \Ztwo^n$ be a unit vector with
  $\lde(\ket{u})=k$. Let $\ket{v} = 2^k\ket{u}$. If $k>0$, the number
  of odd entries in $\ket{v}$ is a multiple of $4$.
\end{lemma}

\begin{proof}
  Since $\braket{u|u}=1$, $\braket{v|v} =4^k$. Thus $\sum v^2_j =
  4^k$. Since the only squares modulo 4 are 0 and 1, and $v_j^2\equiv
  1 \pmod{4}$ if and only if $v_j$ is odd, the number of $v_j$ in
  $\ket{v}$ such that $v_j^2 \equiv 1 \pmod{4}$ is a multiple of 4.
\end{proof}

\cref{lem:twohs,lem:ind} imply the \emph{Column Lemma}, the crux of
the local synthesis algorithm.

\begin{lemma}[Column Lemma]
  \label{lem:column}
  Let $\ket{u}\in\Ztwo^n$ be a unit vector and $\ket{j}$ be a standard
  basis vector. There exists a sequence of generators
  $G_0,\ldots,G_q\in\gens_n$ such that $\left(G_q\cdots G_1\right)
  \ket{u}=\ket{j}$.
\end{lemma}
 
\begin{proof}
Let $k=\lde(\ket{u})$ and proceed by induction on $k$. When $k=0$,
$\ket{u}=\pm \ket{j'}$ for some $0 \leq j' < n$. Indeed, since
$\ket{u}$ is a unit vector, we have $\sum u^2_i = 1$. Since $u_i \in
\Z$, there must be exactly one $i$ such that $u_i=\pm 1$ while all the
other entries of $\ket{u}$ are 0. If $\ket{j'}=\ket{j}$ there is
nothing to do. Otherwise, map $\ket{u}$ to $\ket{j}$ by applying an
optional one-level $(-1)$ generator followed by an optional two-level
$X$ generator.  When $k>0$, by \cref{lem:ind}, the number of odd
entries in $\ket{v} = 2^k \ket{u}$ is a multiple of 4. We can then
group these odd entries into quadruples and apply \cref{lem:twohs} to
each quadruple to reduce the least denominator exponent of the
vector. By induction, we can continuously reduce $k$ until it becomes
$0$, which is the base case.
\end{proof}

\begin{proposition}
  \label{thm:membership}
  Let $U$ be an $n\times n$ matrix. Then $U\in \integral_n$ if, and
  only if, $U$ can be written as a product of elements of $\gens_n$.
\end{proposition}

\begin{proof}
The right-to-left direction follows from the fact that
$\gens_n\subseteq\integral_n$. For the converse, use \cref{lem:column}
to reduce the leftmost unfixed column $U_j$ to $\ket{j}$, $0 \leq j <
n$. After that, repeat the column reduction on the next leftmost
unfixed column until $U$ is reduced to the identity.
\end{proof}

The local synthesis algorithm establishes the left-to-right
implication of \cref{thm:membership}. It expresses an element of
$\integral_n$ as a product of generators from $\gens_n$ and thereby
solves the exact synthesis problem for $\s{X, CX, CCX, K}$. A small
extension of the algorithm shows that $\gensup_n$ generates
$\supin_n$, solving the exact synthesis problem for $\s{X, CX, CCX,
  H}$.

\begin{corollary}
  \label{cor:membership}
  Let $U$ be an $n\times n$ matrix. Then $U\in \supin_n$ if, and only
  if, $U$ can be written as a product of elements of $\gensup_n$.
\end{corollary}

\begin{proof}
  As before, the right-to-left direction follows from the fact that
  $\gensup_n \subseteq \supin_n$. Conversely, let $U\in\supin_n$ and
  write $U$ as $U = M/\sqrt{2}{}^q$, where $M$ is an integer matrix
  and $q=\lde_{\sqrt{2}}(U)$. If $q$ is even, then $U \in
  \integral_n$. By \cref{thm:membership}, $U$ can be written as a
  product of elements of $\gens_n \subset \gensup_n$. If $q$ is odd,
  then by \cite[Lemma~5.9]{AGR2019} $n$ must be even. It follows that
  $\left(I\otimes H\right)U \in \integral_n$. We can conclude by
  applying \cref{thm:membership} to $\left(I\otimes H\right)U$.
\end{proof}

In the rest of this section, we analyze the gate complexity of the
local synthesis algorithm. In the worst case, it takes exponentially
many generators in $\gens_n$ to decompose a unitary in
$\integral_n$. Since $\supin_n$ is simply a scaled version of
$\integral_n$, the same gate complexity holds for the local synthesis
of $\supin_n$ over $\gensup_n$.

\begin{lemma}
  \label{lem:helper2}
  Let $\ket{u} \in \D^n$ with $\lde(\ket{u})=k$. Let $\ket{j}$ be a
  standard basis vector. The number of generators in $\gens_n$
  required by \cref{lem:column} to reduce $\ket{u}$ to $\ket{j}$ is
  $O(nk)$.
\end{lemma}

\begin{proof}
  Let $\ket{v} = 2^k\ket{u}$, then $\ket{v} \in \Z^n$. We proceed by
  case distinction. When $k=0$, there is precisely one non-zero entry
  in $\ket{v}$, which is either $1$ or $-1$. We need at most a
  two-level $X$ gate and a one-level $(-1)$ gate to send $\ket{v}$ to
  $\ket{j}$. Hence the gate complexity over $\gens_n$ is $O(1)$. When
  $k>0$, there are odd entries in $\ket{v}$ and the number of such
  entries must be doubly-even (i.e., a multiple of 4). To reduce $k$
  by $1$ as in \cref{lem:column}, we need to make all of the odd
  entries even. By \cref{lem:twohs}, for each quadruple of odd
  entries, we need at most four one-level $(-1)$ gates and precisely
  one four-level $K$ gate. In the worst case, there are $\lfloor n/4
  \rfloor$ quadruples of odd entries in $\ket{v}$. To reduce $k$ to
  $0$, we thus need at most $(4+1)\lfloor n/4 \rfloor k \in O(nk)$
  elements of $\gens_n$. Therefore, the total number of generators in
  $\gens_n$ required by \cref{lem:column} to reduce $\ket{u}$ to
  $\ket{j}$ is $\max{(O(nk),O(1))} = O(nk)$.
\end{proof}

\begin{proposition}
  \label{prop:agrgatecount}
  Let $U \in \integral_n$ with $\lde(U)=k$. Then, using the local
  synthesis algorithm, $U$ can be represented by a product of
  $O(2^nk)$ elements of $\gens_n$ .
\end{proposition}

\begin{proof}
The local synthesis algorithm starts from the leftmost column of $U$
that is not yet reduced. In the worst case, this column is $U_0$ and
$\lde(U_0)=k$. By \cref{lem:helper2}, we need $O(nk)$ generators in
$\gens_n$ to reduce $U_0$ to $\ket{0}$. While reducing $U_0$, the
local synthesis algorithm may increase the least denominator exponent
of the other columns of $U$. Each row operation potentially increases
the least denominator exponent by $1$. Therefore, the least
denominator exponent of any other column in $U$ may increase to $2k$
during the reduction of $U_0$. Now let $f_{U_i}$ be the cost of
reducing $U_i$ to $\ket{i}$. As the algorithm proceeds from the left
to the right of $U$, $f_{U_i}$ increases as shown below.
\[
f_{U_0} \in O\left(nk\right), \quad f_{U_1} \in O\left((n-1)2k\right),
\quad f_{U_2} \in O\left((n-2)2^2k\right), \quad \ldots, \quad
f_{U_{n-1}} \in O\left(2^{n-1}k\right).
\]
In total, the number of generators from $\gens_n$ that are required to
synthesize $U$ is
\begin{equation}
 S_n =\ \sum_{i=0}^{n-1}f_{U_{i}} = \sum_{i=0}^{n-1}
 (n-i)2^{i}k. \label{sec3:equ:1}
\end{equation}
Multiplying both sides of \Cref{sec3:equ:1} by $2$ yields
\begin{equation}
2S_n = \left(2n + (n-1)2^2 + (n-2)2^3 + (n-3)2^4 + \ldots +
2^{n}\right)k.\label{sec3:equ:2}
\end{equation}
Subtracting \Cref{sec3:equ:1} from \Cref{sec3:equ:2} yields
\[
S_n =\ \left(-n + 2 + 2^2 + \ldots + 2^{n-1} + 2^n\right)k=\left(-n +
2^{n+1}-2\right)k\in O(2^nk).
\]
Hence, the complexity of the local synthesis algorithm of
$\integral_n$ over $\gens_n$ is $O(2^{n}k)$.
\end{proof}

\begin{corollary}
  \label{coro:agrgatecount}
  Let $U \in \supin_n$ with $\lde_{\sqrt{2}}(U)=k$. Then, using the
  local synthesis algorithm, $U$ can be represented by a product of
  $O(2^nk)$ elements of $\gensup_n$.
\end{corollary}

\begin{proof}
  When $k$ is even, $U \in \integral_n$ and, by
  \cref{prop:agrgatecount}, $U$ can be represented by $O(2^nk)$
  generators in $\gens_n \subset \gensup_n$. When $k$ is odd then, by
  \cite[Lemma~5.9]{AGR2019}, $n$ must be even so that $\left(I\otimes
  H\right)U \in \integral_n$. Applying \cref{prop:agrgatecount} to
  $\left(I\otimes H\right)U$ yields a sequence of $O(2^nk)$ generators
  over $\gens_n$ for $\left(I\otimes H\right)U$. Hence, the complexity
  of synthesizing $U$ over $\gensup_n$ is $O(2^nk)$.
\end{proof}

In the context of quantum computation, the dimension of the matrix to
be synthesized is exponential in the number of qubits. That is,
$n=2^m$, where $m$ is the number of qubits. Moreover, the cost of
synthesizing an $m$-qubit circuit for any element of $\gensup_{2^m}$
is linear in $m$. Therefore, the gate complexity of an $m$-qubit
Toffoli-Hadamard circuit synthesized using the local synthesis
algorithms is $O(2^{2^m}mk)$.

\section{The Householder Synthesis Algorithm}
\label{sec:householder_algo}

In this section, we explore how using additional dimensions can be
helpful in quantum circuit synthesis. These results are a direct
adaptation to the Toffoli-Hadamard gate set of the methods introduced
in \cite{kliuchnikov2013synthesis} for the Clifford+$T$ gate
set. Compared to the local synthesis algorithm, the algorithm
presented in this section, which we call the \emph{Householder
synthesis algorithm}, reduces the gate complexity of the produced
circuits from $O(2^n\log(n)k)$ to $O(n^2\log{n}k)$, where $n$ is the
dimension of the input matrix.

\begin{definition}
  \label{def:reflection_op}
  Let $\ket{\psi}$ be an $n$-dimensional unit vector. The
  \emph{reflection operator $R_{\ket{\psi}}$ around $\ket{\psi}$} is
  defined as
  \[
  R_{\ket{\psi}} = I - 2\ket{\psi}\bra{\psi}.
  \]
\end{definition}

Note that if $R$ is a reflection operator about some unit vector, then
$R$ is unitary. Indeed, $R=R^\dagger$ and $R^2 = I$. As a result, if
$\ket{\psi}$ is a unit vector of the form $\ket{v}/\sqrt{2}{}^k$ for
some integer vector $\ket{v}$, then $R_{\ket{\psi}} \in \supin_n$.

We start by showing that if $U\in\supin_n$, then there is an operator
$U'$ constructed from $U$ that can be conveniently factored as a
product of reflections. In what follows, we will use two single-qubit
states:
\[
\ket{+} = \frac{\ket{0}+\ket{1}}{\sqrt{2}}
\qquad \mbox{and} \qquad
\ket{-} = \frac{\ket{0}-\ket{1}}{\sqrt{2}}. 
\]

\begin{proposition}
  \label{prop:decomp}
  Let $U \in \supin_n$ and define
  \[
  U' = \ket{+}\bra{-}\otimes U + \ket{-}\bra{+}\otimes U^\dagger.
  \]
  Then $U' \in \supin_{2n}$ and $U'$ can be factored into $n$
  reflections in $\supin_{2n}$. That is, $U'=R_{\ket{\phi_0}} \cdots
  R_{\ket{\phi_{n-1}}}$, where $R_{\ket{\phi_0}},\ldots,
  R_{\ket{\phi_{n-1}}} \in \supin_{2n}$.
\end{proposition}

\begin{proof}
 Let $U$ and $U'$ be as stated. It can be verified by direct
 computation that $U'$ is unitary. Moreover, since $U \in \supin_n$
 and $\ket{+}\bra{-}$ and $\ket{-}\bra{+}$ are integral matrices
 scaled by $1/2$, it follows that $U' \in \supin_{2n}$. It remains to
 show that $U'$ is a product of reflection operators. Define
 \[
 \ket{\omega_j^{\pm}}=\frac{\ket{-}\ket{j}\pm
   \ket{+}\ket{u_j}}{\sqrt{2}},
 \]
 where $\ket{u_j}$ is the $j$-th column vector in $U$ and $\ket{j}$
 denotes the $j$-th computational basis vector. Since
 $\braket{\omega_j^{+}|\omega_j^{+}} =
 \braket{\omega_j^{-}|\omega_j^{-}} = 1$, both $\ket{\omega_j^{+}}$
 and $\ket{\omega_j^{-}}$ are unit vectors. Moreover, it is easy to
 show that any two distinct $\ket{\omega_j^{\pm}}$ are orthogonal, so
 that $\left\{\ket{\omega_j^{\pm}}\mid j = 0,\ldots,n-1\right\}$ forms
 an orthonormal basis. Now let $P_j^{+} =
 \ket{\omega_j^{+}}\bra{\omega_j^{+}}$ and $P_j^{-} =
 \ket{\omega_j^{-}}\bra{\omega_j^{-}}$. It follows from the
 completeness equation that
 \begin{equation}
   I =\ \sum_{j=0}^{n-1}\bigl(\ket{\omega_j^{+}}\bra{\omega_j^{+}}+
   \ket{\omega_j^{-}}\bra{\omega_j^{-}}\bigr) =
   \sum_{j=0}^{n-1}\bigl(P_j^{+} + P_j^{-}\bigr).\label{equ:1}
 \end{equation}
 Furthermore, $\ket{\omega_j^{+}}$ and $\ket{\omega_j^{-}}$ are the
 $+1$ and $-1$ eigenstates of $U'$, respectively. Now note that $U'$
 is Hermitian and thus normal. Hence, by the spectral theorem, we have
 \begin{equation}
   U' = \sum_{j=0}^{n-1}\bigl(\ket{\omega_j^{+}}\bra{\omega_j^{+}}-
   \ket{\omega_j^{-}}\bra{\omega_j^{-}}\bigr) =
   \sum_{j=0}^{n-1}\bigl(P_j^{+}-P_j^{-}\bigr).\label{equ:2}
 \end{equation}
 From \Cref{equ:1,equ:2}, $I-U' = 2\sum_{j=0}^{n-1}P_j^{-}$, which
 implies that
 \begin{equation}
   U' =\ I - 2\sum_{j=0}^{n-1}P_j^{-} =
   I - 2\sum_{j=0}^{n-1}\ket{\omega_j^{-}}\bra{\omega_j^{-}}
   =\prod_{j=0}^{n-1}\left(I-2\ket{\omega_j^{-}}\bra{\omega_j^{-}}\right)
   =\prod_{j=0}^{n-1}R_{\ket{\omega_j^{-}}}.\label{equ:3}
 \end{equation}
 Since $\ket{\omega_j^{-}}$ is a unit vector of the form
 $\ket{v_j}/\sqrt{2}{}^k$ where $\ket{v_j}$ is an integer vector,
 $R_{\ket{\omega_j^{-}}}\in\supin_{2n}$. This completes the proof.
\end{proof}

By noting that $\ket{+}\bra{-}$ and $\ket{-}\bra{+}$ are matrices with
dyadic entries, one can reason as in the proof of \cref{prop:decomp}
to show that an analogous result holds for $U\in\integral_n$, rather
than $U\in\supin_n$.

\begin{proposition}
  \label{cor:decomp}
  Let $U \in \integral_n$ and define
  \[
  U' = \ket{+}\bra{-}\otimes U + \ket{-}\bra{+}\otimes U^\dagger.
  \]
  Then $U' \in \integral_{2n}$ and $U'$ can be factored into $n$
  reflections in $\integral_{2n}$. That is, $U'=R_{\ket{\phi_0}}
  \cdots R_{\ket{\phi_{n-1}}}$, where $R_{\ket{\phi_0}},\ldots,
  R_{\ket{\phi_{n-1}}} \in \integral_{2n}$.
\end{proposition}

\begin{proposition}
  \label{prop:gatecount}
  Let $\ket{\psi}=\ket{v}/\sqrt{2}^k$ be an $n$-dimensional unit
  vector, where $\ket{v}$ is an integer vector. Assume that
  $\lde_{\sqrt{2}}(\ket{\psi})=k$. Then the reflection operator
  $R_{\ket{\psi}}$ can be exactly represented by $O(nk)$ generators
  over $\gensup_{n}$.
\end{proposition}

\begin{proof}
  Let $\ket{\psi}$ be as stated. When $k$ is even, then, by
  \cref{lem:column}, there exists a word $G$ over $\gens_n$ such that
  \begin{equation}
      G\ket{\psi} = \ket{0}.\label{equ:4}
  \end{equation}
  Since the elements of $\gens_n$ are self-inverse, the word
  $G^\dagger$ obtained by reversing $G$ is a word over $\gens_n$ such
  that $G^\dagger\ket{0} = \ket{\psi}$. Moreover, we have
  $G^\dagger R_{\ket{0}}G = R_{\ket{\psi}}$, since
  \begin{equation}
    G^\dagger R_{\ket{0}}G = G^\dagger \left(I-2\ket{0}\bra{0}\right)G
    = I - 2\left(G^\dagger
    \ket{0}\right)\left(G^\dagger\ket{0}\right)^\dagger =
    R_{G^\dagger\ket{0}} = R_{\ket{\psi}}.
  \end{equation}
  Hence the number of elements of $\gens_n$ that are needed to
  represent $R_{\ket{\psi}}$ is equal to the number of generators
  needed to represent $G$, $G^\dagger$, and $R_{\ket{0}}$. Note that
  \[
  R_{\ket{0}} = I - 2\ket{0}\bra{0} = \mone{0} \in \gens_n.
  \]
  Moreover, the number of generators needed to represent $G^\dagger$
  is equal to the number of generators needed to represent $G$. By
  \cref{lem:helper2}, $O(nk)$ generators are needed for this. Hence,
  $R_{\ket{\psi}}$ can be exactly represented by $O(nk)$ generators
  over $\gens_n \subset \gensup_n$.  When $k$ is odd, we can reason as
  in \cref{cor:membership} to show that $R_{\ket{\psi}}$ can be
  represented as a product of $O(nk)$ generators from $\gensup_n$.
\end{proof}

\begin{proposition}
  \label{prop:impgatecount}
  Let $U \in \supin_n$ and $U'\in\supin_{2n}$ be as in
  \cref{prop:decomp} and assume that $\lde_{\sqrt{2}}(U)=k$. Then $U'$
  can be represented by $O(n^2k)$ generators from $\gensup_n$.
  \end{proposition}

  \begin{proof}
  By \cref{prop:decomp}, $U'$ can be expressed as a product $n$
  reflections. By \cref{prop:gatecount}, each one of these reflections
  can be exactly represented by $O(nk)$ generators from
  $\gensup_n$. Therefore, to express $U'$, we need $n \cdot O(nk) =
  O(n^2k)$ generators from $\gensup_n$.
\end{proof}

\begin{corollary}
  \label{cor:impgatecount}
  Let $U \in \integral_n$ and $U'\in\integral_{2n}$ be as in
  \cref{cor:decomp} and assume that $\lde(U)=k$. Then $U'$
  can be represented by $O(n^2k)$ generators from $\gens_n$.
\end{corollary}

To conclude this section, we use \cref{prop:decomp} to define the
Householder synthesis algorithm, which produces circuits of size
$O(4^mmk)$. Suppose that $n=2^m$, where $m$ is the number of qubits on
which a given operator $U\in\supin_{2^m}$ acts. Suppose moreover that
$\lde_{\sqrt{2}}(U)=k$. The operator $U'$ of \cref{prop:impgatecount}
can be represented as a product of $O(n^2k)=O(4^mk)$ elements of
$\gensup_{2^{m+1}}$. Since any element of $\gensup_{2^{m+1}}$ can be
represented by a Toffoli-Hadamard circuit of gate count $O(m)$, we get
a circuit $D$ of size $O(4^mmk)$ for $U'$. Now consider the circuit $C
= (H\otimes I)D (HX\otimes I)$. For any state $\ket{\phi}$, we have
\[
C\ket{0}\ket{\phi} = (H\otimes I)D (HX\otimes I)\ket{0}\ket{\phi} =
(H\otimes I)D \ket{-}\ket{\phi} = (H\otimes I)\ket{+}U\ket{\phi} =
\ket{0}U\ket{\phi}.
\]
Hence, $C$ is a Toffoli-Hadamard circuit for $U$ (which uses an
additional ancillary qubit).

The Householder exact synthesis algorithm can be straightforwardly
defined in the case of circuits over the gate set $\s{X, CX, CCX, K}$,
with the small caveat that two additional ancillary qubits are
required, since one cannot prepare a single qubit in the state
$\ket{-}$ over $\s{X, CX, CCX, K}$.

\section{The Global Synthesis Algorithm}
\label{sec:globalsyn_algo}

The local synthesis algorithm factorizes a matrix by reducing one
column at a time. As we saw in \cref{sec:AGR_algo}, this approach can
lead to large circuits, since reducing the least (scaled) denominator
exponent of one column may increase that of the subsequent columns. We
now take a global view of the matrix, focusing on matrices of
dimension 2, 4, and 8 (i.e., matrices on 1, 2, and 3 qubits). Through
a careful study of the structure of these matrices, we define a
synthesis algorithm that reduces the least (scaled) denominator
exponent of the entire matrix at every iteration. We refer to this
alternative synthesis algorithm as the \emph{global synthesis
algorithm}.

\subsection{Binary Patterns}
\label{ssec:binpat}

We associate a binary matrix (i.e., a matrix over $\Z_2$) to every
element of $\supin_n$. These binary matrices, which we call
\emph{binary patterns}, will be useful in designing a global synthesis
algorithm.

\begin{definition}
  \label{def:residuematrix}
  Let $U \in \supin_n$ and write $U$ as $U = M/\sqrt{2}{}^k$ with
  $\lde_{\sqrt{2}}(U) = k$. The \emph{binary pattern of $U$} is the
  binary matrix $\overline{U}$ defined by $\overline{U}_{i,j}=M_{i,j}
  \pmod{2}$.
\end{definition}

The matrix $\overline{U}$ is the binary matrix obtained by taking the
residue modulo 2 of every entry of the integral part of $U$ (when $U$
is written using its least scaled denominator exponent). The next two
lemmas establish important properties of binary patterns.

\begin{lemma}
  \label{lem:weight}
  Let $U \in \supin_n$ with $\lde_{\sqrt{2}}(U) = k$. If $k >1$, then
  the number of 1's in any column of $\overline{U}$ is doubly-even.
\end{lemma}

\begin{proof}
  Consider an arbitrary column $\ket{u} = \ket{v}/\sqrt{2}{}^k$ of
  $U$. Let $\ket{\overline{u}}$ be the corresponding column in
  $\overline{U}$. Since $\braket{u|u} = 1$, we have $\sum v^2_i =
  2^k$. Thus, when $k >1$, we have $\sum v^2_i\equiv 0\pmod{4}$. Since
  $v_i^2\equiv 1\pmod{4}$ if and only if $v_i\equiv 1 \pmod{2}$, and
  since the only squares modulo 4 are 0 and 1, the number of odd $v_i$
  must be a multiple of $4$. Hence, the number of 1's in any column of
  $\overline{U}$ is doubly-even.
\end{proof}

\begin{lemma}
  \label{lem:collision}
  Let $U \in \supin_n$ with $\lde_{\sqrt{2}}(U) = k$. If $k > 0$, then
  any two distinct columns of $\overline{U}$ have evenly many $1$'s in
  common.
\end{lemma}

\begin{proof}
  Consider two distinct columns $\ket{u}$ and $\ket{w}$ of $U$. Let
  $\ket{\overline{u}}$ and $\ket{\overline{w}}$ be the corresponding
  columns in $\overline{U}$. Since $U$ is orthogonal, we have
  \begin{equation}
  \label{eq:4}
  \braket{u|w} = \sum_{i=0}^{n-1} u_i w_i = 0. 
  \end{equation}
  Taking \cref{eq:4} modulo $2$ implies that $\lvert\{i ~;~
  \overline{u}_i = \overline{w}_i = 1\}\rvert\equiv 0 \pmod{2}$, as
  desired.
\end{proof}

\cref{lem:weight,lem:collision} also hold for the rows of
$\overline{U}$. The proofs are similar, so they are omitted
here. These lemmas show that the binary matrices that are the binary
pattern of an element of $\supin_n$ form a strict subset of
$\Z_2^{n\times n}$. The proposition below gives a characterization of
this subset for $n=8$. The proof of the proposition is a long case
distinction which can be found in \cref{chp:binaryPatterns}.

\begin{figure}[t]
{\scriptsize
\[
A=\begin{bmatrix}
    1 & 1 & 1 & 1 & 1 & 1 & 1 & 1\\
    1 & 1 & 1 & 1 & 1 & 1 & 1 & 1\\
    1 & 1 & 1 & 1 & 1 & 1 & 1 & 1\\
    1 & 1 & 1 & 1 & 1 & 1 & 1 & 1\\
    1 & 1 & 1 & 1 & 1 & 1 & 1 & 1\\
    1 & 1 & 1 & 1 & 1 & 1 & 1 & 1\\
    1 & 1 & 1 & 1 & 1 & 1 & 1 & 1\\
    1 & 1 & 1 & 1 & 1 & 1 & 1 & 1
  \end{bmatrix}, \quad B = \begin{bmatrix}
    1 & 1 & 1 & 1 & 1 & 1 & 1 & 1\\
    1 & 1 & 1 & 1 & 1 & 1 & 1 & 1\\
    1 & 1 & 1 & 1 & 1 & 1 & 1 & 1\\
    1 & 1 & 1 & 1 & 1 & 1 & 1 & 1\\
    1 & 1 & 1 & 1 & 0 & 0 & 0 & 0\\
    1 & 1 & 1 & 1 & 0 & 0 & 0 & 0\\
    1 & 1 & 1 & 1 & 0 & 0 & 0 & 0\\
    1 & 1 & 1 & 1 & 0 & 0 & 0 & 0
  \end{bmatrix}, \quad C = \begin{bmatrix}
    1 & 1 & 1 & 1 & 1 & 1 & 1 & 1\\
    1 & 1 & 1 & 1 & 1 & 1 & 1 & 1\\
    1 & 1 & 1 & 1 & 0 & 0 & 0 & 0\\
    1 & 1 & 1 & 1 & 0 & 0 & 0 & 0\\
    1 & 1 & 0 & 0 & 1 & 1 & 0 & 0\\
    1 & 1 & 0 & 0 & 1 & 1 & 0 & 0\\
    1 & 1 & 0 & 0 & 0 & 0 & 1 & 1\\
    1 & 1 & 0 & 0 & 0 & 0 & 1 & 1
  \end{bmatrix}, \quad
D=\begin{bmatrix}
    1 & 1 & 1 & 1 & 0 & 0 & 0 & 0\\
    1 & 1 & 1 & 1 & 0 & 0 & 0 & 0\\
    1 & 1 & 1 & 1 & 0 & 0 & 0 & 0\\
    1 & 1 & 1 & 1 & 0 & 0 & 0 & 0\\
    1 & 1 & 0 & 0 & 1 & 1 & 0 & 0\\
    1 & 1 & 0 & 0 & 1 & 1 & 0 & 0\\
    1 & 1 & 0 & 0 & 1 & 1 & 0 & 0\\
    1 & 1 & 0 & 0 & 1 & 1 & 0 & 0
  \end{bmatrix}, 
   \]
 \[
  E = \begin{bmatrix}
    1 & 1 & 1 & 1 & 1 & 1 & 1 & 1\\
    1 & 1 & 1 & 1 & 1 & 1 & 1 & 1\\
    1 & 1 & 1 & 1 & 0 & 0 & 0 & 0\\
    1 & 1 & 1 & 1 & 0 & 0 & 0 & 0\\
    0 & 0 & 0 & 0 & 1 & 1 & 1 & 1\\
    0 & 0 & 0 & 0 & 1 & 1 & 1 & 1\\
    0 & 0 & 0 & 0 & 0 & 0 & 0 & 0\\
    0 & 0 & 0 & 0 & 0 & 0 & 0 & 0
  \end{bmatrix}, \quad F = \begin{bmatrix}
    1 & 1 & 1 & 1 & 0 & 0 & 0 & 0\\
    1 & 1 & 1 & 1 & 0 & 0 & 0 & 0\\
    1 & 1 & 0 & 0 & 1 & 1 & 0 & 0\\
    1 & 1 & 0 & 0 & 1 & 1 & 0 & 0\\
    1 & 0 & 1 & 0 & 1 & 0 & 1 & 0\\
    1 & 0 & 1 & 0 & 1 & 0 & 1 & 0\\
    1 & 0 & 0 & 1 & 0 & 1 & 1 & 0\\
    1 & 0 & 0 & 1 & 0 & 1 & 1 & 0
  \end{bmatrix},\quad
G=\begin{bmatrix}
    1 & 1 & 1 & 1 & 0 & 0 & 0 & 0\\
    1 & 1 & 1 & 1 & 0 & 0 & 0 & 0\\
    1 & 1 & 0 & 0 & 1 & 1 & 0 & 0\\
    1 & 1 & 0 & 0 & 1 & 1 & 0 & 0\\
    0 & 0 & 1 & 1 & 1 & 1 & 0 & 0\\
    0 & 0 & 1 & 1 & 1 & 1 & 0 & 0\\
    0 & 0 & 0 & 0 & 0 & 0 & 0 & 0\\
    0 & 0 & 0 & 0 & 0 & 0 & 0 & 0
  \end{bmatrix}, \quad H = \begin{bmatrix}
    1 & 1 & 1 & 1 & 0 & 0 & 0 & 0\\
    1 & 1 & 1 & 1 & 0 & 0 & 0 & 0\\
    1 & 1 & 0 & 0 & 1 & 1 & 0 & 0\\
    1 & 1 & 0 & 0 & 1 & 1 & 0 & 0\\
    0 & 0 & 1 & 1 & 0 & 0 & 1 & 1\\
    0 & 0 & 1 & 1 & 0 & 0 & 1 & 1\\
    0 & 0 & 0 & 0 & 1 & 1 & 1 & 1\\
    0 & 0 & 0 & 0 & 1 & 1 & 1 & 1
  \end{bmatrix}, 
   \]
 \[
 \ I = \begin{bmatrix}
    1 & 1 & 1 & 1 & 0 & 0 & 0 & 0\\
    1 & 1 & 1 & 1 & 0 & 0 & 0 & 0\\
    1 & 1 & 1 & 1 & 0 & 0 & 0 & 0\\
    1 & 1 & 1 & 1 & 0 & 0 & 0 & 0\\
    0 & 0 & 0 & 0 & 1 & 1 & 1 & 1\\
    0 & 0 & 0 & 0 & 1 & 1 & 1 & 1\\
    0 & 0 & 0 & 0 & 1 & 1 & 1 & 1\\
    0 & 0 & 0 & 0 & 1 & 1 & 1 & 1
  \end{bmatrix},\quad
  J = \begin{bmatrix}
    1 & 1 & 1 & 1 & 0 & 0 & 0 & 0\\
    1 & 1 & 1 & 1 & 0 & 0 & 0 & 0\\
    1 & 1 & 1 & 1 & 0 & 0 & 0 & 0\\
    1 & 1 & 1 & 1 & 0 & 0 & 0 & 0\\
    0 & 0 & 0 & 0 & 0 & 0 & 0 & 0\\
    0 & 0 & 0 & 0 & 0 & 0 & 0 & 0\\
    0 & 0 & 0 & 0 & 0 & 0 & 0 & 0\\
    0 & 0 & 0 & 0 & 0 & 0 & 0 & 0
  \end{bmatrix}, \quad K = \begin{bmatrix}
    1 & 1 & 1 & 1 & 1 & 1 & 1 & 1\\
    1 & 1 & 1 & 1 & 1 & 1 & 1 & 1\\
    1 & 1 & 1 & 1 & 1 & 1 & 1 & 1\\
    1 & 1 & 1 & 1 & 1 & 1 & 1 & 1\\
    0 & 0 & 0 & 0 & 0 & 0 & 0 & 0\\
    0 & 0 & 0 & 0 & 0 & 0 & 0 & 0\\
    0 & 0 & 0 & 0 & 0 & 0 & 0 & 0\\
    0 & 0 & 0 & 0 & 0 & 0 & 0 & 0
  \end{bmatrix}, \quad \ L=\begin{bmatrix}
    1 & 1 & 1 & 1 & 1 & 1 & 1 & 1\\
    1 & 1 & 1 & 1 & 0 & 0 & 0 & 0\\
    1 & 1 & 0 & 0 & 1 & 1 & 0 & 0\\
    1 & 1 & 0 & 0 & 0 & 0 & 1 & 1\\
    1 & 0 & 1 & 0 & 1 & 0 & 1 & 0\\
    1 & 0 & 1 & 0 & 0 & 1 & 0 & 1\\
    1 & 0 & 0 & 1 & 1 & 0 & 0 & 1\\
    1 & 0 & 0 & 1 & 0 & 1 & 1 & 0
  \end{bmatrix},
   \]
 \[
\!\!\!\! M=\begin{bmatrix}
    1 & 1 & 1 & 1 & 0 & 0 & 0 & 0\\
    1 & 1 & 0 & 0 & 1 & 1 & 0 & 0\\
    1 & 0 & 1 & 0 & 1 & 0 & 1 & 0\\
    1 & 0 & 0 & 1 & 0 & 1 & 1 & 0\\
    0 & 1 & 1 & 0 & 1 & 0 & 0 & 1\\
    0 & 1 & 0 & 1 & 0 & 1 & 0 & 1\\
    0 & 0 & 1 & 1 & 0 & 0 & 1 & 1\\
    0 & 0 & 0 & 0 & 1 & 1 & 1 & 1\\
  \end{bmatrix}, \quad N=\begin{bmatrix}
    1 & 1 & 1 & 1 & 0 & 0 & 0 & 0\\
    1 & 1 & 0 & 0 & 1 & 1 & 0 & 0\\
    1 & 0 & 1 & 0 & 1 & 0 & 1 & 0\\
    1 & 0 & 0 & 1 & 0 & 1 & 1 & 0\\
    0 & 1 & 1 & 0 & 0 & 1 & 1 & 0\\
    0 & 1 & 0 & 1 & 1 & 0 & 1 & 0\\
    0 & 0 & 1 & 1 & 1 & 1 & 0 & 0\\
    0 & 0 & 0 & 0 & 0 & 0 & 0 & 0\\
  \end{bmatrix}.
\]  
}
\caption{Binary patterns for the elements of $\supin_8$.}
\label{fig:patterns}
\end{figure}

\begin{restatable}{proposition}{bt}
  \label{thm:binary_patterns}
  Let $U\in\supin_{8}$ with $\lde_{\sqrt{2}}(U)\geq2$. Then up to row
  permutation, column permutation, and taking the transpose,
  $\overline{U}$ is one of the 14 binary patterns in
  \cref{fig:patterns}.
\end{restatable}
  
\begin{definition}
  Let $n$ be even and let $B \in \Z_2^{n \times n}$. We say that $B$
  is \emph{row-paired} if the rows of $B$ can be partitioned into
  identical pairs. Similarly, we say that $B$ is \emph{column-paired}
  if the columns of $B$ can be partitioned into identical pairs.
\end{definition}

Note that, for $U \in \supin_n$, if $\overline{U}$ is row-paired, then
$\overline{U^\intercal}$ is column-paired. Indeed, if $\overline{U}$
is row-paired, then $\overline{U}^\intercal$ is column-paired so that
$\overline{U^\intercal}=\overline{U}^\intercal$ is column-paired.

Row-paired binary patterns will play an important role in the global
synthesis algorithm. Intuitively, if $\overline{U}$ is row-paired,
then one can permute the rows of $U$ to place identical rows next to
one another, at which point a single Hadamard gate can be used to
globally reduce the least scaled denominator exponent of $U$. This
intuition is detailed in \Cref{lem:rowpaired_red}, where $S_n$ denotes
the symmetric group on $n$ letters.

\begin{lemma}
  \label{lem:rowpaired_red}
  Let $n$ be even and let $U \in \supin_n$. If $\overline{U}$ is
  row-paired, then there exists $P \in S_n$ such that
  \[
  \lde_{\sqrt{2}}((I\otimes H) P U) < \lde_{\sqrt{2}}(U).
  \]
\end{lemma}

\begin{proof}
  Let $U=M/\sqrt{2}^k$ and let $r_0, \ldots, r_{n-1}$ be the rows of
  $M$. Because $\overline{U}$ is row-paired, there exists some $P \in
  S_n$ such that
  \[
   PU = \frac{1}{\sqrt{2}^k}\begin{bmatrix}
    r_0\\
    \vdots\\
    r_{n-1}
  \end{bmatrix}, 
  \]
  with $r_0 \equiv r_1$, $r_2 \equiv r_3$, \ldots, and $ r_{n-2}
  \equiv r_{n-1}$ modulo 2. Since $I\otimes H$ is the block diagonal
  matrix $I\otimes H=\diag(H, H, \ldots, H)$, left-multiplying $PU$ by
  $I\otimes H$ yields
  \[
  \left(I \otimes H\right) P U = 
  \begin{bmatrix}
    r_0\\
    \vdots\\
    r_{n-1}
  \end{bmatrix}= \frac{1}{\sqrt{2}^{k+1}}\begin{bmatrix}
    r_0+r_1\\
    r_0-r_1\\
    \vdots\\
    r_{n-2}+r_{n-1}\\
    r_{n-2}-r_{n-1}
  \end{bmatrix} = \frac{2}{\sqrt{2}^{k+1}}\begin{bmatrix}
    r'_0\\
    \vdots\\
    r'_{n-1}
  \end{bmatrix}
  = \frac{1}{\sqrt{2}^{k-1}}\begin{bmatrix}
    r'_0\\
    \vdots\\
    r'_{n-1}
  \end{bmatrix},
  \]
  for some integer row vectors $r'_0,\ldots,r'_{n-1}$. Thus,
  $\lde_{\sqrt{2}}((I\otimes H)P U) < \lde_{\sqrt{2}}(U)$ as desired.
\end{proof}

\begin{lemma}
  \label{lem:columnpaired_red}
  Let $n$ be even and let $U \in \supin_n$ . If $\overline{U}$ is
  column-paired, then there exists $P \in S_n$ such that
  \[
  \lde_{\sqrt{2}}(UP (I\otimes H)) < \lde_{\sqrt{2}}(U).
  \]
\end{lemma}

\begin{proof}
  Since $\overline{U}$ is column-paired, $\overline{U^\intercal}$ is
  row-paired. By \cref{lem:rowpaired_red}, there exists $Q \in S_n$
  such that $\lde_{\sqrt{2}}((I\otimes H)Q U^\intercal) <
  \lde_{\sqrt{2}}(U^\intercal)$. Hence, letting $P=Q^\intercal$, and
  using the fact that the least scaled denominator exponent of an
  element of $\supin_n$ is the same as that of its transpose, we get
  \[
  \lde_{\sqrt 2}(UP(I\otimes H)) = \lde_{\sqrt 2}((UP(I\otimes
  H))^\intercal)=\lde_{\sqrt 2} ((I\otimes H)Q U^\intercal)<
  \lde_{\sqrt{2}}(U^\intercal) = \lde_{\sqrt{2}}(U). \qedhere
  \]
\end{proof}

\begin{lemma}
  \label{lem:unice}
  Let $U \in \supin_8$ with $\lde_{\sqrt{2}}(U) = k$. If
  $\overline{U}$ is neither row-paired nor column-paired, then, up to
  row permutation, column permutation, and taking the transpose,
  $\overline{\left(I \otimes H\right) U \left(I \otimes H\right)}$ is
  row-paired and $\lde_{\sqrt{2}}(\left(I \otimes H\right) U \left(I
  \otimes H\right)) \leq \lde_{\sqrt{2}}(U)$.
\end{lemma}

\begin{proof}
  Let $U$ be as stated. By \cref{thm:binary_patterns}, up to row
  permutation, column permutation, and taking the transpose,
  $\overline{U}$ is one of the binary patterns in
  \cref{fig:patterns}. Since $\overline{U}$ is neither row-paired nor
  column-paired, $\overline{U}$ is $L$, $M$, or $N$. Write
  $\overline{U}$ as the $4\times 4$ block matrix
  \[
  \overline{U} = \begin{bmatrix}
    \ P_{0,0} \
    & \ P_{0,1} \ & \ P_{0,2} \ & \ P_{0,3}\ \\
   P_{1,0}
    & P_{1,1} & P_{1,2} & P_{1,3}\\
   P_{2,0}
    & P_{2,1} & P_{2,2} & P_{2,3}\\
   P_{3,0}
    & P_{3,1} & P_{3,2} & P_{3,3}\\
  \end{bmatrix},
  \]
  where $P_{i,j}$ is a $2\times 2$ binary matrix. By inspection of
  \cref{fig:patterns}, since $\overline{U}$ is one of $L$, $M$, or
  $N$, we see that each $P_{i,j}$ is one of the binary matrices below:
  \[
  \begin{bmatrix}
    \ 1 \ & \ 1 \ \\
    \ 1 \ & \ 1 \
  \end{bmatrix}, \quad \begin{bmatrix}
   \ 0 \ & \ 0 \ \\
   \ 0 \ & \ 0 \
  \end{bmatrix}, \quad  \begin{bmatrix}
   \ 1 \ & \ 1 \ \\
  \ 0 \ & \ 0 \
  \end{bmatrix}, \quad  \begin{bmatrix}
   \ 0 \ & \ 0 \ \\
   \ 1 \ & \ 1 \
  \end{bmatrix}, \quad \begin{bmatrix}
   \ 1 \ & \ 0 \ \\
   \ 1 \ & \ 0 \
  \end{bmatrix}, \quad  \begin{bmatrix}
   \ 0 \ & \ 1 \ \\
   \ 0 \ & \ 1 \
  \end{bmatrix}, \quad  \begin{bmatrix}
   \ 1 \ & \ 0 \ \\
   \ 0 \ & \ 1 \
    \end{bmatrix}, \quad \mbox{and} \quad  \begin{bmatrix}
    \ 0 \ & \ 1 \ \\
    \ 1 \ & \ 0 \
  \end{bmatrix}.
  \]
  In particular, each $P_{i,j}$ has evenly many nonzero entries. Now
  write $U$ as the $4\times 4$ block matrix
  \[
  U=
  \frac{1}{\sqrt{2}{}^{k}} \begin{bmatrix}
    \ Q_{0,0} \
    & \ Q_{0,1} \ & \ Q_{0,2} \ & \ Q_{0,3}\ \\
   Q_{1,0}
    & Q_{1,1} & Q_{1,2} & Q_{1,3}\\
   Q_{2,0}
    & Q_{2,1} & Q_{2,2} & Q_{2,3}\\
   Q_{3,0}
    & Q_{3,1} & Q_{3,2} & Q_{3,3}\\
  \end{bmatrix},
  \]
  where $Q_{i,j}$ is a $2\times 2$ integer matrix such that
  $Q_{i,j}=P_{i,j}$ modulo 2. As $I\otimes H=\diag(H, H, H, H)$, we
  have 
  \[
  (I\otimes H)U(I\otimes H)=
  \frac{1}{\sqrt{2}{}^k} \begin{bmatrix}
    \ Q'_{0,0} \
    & \ Q'_{0,1} \ & \ Q'_{0,2} \ & \ Q'_{0,3}\ \\
   Q'_{1,0}
    & Q'_{1,1} & Q'_{1,2} & Q'_{1,3}\\
   Q'_{2,0}
    & Q'_{2,1} & Q'_{2,2} & Q'_{2,3}\\
   Q'_{3,0}
    & Q'_{3,1} & Q'_{3,2} & Q'_{3,3}\\
  \end{bmatrix},
  \]
  where $Q'_{i,j} = HQ_{i,j}H$. Since $Q_{i,j}$ is an integer matrix
  with evenly many odd entries and, since for any integers $w$, $x$,
  $y$, and $z$, we have
  \[
  H \begin{bmatrix} \ w \ & \ x \ \\ \ y \ & \ z \ \end{bmatrix}H =
    \frac{1}{2} \begin{bmatrix} \ w+x+y+z \ & \ w-x+y-z \ \\ \ w+x-y-z
      \ & \ w-x-y+z \ \end{bmatrix},
  \]
  it follows that $Q'_{i,j}=HQ_{i,j}H$ is itself an integer
  matrix. Thus, $\lde_{\sqrt{2}}(\left(I \otimes H\right) U \left(I
  \otimes H\right)) \leq \lde_{\sqrt{2}}(U)$. A long but
  straightforward calculation shows that $\overline{(I\otimes
    H)U(I\otimes H)}$ is in fact row-paired.
\end{proof}

\subsection{The 1- and 2-Qubit Cases}
\label{ssec:onetwoqs}

We now discuss the exact synthesis problem for $\supin_2$ and
$\supin_4$. The problem is simple in these cases because the groups
are finite. Despite their simplicity, these instances of the problem
shed some light on our method for defining a global synthesis
algorithm for $\supin_8$.

\begin{proposition}
  \label{lem:L2lde}
  If $U \in \supin_2$, then $\lde_{\sqrt{2}}(U) \leq 1$.
\end{proposition}

\begin{proof}
  Let $k=\lde_{\sqrt{2}}(U)$ and suppose that $k\geq 2$. Let $\ket{u}$
  be the first column of $U$ with $\lde(\ket{u}) = k$ and let
  $\ket{v}=2^k \ket{u}$. As $\braket{u|u} = 1$, we have $ v_0^2 +
  v_1^2 = 2^{k} \equiv 0 \pmod{4}$, since $k\geq 2$. Therefore, $v_0
  \equiv v_1 \equiv 0 \pmod{2}$. This is a contradiction since at
  least one of $v_0$ and $v_1$ must be odd for $k$ to be minimal.
\end{proof}

\begin{lemma}
  \label{lem:O4lde_helper}
  Let $a \in \Z$. Then $a^2 \equiv 1\pmod{8}$ if and only if $a \equiv
  1\pmod{2}$.
\end{lemma}

\begin{proof}
  If $a\equiv 0 \pmod{2}$, then $a^2$ is even, so $a^2\not\equiv 1
  \pmod{8}$. If $a \equiv 1\pmod{2}$, then $a = 2q + 1$ for some $q
  \in \Z$, so that $a^2 = 4q^2 + 4q + 1$. If $q = 2p$ for some $p \in
  \Z$, then $a^2 = 1 + 8(2p^2 + p) \equiv 1\pmod{8}$. Otherwise, $q =
  2p + 1$ for some $p \in \Z$ and $a^2 = 1 + 8(2p^2+3p+1) \equiv
  1\pmod{8}$.
\end{proof}

\begin{proposition}
  \label{lem:L4lde}
  If $U \in \supin_4$ then $\lde_{\sqrt{2}}(U) \leq 2$.
\end{proposition}

\begin{proof}
  Let $k=\lde_{\sqrt{2}}(U)$ and suppose that $k\geq 3$. Let $\ket{u}$
  be the first column of $U$ with $\lde_{\sqrt 2}(\ket{u}) = k$ and
  let $\ket{v}=\sqrt{2}^k \ket{u}$. By reasoning as in \cref{lem:ind},
  we see that the number of odd entries in $\ket{v}$ must be
  doubly-even. Hence, $ v_0 \equiv v_1 \equiv v_2 \equiv v_3 \equiv
  1\pmod{2}$. By \cref{lem:O4lde_helper}, $v^2_0 \equiv v^2_1 \equiv
  v^2_2 \equiv v^2_3 \equiv 1\pmod{8}$. As $\braket{u|u} = 1$, we have
  $v^2_0 + v^2_1 + v^2_2 + v^2_3 = 4^k \equiv 0\pmod{8}$. This is a
  contradiction since we in fact have $v^2_0 + v^2_1 + v^2_2 + v^2_3
  \equiv 4 \pmod{8}$.
\end{proof}

It follows from \cref{lem:L4lde} that $\supin_4$ is finite. Indeed, by
\cref{lem:L4lde}, the least scaled denominator exponent of an element
of $\supin_4$ can be no more than 2. As a consequence, the number of
possible columns for a matrix in $\supin_4$ is upper bounded by the
number of integer solutions to the equation $v_0^2+v_1^2+v_2^2+v_3^2 =
2^k$, which is finite since $k\leq 2$. \cref{lem:L2lde} similarly
implies that $\supin_2$ is finite.

In principle, one can therefore define an exact synthesis algorithm
for $\supin_4$ by explicitly constructing a circuit for every element
of the group using, e.g., the local algorithm of
\cref{sec:AGR_algo}. We now briefly outline a different approach to
solving this problem.

\begin{lemma}
  \label{lem:O4sglobalsyn_helper2}
  Let $U \in \supin_4$. If $\lde_{\sqrt{2}}(U) \geq 1$, then, up to row
  permutation and column permutation $\overline{U}$ is one of the binary
  patterns below.
  \[
  B_0 = \begin{bmatrix}
   \ 1 \ & \ 1 \ & \ 0 \ & \ 0 \ \\
    1 & 1 & 0 & 0\\
    0 & 0 & 0 & 0\\
    0 & 0 & 0 & 0
  \end{bmatrix}, \qquad B_1 = \begin{bmatrix}
    \ 1 \ & \ 1 \ & \ 0 \ & \ 0 \ \\
    1 & 1 & 0 & 0\\
    0 & 0 & 1 & 1\\
    0 & 0 & 1 & 1
  \end{bmatrix}, \qquad B_2 = \begin{bmatrix}
      \ 1 \ & \ 1 \ & \ 1 \ & \ 1 \ \\
       1 & 1 & 1 & 1\\
       1 & 1 & 1 & 1\\
       1 & 1 & 1 & 1
  \end{bmatrix}.
  \]
\end{lemma}

\begin{proof}
  By \cref{lem:L4lde}, we only need to consider the cases
  $\lde_{\sqrt{2}}(U)=1$ and $\lde_{\sqrt{2}}(U)=2$. When
  $\lde_{\sqrt{2}}(U)=2$, by \cref{lem:weight}, $\overline{U} =
  B_2$. When $\lde_{\sqrt{2}}(U)=1$, then the rows and columns of
  $\sqrt{2}U$ are integer vectors of norm no more than 2 and must
  therefore contain 0 or 2 odd entries. It then follows from
  \cref{lem:collision} that the only two possible binary patterns for
  $U$ are $B_0$ and $B_1$, up to row permutation and column
  permutation.
\end{proof}

\begin{proposition}
  \label{thm:O4sglobalsyn}
  Let $U \in \supin_4$. Then $U$ can be represented by $O(1)$
  generators in $\gensup_4$.
\end{proposition}

\begin{proof}
  Let $\lde_{\sqrt 2}(U)=k$. By \cref{lem:L4lde}, $k \leq 2$. When $k
  = 0$, $U$ is a signed permutation matrix and can therefore be
  written as a product of no more than $3$ two-level $X$ gates and $4$
  one-level $(-1)$ gates. When $k>0$, then, by
  \cref{lem:O4sglobalsyn_helper2}, $\overline{U}$ is one of $B_0$,
  $B_1$, or $B_2$. Since all of these binary patterns are row-paired,
  we can apply \cref{lem:rowpaired_red} to reduce the least scaled
  denominator exponent of $U$.
\end{proof}

The exact synthesis algorithm given in the proof of
\cref{thm:O4sglobalsyn} is the global synthesis algorithm for
$\supin_4$. The algorithm relies on \cref{lem:O4sglobalsyn_helper2},
which characterizes the possible binary patterns for elements of
$\supin_4$.

\subsection{The 3-Qubit Case}
\label{ssec:threequbits}

We now turn to the case of $\supin_8$ (and $\integral_8$). This case
is more complex than the one discussed in the previous section,
because $\supin_8$ is an infinite group. Luckily, the characterization
given in \cref{thm:binary_patterns} allows us to proceed as in
\cref{thm:O4sglobalsyn}.

\begin{proposition}
  \label{thm:L8globalsyn}
  Let $U \in \supin_8$ with $\lde_{\sqrt{2}}(U)=k$. Then $U$ can be
  represented by $O(k)$ generators in $\gensup_8$ using the
  \textit{global synthesis algorithm}.
\end{proposition}

\begin{proof}
  By induction on $k$. There are only finitely many elements in
  $\supin_8$ with $k \leq 1$, so each one of them can be represented
  by a product of $O(1)$ elements of $\gensup_8$. When $k \geq 2$, by
  \cref{thm:binary_patterns}, $\overline{U}$ must be one of the $14$
  binary patterns in \cref{fig:patterns}. When $\overline{U}$ is
  row-paired, by \cref{lem:rowpaired_red}, there exists some $P \in
  S_8$ such that
  \[
  \lde_{\sqrt{2}}(\left(I \otimes H\right)PU)\leq k-1.
  \] 
  If $\overline{U}$ is not row-paired, then, by inspection of
  \cref{fig:patterns}, $\overline{U}$ is neither row-paired nor
  column-paired and so, by \cref{lem:unice}, $U' = \left(I \otimes
  H\right)U\left(I \otimes H\right)$ is row-paired and
  $\lde_{\sqrt{2}}(U') \leq \lde_{\sqrt{2}}(U)$. Thus, by
  \cref{lem:rowpaired_red}, there exists $P \in S_8$ such that
  \[
  \lde_{\sqrt{2}}\left(\left(I \otimes H\right)PU'\right)\leq k-1.
  \]
  Continuing in this way, and writing each element of $S_8$ as a
  constant number of elements of $\gensup_8$, we obtain a sequence of
  $O(k)$ elements of $\gensup_8$ whose product represents $U$.
\end{proof}

We end this section by showing that the global synthesis algorithm for
$\supin_8$ given in \cref{thm:L8globalsyn} can be used to define a
global synthesis algorithm $\integral_8$ of similar asymptotic
cost. The idea is to consider an element $U$ of $\integral_8$ as an
element of $\supin_8$ (which is possible since $\integral_8 \subseteq
\supin_8$) and to apply the algorithm of \cref{thm:L8globalsyn} to
$U$. This yields a decomposition of $U$ that contains evenly many
$I\otimes H$ gates, but these can be removed through rewriting as in
\cite{li2021generators}.

\begin{lemma}
  \label{lem:rewrite}
  For any word $W$ over $\s{\mone{a}, \xx{a,b}~;~ 0\leq a<b<n}$, there
  exists a word $W'$ over $\gens_n$ such that $\left(I\otimes H\right)
  W = W'\left(I\otimes H\right)$. Moreover, if $W$ has length $\ell$,
  then $W'$ has length $c\ell$ for some positive integer $c$ that
  depends on $n$.
\end{lemma}

\begin{proof}
  Consider the relations below, where $a$ is assumed to be even in
  \Cref{eq:2,eq:44} and $a$ is assumed to be odd in \Cref{eq:3,eq:5}.
  \begin{align}
  (I \otimes H)(I \otimes H) &= \epsilon \label{eq:1}\\
  (I \otimes H)(-1)_{[a]} &= (-1)_{[a]}X_{[a,a+1]}(-1)_{[a]}(I \otimes H) \label{eq:2} \\
  (I \otimes H)(-1)_{[a]} &= X_{[a-1,a]}(I \otimes H) \label{eq:3}\\
  (I \otimes H)X_{[a,a+1]} &= (-1)_{[a+1]}(I \otimes H) \label{eq:44}\\
  (I \otimes H)X_{[a,a+1]} &= K_{[a-1,a,a+1,a+2]}X_{[a,a+1]}(I \otimes H) \label{eq:5}
  \end{align}
  The relations show that commuting $I \otimes H$ with $\mone{a}$ or
  $\xx{a,a+1}$ adds only a constant number of gates. To commute
  $I\otimes H$ with $\xx{a,b}$, one can first express $\xx{a,b}$ in
  terms of $\xx{a,a+1}$, and then apply the relations above. The
  result then follows by induction on the length of $W$.
\end{proof}

\begin{proposition}
  \label{thm:O8globalsyn}
  Let $U \in \integral_8$ with $\lde(U)=k$. Then $U$ can be
  represented by $O(k)$ generators in $\gens_8$ using the
  \textit{global synthesis algorithm}.
\end{proposition}

\begin{proof}
  By \cref{thm:L8globalsyn}, one can find a word $W$ of length $O(k)$
  over $\gensup_8$ that represents $U$ and contains evenly many
  occurrences of $I \otimes H$. By construction, each pair of
  $I\otimes H$ gates is separated by a word over $\s{\mone{a},
    \xx{a,b}~;~ 0\leq a<b<n}$ and can thus be eliminated by an
  application of \cref{lem:rewrite}. This yields a new word $W'$ over
  $\gens_8$ of length $O(k)$.
\end{proof}

\section{Conclusion}
\label{sec:conclusion}

In this paper, we studied the synthesis of Toffoli-Hadamard
circuits. We focused on circuits over the gate sets $\s{X, CX, CCX,
  H}$ and $\s{X, CX, CCX, K}$. Because circuits over these gate sets
correspond to matrices in the groups $\supin_n$ and $\integral_n$,
respectively, each circuit synthesis problem reduces to a
factorization problem in the corresponding matrix group. The existing
local synthesis algorithm was introduced in \cite{AGR2019}. We
proposed two alternative algorithms.

Our first algorithm, the Householder synthesis algorithm, is an
adaptation of prior work by Kliuchnikov
\cite{kliuchnikov2013synthesis} and applies to matrices of arbitrary
size. The Householder algorithm first factors the given matrix as a
product of reflection operators, and then synthesizes each reflection
in this factorization. The Householder algorithm uses an additional
qubit, but reduces the overall complexity of the synthesized circuit
from $O(2^n\log(n)k)$ to $O(n^2\log(n)k)$.

Our second algorithm, the global synthesis algorithm, is inspired by
prior work of Russell, Niemann and others
\cite{russell2014exact,niemann2020advanced}. The global algorithm
relies on a small dictionary of binary patterns which ensures that
every step of the algorithm strictly decreases the least denominator
exponent of the matrix to be synthesized. Because this second
algorithm only applies to matrices of dimension 2, 4, and 8, it is
difficult to compare its complexity with that of the other
methods. However, the global nature of the algorithm makes it
plausible that it would outperform the method of \cite{AGR2019} in
practice, and we leave this as an avenue for future research.

Looking forward, many questions remain. Firstly, it would be
interesting to compare the algorithms in practice. Further afield, we
would like to find a standalone global synthesis for $\integral_8$,
rather than relying on the corresponding result for $\supin_8$ and the
commutation of generators. This may require a careful study of residue
patterns modulo 4, rather than modulo 2, as we did here. Finally, we
hope to extend the global synthesis method to larger, or even
arbitrary, dimensions.

\section{Acknowledgement}
\label{sec:acknowledgement}

Part of this research was carried out during SML's undergraduate
honours work at Dalhousie University. The authors would like to thank
Jiaxin Huang and John van de Wetering for enlightening
discussions. The circuit diagrams in this paper were typeset using
Quantikz \cite{quantikz}.

\bibliographystyle{splncs04}
\bibliography{impsyn}

\newpage
\appendix
\addtocontents{toc}{\protect\setcounter{tocdepth}{0}}

\section{Proof of \cref{thm:binary_patterns}}
\label{chp:binaryPatterns}
\bt*

\begin{proof}
Let $u_i$ denote the $i$-th column of $U$, and $u_i^\dagger$ denote
the the $i$-th row of $U$, $0 \leq i < 8$. Let $\lVert v \rVert$
denote the hamming weight of $v$, where $v$ is a string of binary
bits. Proceed by case distinction.
\begin{mycases}
 \case There are identical rows or columns in $U$.  Up to
 transposition, suppose $U$ has two rows that are identical. By
 \Cref{prop:nice}, $U \in \mathcal{B}_0$ up to permutation.
 \case There are no identical rows or columns in $U$.  By
 \Cref{prop:notnice}, $U \in \mathcal{B}_1$ up to permutation.
\end{mycases}

\end{proof}

\subsection{Binary Patterns that are either Row-paired or Column-paired}

\begin{definition}
We define the set $\mathcal{B}_0$ of binary matrices as $\mathcal{B}_0 =
\{A,B,C,D,E,F,G,H,I,J,K\}$, where
{\scriptsize
\[
A=\begin{bmatrix}
    \ 1 \ & \ 1 \ & \ 1 \ & \ 1 \ & \ 1 \ & \ 1 \ & \ 1 \ & \ 1 \ \\
    1 & 1 & 1 & 1 & 1 & 1 & 1 & 1\\
    1 & 1 & 1 & 1 & 1 & 1 & 1 & 1\\
    1 & 1 & 1 & 1 & 1 & 1 & 1 & 1\\
    1 & 1 & 1 & 1 & 1 & 1 & 1 & 1\\
    1 & 1 & 1 & 1 & 1 & 1 & 1 & 1\\
    1 & 1 & 1 & 1 & 1 & 1 & 1 & 1\\
    1 & 1 & 1 & 1 & 1 & 1 & 1 & 1
  \end{bmatrix}, \quad B = \begin{bmatrix}
    \ 1 \ & \ 1 \ & \ 1 \ & \ 1 \ & \ 1 \ & \ 1 \ & \ 1 \ & \ 1 \ \\
    1 & 1 & 1 & 1 & 1 & 1 & 1 & 1\\
    1 & 1 & 1 & 1 & 1 & 1 & 1 & 1\\
    1 & 1 & 1 & 1 & 1 & 1 & 1 & 1\\
    1 & 1 & 1 & 1 & 0 & 0 & 0 & 0\\
    1 & 1 & 1 & 1 & 0 & 0 & 0 & 0\\
    1 & 1 & 1 & 1 & 0 & 0 & 0 & 0\\
    1 & 1 & 1 & 1 & 0 & 0 & 0 & 0
  \end{bmatrix}, \quad C = \begin{bmatrix}
    \ 1 \ & \ 1 \ & \ 1 \ & \ 1 \ & \ 1 \ & \ 1 \ & \ 1 \ & \ 1 \ \\
    1 & 1 & 1 & 1 & 1 & 1 & 1 & 1\\
    1 & 1 & 1 & 1 & 0 & 0 & 0 & 0\\
    1 & 1 & 1 & 1 & 0 & 0 & 0 & 0\\
    1 & 1 & 0 & 0 & 1 & 1 & 0 & 0\\
    1 & 1 & 0 & 0 & 1 & 1 & 0 & 0\\
    1 & 1 & 0 & 0 & 0 & 0 & 1 & 1\\
    1 & 1 & 0 & 0 & 0 & 0 & 1 & 1
  \end{bmatrix},
 \]
 \[
D=\begin{bmatrix}
   \ 1 \ & \ 1 \ & \ 1 \  & \ 1 \  & \ 0 \  & \ 0 \  & \ 0 \  & \ 0 \ \\
    1 & 1 & 1 & 1 & 0 & 0 & 0 & 0\\
    1 & 1 & 1 & 1 & 0 & 0 & 0 & 0\\
    1 & 1 & 1 & 1 & 0 & 0 & 0 & 0\\
    1 & 1 & 0 & 0 & 1 & 1 & 0 & 0\\
    1 & 1 & 0 & 0 & 1 & 1 & 0 & 0\\
    1 & 1 & 0 & 0 & 1 & 1 & 0 & 0\\
    1 & 1 & 0 & 0 & 1 & 1 & 0 & 0
  \end{bmatrix}, \quad E = \begin{bmatrix}
    \ 1 \ & \ 1 \ & \ 1 \ & \ 1 \ & \ 1 \ & \ 1 \ & \ 1 \ & \ 1 \ \\
    1 & 1 & 1 & 1 & 1 & 1 & 1 & 1\\
    1 & 1 & 1 & 1 & 0 & 0 & 0 & 0\\
    1 & 1 & 1 & 1 & 0 & 0 & 0 & 0\\
    0 & 0 & 0 & 0 & 1 & 1 & 1 & 1\\
    0 & 0 & 0 & 0 & 1 & 1 & 1 & 1\\
    0 & 0 & 0 & 0 & 0 & 0 & 0 & 0\\
    0 & 0 & 0 & 0 & 0 & 0 & 0 & 0
  \end{bmatrix}, \quad F = \begin{bmatrix}
    \ 1 \ & \ 1 \ & \ 1 \  & \ 1 \  & \ 0 \  & \ 0 \  & \ 0 \  & \ 0 \ \\
    1 & 1 & 1 & 1 & 0 & 0 & 0 & 0\\
    1 & 1 & 0 & 0 & 1 & 1 & 0 & 0\\
    1 & 1 & 0 & 0 & 1 & 1 & 0 & 0\\
    1 & 0 & 1 & 0 & 1 & 0 & 1 & 0\\
    1 & 0 & 1 & 0 & 1 & 0 & 1 & 0\\
    1 & 0 & 0 & 1 & 0 & 1 & 1 & 0\\
    1 & 0 & 0 & 1 & 0 & 1 & 1 & 0
  \end{bmatrix},
\]

\[
G=\begin{bmatrix}
    \ 1 \ & \ 1 \ & \ 1 \  & \ 1 \  & \ 0 \  & \ 0 \  & \ 0 \  & \ 0 \ \\
    1 & 1 & 1 & 1 & 0 & 0 & 0 & 0\\
    1 & 1 & 0 & 0 & 1 & 1 & 0 & 0\\
    1 & 1 & 0 & 0 & 1 & 1 & 0 & 0\\
    0 & 0 & 1 & 1 & 1 & 1 & 0 & 0\\
    0 & 0 & 1 & 1 & 1 & 1 & 0 & 0\\
    0 & 0 & 0 & 0 & 0 & 0 & 0 & 0\\
    0 & 0 & 0 & 0 & 0 & 0 & 0 & 0
  \end{bmatrix}, \quad H = \begin{bmatrix}
    \ 1 \ & \ 1 \ & \ 1 \  & \ 1 \  & \ 0 \  & \ 0 \  & \ 0 \  & \ 0 \ \\
    1 & 1 & 1 & 1 & 0 & 0 & 0 & 0\\
    1 & 1 & 0 & 0 & 1 & 1 & 0 & 0\\
    1 & 1 & 0 & 0 & 1 & 1 & 0 & 0\\
    0 & 0 & 1 & 1 & 0 & 0 & 1 & 1\\
    0 & 0 & 1 & 1 & 0 & 0 & 1 & 1\\
    0 & 0 & 0 & 0 & 1 & 1 & 1 & 1\\
    0 & 0 & 0 & 0 & 1 & 1 & 1 & 1
  \end{bmatrix}, \quad I = \begin{bmatrix}
    \ 1 \ & \ 1 \ & \ 1 \  & \ 1 \  & \ 0 \  & \ 0 \  & \ 0 \  & \ 0 \ \\
    1 & 1 & 1 & 1 & 0 & 0 & 0 & 0\\
    1 & 1 & 1 & 1 & 0 & 0 & 0 & 0\\
    1 & 1 & 1 & 1 & 0 & 0 & 0 & 0\\
    0 & 0 & 0 & 0 & 1 & 1 & 1 & 1\\
    0 & 0 & 0 & 0 & 1 & 1 & 1 & 1\\
    0 & 0 & 0 & 0 & 1 & 1 & 1 & 1\\
    0 & 0 & 0 & 0 & 1 & 1 & 1 & 1
  \end{bmatrix},
\]

  \[
  J = \begin{bmatrix}
    \ 1 \ & \ 1 \ & \ 1 \  & \ 1 \  & \ 0 \  & \ 0 \  & \ 0 \  & \ 0 \ \\
    1 & 1 & 1 & 1 & 0 & 0 & 0 & 0\\
    1 & 1 & 1 & 1 & 0 & 0 & 0 & 0\\
    1 & 1 & 1 & 1 & 0 & 0 & 0 & 0\\
    0 & 0 & 0 & 0 & 0 & 0 & 0 & 0\\
    0 & 0 & 0 & 0 & 0 & 0 & 0 & 0\\
    0 & 0 & 0 & 0 & 0 & 0 & 0 & 0\\
    0 & 0 & 0 & 0 & 0 & 0 & 0 & 0
  \end{bmatrix}, \quad K = \begin{bmatrix}
    \ 1 \ & \ 1 \ & \ 1 \ & \ 1 \ & \ 1 \ & \ 1 \ & \ 1 \ & \ 1 \ \\
    1 & 1 & 1 & 1 & 1 & 1 & 1 & 1\\
    1 & 1 & 1 & 1 & 1 & 1 & 1 & 1\\
    1 & 1 & 1 & 1 & 1 & 1 & 1 & 1\\
    0 & 0 & 0 & 0 & 0 & 0 & 0 & 0\\
    0 & 0 & 0 & 0 & 0 & 0 & 0 & 0\\
    0 & 0 & 0 & 0 & 0 & 0 & 0 & 0\\
    0 & 0 & 0 & 0 & 0 & 0 & 0 & 0
  \end{bmatrix}.
\]  
}
\end{definition}

\begin{proposition}
\label{prop:nice}
Let $U \in \Z_2^{8 \times 8}$. Suppose $U$ satisfies
\cref{lem:collision,lem:weight}. If $U$ has two rows that are
identical, then $U \in \mathcal{B}_0$ up to permutation and
transposition.
\end{proposition}

\begin{proof}
Let $u_i$ denote the $i$-th column of $U$, and $u_i^\dagger$ denote
the the $i$-th row of $U$, $0 \leq i < 8$. Let $\lVert v \rVert$
denote the hamming weight of $v$, where $v$ is a string of binary
bits. Up to permutation, suppose $\lVert u_0^\dagger \rVert = \lVert
u_1^\dagger \rVert $. By \cref{lem:weight}, $\lVert u_0^\dagger \rVert
= 8$ or $\lVert u_0^\dagger \rVert = 4$. Proceed by case distinction,
we summarized the derivation of binary patterns in
\cref{fig:casedist1} and \cref{fig:casedist2}.

\begin{mycases}
 \case
          $\lVert u_0^\dagger \rVert = \lVert u_1^\dagger \rVert = 8$.
\subcase
          $\lVert u_0 \rVert = 8$.
\ssubcase
          $\lVert u_1 \rVert = 8$.
\sssubcase
          $\lVert u_2 \rVert = 8$.
\ssssubcase
          $\lVert u_3 \rVert = 8$.
\sssssubcase
          $\lVert u_4 \rVert = 8$, then
          \[
U = \begin{bmatrix}
  1 & 1 & 1 & 1 & 1 & 1 & 1 & 1 \\
  1 & 1 & 1 & 1 & 1 & 1 & 1 & 1 \\
  1 & 1 & 1 & 1 & 1 &   &   &   \\
  1 & 1 & 1 & 1 & 1 &   &   &   \\
  1 & 1 & 1 & 1 & 1 &   &   &   \\
  1 & 1 & 1 & 1 & 1 &   &   &   \\
  1 & 1 & 1 & 1 & 1 &   &   &   \\
  1 & 1 & 1 & 1 & 1 &   &   & 
\end{bmatrix} \ \xrightarrow{\cref{lem:weight}} \ U = \begin{bmatrix}
  1 & 1 & 1 & 1 & 1 & 1 & 1 & 1 \\
  1 & 1 & 1 & 1 & 1 & 1 & 1 & 1 \\
  1 & 1 & 1 & 1 & 1 & \color{IMSGreen}1 & \color{IMSGreen}1 & \color{IMSGreen}1 \\
  1 & 1 & 1 & 1 & 1 & \color{IMSGreen}1 & \color{IMSGreen}1 & \color{IMSGreen}1 \\
  1 & 1 & 1 & 1 & 1 & \color{IMSGreen}1 & \color{IMSGreen}1 & \color{IMSGreen}1 \\
  1 & 1 & 1 & 1 & 1 & \color{IMSGreen}1 & \color{IMSGreen}1 & \color{IMSGreen}1 \\
  1 & 1 & 1 & 1 & 1 & \color{IMSGreen}1 & \color{IMSGreen}1 & \color{IMSGreen}1 \\
  1 & 1 & 1 & 1 & 1 & \color{IMSGreen}1 & \color{IMSGreen}1 & \color{IMSGreen}1
\end{bmatrix} = A.
\]
\sssssubcase
$\lVert u_4 \rVert = 4$, then
\[
U = \begin{bmatrix}
  1 & 1 & 1 & 1 & 1 & 1 & 1 & 1 \\
  1 & 1 & 1 & 1 & 1 & 1 & 1 & 1 \\
  1 & 1 & 1 & 1 & 1 &   &   &   \\
  1 & 1 & 1 & 1 & 1 &   &   &   \\
  1 & 1 & 1 & 1 & 0 &   &   &   \\
  1 & 1 & 1 & 1 & 0 &   &   &   \\
  1 & 1 & 1 & 1 & 0 &   &   &   \\
  1 & 1 & 1 & 1 & 0 &   &   & 
\end{bmatrix} \ \xrightarrow{\cref{lem:weight}} \ U = \begin{bmatrix}
  1 & 1 & 1 & 1 & 1 & 1 & 1 & 1 \\
  1 & 1 & 1 & 1 & 1 & 1 & 1 & 1 \\
  1 & 1 & 1 & 1 & 1 & \color{IMSGreen}1 & \color{IMSGreen}1 & \color{IMSGreen}1 \\
  1 & 1 & 1 & 1 & 1 & \color{IMSGreen}1 & \color{IMSGreen}1 & \color{IMSGreen}1 \\
  1 & 1 & 1 & 1 & 0 & \color{IMSGreen}0 & \color{IMSGreen}0 & \color{IMSGreen}0 \\
  1 & 1 & 1 & 1 & 0 & \color{IMSGreen}0 & \color{IMSGreen}0 & \color{IMSGreen}0 \\
  1 & 1 & 1 & 1 & 0 & \color{IMSGreen}0 & \color{IMSGreen}0 & \color{IMSGreen}0 \\
  1 & 1 & 1 & 1 & 0 & \color{IMSGreen}0 & \color{IMSGreen}0 & \color{IMSGreen}0
\end{bmatrix} = B.
\]
\ssssubcase
          $\lVert u_3 \rVert = 4$. Let $(x,y)$ be a pair of the entries in the $i$-th column of rows $2$ and $3$ as shown below, for $4 \leq i < 8$. By \cref{lem:collision} with $u_3$, $x=y$. Hence $ u_2^\dagger = u_3 ^\dagger$.
          \[
U = \begin{bmatrix}
  1 & 1 & 1 & 1 & 1 & 1 & 1 & 1 \\
  1 & 1 & 1 & 1 & 1 & 1 & 1 & 1 \\
  1 & 1 & 1 & 1 & x  &   &   &   \\
  1 & 1 & 1 & 1 & y  &   &   &   \\
  1 & 1 & 1 & 0 &   &   &   &   \\
  1 & 1 & 1 & 0 &   &   &   &   \\
  1 & 1 & 1 & 0 &   &   &   &   \\
  1 & 1 & 1 & 0 &   &   &   & 
\end{bmatrix}.
\]
\setcounter{sssssubcases}{0}
\sssssubcase
$\lVert u_2^\dagger \rVert = 8$, then
\[
U = \begin{bmatrix}
  1 & 1 & 1 & 1 & 1 & 1 & 1 & 1 \\
  1 & 1 & 1 & 1 & 1 & 1 & 1 & 1 \\
  1 & 1 & 1 & 1 & 1 & 1 & 1 & 1 \\
  1 & 1 & 1 & 1 & 1 & 1 & 1 & 1 \\
  1 & 1 & 1 & 0 &   &   &   &   \\
  1 & 1 & 1 & 0 &   &   &   &   \\
  1 & 1 & 1 & 0 &   &   &   &   \\
  1 & 1 & 1 & 0 &   &   &   & 
\end{bmatrix} \ \xrightarrow{\cref{lem:weight}} \  \begin{bmatrix}
  1 & 1 & 1 & 1 & 1 & 1 & 1 & 1 \\
  1 & 1 & 1 & 1 & 1 & 1 & 1 & 1 \\
  1 & 1 & 1 & 1 & 1 & 1 & 1 & 1 \\
  1 & 1 & 1 & 1 & 1 & 1 & 1 & 1 \\
  1 & 1 & 1 & 0 & \color{IMSGreen}1 & \color{IMSGreen}0 & \color{IMSGreen}0 & \color{IMSGreen}0 \\
  1 & 1 & 1 & 0 & \color{IMSGreen}1 & \color{IMSGreen}0 & \color{IMSGreen}0 & \color{IMSGreen}0 \\
  1 & 1 & 1 & 0 & \color{IMSGreen}1 & \color{IMSGreen}0 & \color{IMSGreen}0 & \color{IMSGreen}0 \\
  1 & 1 & 1 & 0 & \color{IMSGreen}1 & \color{IMSGreen}0 & \color{IMSGreen}0 & \color{IMSGreen}0
\end{bmatrix}=B.
\]
\sssssubcase
$\lVert u_2^\dagger \rVert = 4$, then
\[
U = \begin{bmatrix}
  1 & 1 & 1 & 1 & 1 & 1 & 1 & 1 \\
  1 & 1 & 1 & 1 & 1 & 1 & 1 & 1 \\
  1 & 1 & 1 & 1 & 0 & 0 & 0 & 0 \\
  \color{red}1 & \color{red}1 & \color{red}1 & \color{red}1 & 0 & 0 & 0 & 0 \\
  \color{red}1 & \color{red}1 & \color{red}1 & \color{red}0 &  &   &   &   \\
  1 & 1 & 1 & 0 &  &   &   &   \\
  1 & 1 & 1 & 0 &  &   &   &   \\
  1 & 1 & 1 & 0 &  &   &   & 
\end{bmatrix}.
\]
Note that rows $3$ and $4$ violate \cref{lem:collision}, so this case is not possible.
\sssubcase $\lVert u_2 \rVert = 4$. Let $(x,y)$ be a pair of the
entries in the $i$-th column of rows $2$ and $3$ as shown below, for
$3 \leq i < 8$. By \cref{lem:collision} with $u_2$, $x=y$. Hence $
u_2^\dagger = u_3 ^\dagger$.
\[
U = \begin{bmatrix}
  1 & 1 & 1 & 1 & 1 & 1 & 1 & 1 \\
  1 & 1 & 1 & 1 & 1 & 1 & 1 & 1 \\
  1 & 1 & 1 & x  &   &   &   &   \\
  1 & 1 & 1 & y  &   &   &   &   \\
  1 & 1 & 0 &   &   &   &   &   \\
  1 & 1 & 0 &   &   &   &   &   \\
  1 & 1 & 0 &   &   &   &   &   \\
  1 & 1 & 0 &   &   &   &   & 
\end{bmatrix}.
\]
\setcounter{ssssubcases}{0}
\ssssubcase
\[
U = \begin{bmatrix}
  1 & 1 & 1 & 1 & 1 & 1 & 1 & 1 \\
  1 & 1 & 1 & 1 & 1 & 1 & 1 & 1 \\
  1 & 1 & 1 & 1 & 1 & 1 & 1 & 1 \\
  1 & 1 & 1 & 1 & 1 & 1 & 1 & 1 \\
  1 & 1 & 0 &   &   &   &   &   \\
  1 & 1 & 0 &   &   &   &   &   \\
  1 & 1 & 0 &   &   &   &   &   \\
  1 & 1 & 0 &   &   &   &   & 
\end{bmatrix}\ \xrightarrow{\cref{lem:weight}} \  \begin{bmatrix}
  1 & 1 & 1 & 1 & 1 & 1 & 1 & 1 \\
  1 & 1 & 1 & 1 & 1 & 1 & 1 & 1 \\
  1 & 1 & 1 & 1 & 1 & 1 & 1 & 1 \\
  1 & 1 & 1 & 1 & 1 & 1 & 1 & 1 \\
  1 & 1 & 0 & \color{IMSGreen}1 & \color{IMSGreen}1 & \color{IMSGreen}0 & \color{IMSGreen}0 & \color{IMSGreen}0 \\
  1 & 1 & 0 &   &   &   &   &   \\
  1 & 1 & 0 &   &   &   &   &   \\
  1 & 1 & 0 &   &   &   &   & 
\end{bmatrix} \ \xrightarrow{\cref{lem:weight}} \
\begin{bmatrix}
  1 & 1 & 1 & 1 & 1 & 1 & 1 & 1 \\
  1 & 1 & 1 & 1 & 1 & 1 & 1 & 1 \\
  1 & 1 & 1 & 1 & 1 & 1 & 1 & 1 \\
  1 & 1 & 1 & 1 & 1 & 1 & 1 & 1 \\
  1 & 1 & 0 & 1 & 1 & 0 & 0 & 0 \\
  1 & 1 & 0 & \color{IMSGreen}1 & \color{IMSGreen}1 &   &   &   \\
  1 & 1 & 0 & \color{IMSGreen}1 & \color{IMSGreen}1 &   &   &   \\
  1 & 1 & 0 & \color{IMSGreen}1 & \color{IMSGreen}1 &   &   & 
\end{bmatrix}.
\]
By \cref{lem:weight} for rows $5$, $6$, and $7$, we have
\[
U = \begin{bmatrix}
  1 & 1 & 1 & 1 & 1 & 1 & 1 & 1 \\
  1 & 1 & 1 & 1 & 1 & 1 & 1 & 1 \\
  1 & 1 & 1 & 1 & 1 & 1 & 1 & 1 \\
  1 & 1 & 1 & 1 & 1 & 1 & 1 & 1 \\
  1 & 1 & 0 & 1 & 1 & 0 & 0 & 0 \\
  1 & 1 & 0 & 1 & 1 & \color{IMSGreen}0 & \color{IMSGreen}0 & \color{IMSGreen}0 \\
  1 & 1 & 0 & 1 & 1 & \color{IMSGreen}0 & \color{IMSGreen}0 & \color{IMSGreen}0 \\
  1 & 1 & 0 & 1 & 1 & \color{IMSGreen}0 & \color{IMSGreen}0 & \color{IMSGreen}0 \end{bmatrix}=B.
\]
\ssssubcase
$\lVert u_2^\dagger \rVert = 4$, then up to column permutation,
\[
U = \begin{bmatrix}
  1 & 1 & 1 & 1 & 1 & 1 & 1 & 1 \\
  1 & 1 & 1 & 1 & 1 & 1 & 1 & 1 \\
  1 & 1 & 1 & 1 & 0 & 0 & 0 & 0 \\
  1 & 1 & 1 & 1 & 0 & 0 & 0 & 0 \\
  1 & 1 & 0 &   &   &   &   &   \\
  1 & 1 & 0 &   &   &   &   &   \\
  1 & 1 & 0 &   &   &   &   &   \\
  1 & 1 & 0 &   &   &   &   & 
\end{bmatrix} \ \xrightarrow{\cref{lem:collision}} \  \begin{bmatrix}
  1 & 1 & 1 & 1 & 1 & 1 & 1 & 1 \\
  1 & 1 & 1 & 1 & 1 & 1 & 1 & 1 \\
  1 & 1 & 1 & 1 & 0 & 0 & 0 & 0 \\
  1 & 1 & 1 & 1 & 0 & 0 & 0 & 0 \\
  1 & 1 & 0 & \color{IMSGreen}0 &   &   &   &   \\
  1 & 1 & 0 & \color{IMSGreen}0 &   &   &   &   \\
  1 & 1 & 0 & \color{IMSGreen}0 &   &   &   &   \\
  1 & 1 & 0 & \color{IMSGreen}0 &   &   &   & 
\end{bmatrix}\ \xrightarrow{\cref{lem:weight}} \ \begin{bmatrix}
  1 & 1 & 1 & 1 & 1 & 1 & 1 & 1 \\
  1 & 1 & 1 & 1 & 1 & 1 & 1 & 1 \\
  1 & 1 & 1 & 1 & 0 & 0 & 0 & 0 \\
  1 & 1 & 1 & 1 & 0 & 0 & 0 & 0 \\
  1 & 1 & 0 & 0 & \color{IMSGreen}1 & \color{IMSGreen}1 & \color{IMSGreen}0 & \color{IMSGreen}0 \\
  1 & 1 & 0 & 0 & x & y &   &   \\
  1 & 1 & 0 & 0 &   &   &   &   \\
  1 & 1 & 0 & 0 &   &   &   & 
\end{bmatrix}
\]
\[
\xrightarrow{\cref{lem:collision}} \  \begin{bmatrix}
  1 & 1 & 1 & 1 & 1 & 1 & 1 & 1 \\
  1 & 1 & 1 & 1 & 1 & 1 & 1 & 1 \\
  1 & 1 & 1 & 1 & 0 & 0 & 0 & 0 \\
  1 & 1 & 1 & 1 & 0 & 0 & 0 & 0 \\
  1 & 1 & 0 & 0 & 1 & 1 & 0 & 0 \\
  1 & 1 & 0 & 0 & \color{IMSGreen}1 & \color{IMSGreen}1 &   &   \\
  1 & 1 & 0 & 0 & \color{IMSGreen}0 & \color{IMSGreen}0 &   &   \\
  1 & 1 & 0 & 0 & \color{IMSGreen}0 & \color{IMSGreen}0 &   & 
\end{bmatrix}\ \xrightarrow{\cref{lem:weight}} \ \begin{bmatrix}
  1 & 1 & 1 & 1 & 1 & 1 & 1 & 1 \\
  1 & 1 & 1 & 1 & 1 & 1 & 1 & 1 \\
  1 & 1 & 1 & 1 & 0 & 0 & 0 & 0 \\
  1 & 1 & 1 & 1 & 0 & 0 & 0 & 0 \\
  1 & 1 & 0 & 0 & 1 & 1 & 0 & 0 \\
  1 & 1 & 0 & 0 & 1 & 1 & 0 & 0 \\
  1 & 1 & 0 & 0 & 0 & 0 & \color{IMSGreen}1 & \color{IMSGreen}1 \\
  1 & 1 & 0 & 0 & 0 & 0 & \color{IMSGreen}1 & \color{IMSGreen}1
\end{bmatrix}=C.
\]
\ssubcase $\lVert u_1 \rVert = 4$. Let $(x,y)$ be a pair of the
entries in the $i$-th column of rows $2$ and $3$ as shown below, for
$2 \leq i < 8$. By \cref{lem:collision} with $u_1$, $x=y$. Hence $
u_2^\dagger = u_3^\dagger$.

\[
U = \begin{bmatrix}
  1 & 1 & 1 & 1 & 1 & 1 & 1 & 1 \\
  1 & 1 & 1 & 1 & 1 & 1 & 1 & 1 \\
  1 & 1 & x &   &   &   &   &   \\
  1 & 1 & y &   &   &   &   &   \\
  1 & 0 &   &   &   &   &   &   \\
  1 & 0 &   &   &   &   &   &   \\
  1 & 0 &   &   &   &   &   &   \\
  1 & 0 &   &   &   &   &   & 
\end{bmatrix}.
\]
\setcounter{sssubcases}{0}
\sssubcase
$\lVert u_2^\dagger \rVert = 8$, then
\[
U = \begin{bmatrix}
  1 & 1 & 1 & 1 & 1 & 1 & 1 & 1 \\
  1 & 1 & 1 & 1 & 1 & 1 & 1 & 1 \\
  1 & 1 & 1 & 1 & 1 & 1 & 1 & 1 \\
  1 & 1 & 1 & 1 & 1 & 1 & 1 & 1 \\
  1 & 0 &   &   &   &   &   &   \\
  1 & 0 &   &   &   &   &   &   \\
  1 & 0 &   &   &   &   &   &   \\
  1 & 0 &   &   &   &   &   & 
\end{bmatrix} \ \xrightarrow{\cref{lem:weight}} \  \begin{bmatrix}
  1 & 1 & 1 & 1 & 1 & 1 & 1 & 1 \\
  1 & 1 & 1 & 1 & 1 & 1 & 1 & 1 \\
  1 & 1 & 1 & 1 & 1 & 1 & 1 & 1 \\
  1 & 1 & 1 & 1 & 1 & 1 & 1 & 1  \\
  1 & 0 & \color{IMSGreen}1  & \color{IMSGreen}1  & \color{IMSGreen}1  & \color{IMSGreen}0  &  \color{IMSGreen}0 & \color{IMSGreen}0  \\
  1 & 0 &   &   &   &   &   &   \\
  1 & 0 &   &   &   &   &   &   \\
  1 & 0 &   &   &   &   &   & 
\end{bmatrix} \ \xrightarrow{\cref{lem:weight}}
\]
\[
\begin{bmatrix}
  1 & 1 & 1 & 1 & 1 & 1 & 1 & 1 \\
  1 & 1 & 1 & 1 & 1 & 1 & 1 & 1 \\
  1 & 1 & 1 & 1 & 1 & 1 & 1 & 1 \\
  1 & 1 & 1 & 1 & 1 & 1 & 1 & 1  \\
  1 & 0 & 1 & 1 & 1 & 0 & 0 & 0  \\
  1 & 0 & \color{IMSGreen}1  &  \color{IMSGreen}1  & \color{IMSGreen}1   &   &   &   \\
  1 & 0 & \color{IMSGreen}1   &  \color{IMSGreen}1  &  \color{IMSGreen}1  &   &   &   \\
  1 & 0 & \color{IMSGreen}1   & \color{IMSGreen}1   & \color{IMSGreen}1   &   &   & 
\end{bmatrix} \ \xrightarrow{\cref{lem:weight}} \  \begin{bmatrix}
  1 & 1 & 1 & 1 & 1 & 1 & 1 & 1 \\
  1 & 1 & 1 & 1 & 1 & 1 & 1 & 1 \\
  1 & 1 & 1 & 1 & 1 & 1 & 1 & 1 \\
  1 & 1 & 1 & 1 & 1 & 1 & 1 & 1  \\
  1 & 0 & 1 & 1 & 1 & 0 & 0 & 0  \\
  1 & 0 & 1 & 1 & 1 & \color{IMSGreen}0  &  \color{IMSGreen}0  & \color{IMSGreen}0   \\
  1 & 0 & 1 & 1 & 1 & \color{IMSGreen}0   &  \color{IMSGreen}0  & \color{IMSGreen}0   \\
  1 & 0 & 1 & 1 & 1 & \color{IMSGreen}0   &  \color{IMSGreen}0  & \color{IMSGreen}0 
\end{bmatrix}=B.
\]
\sssubcase $\lVert u_2^\dagger \rVert = 4$. By \cref{lem:collision}
with row $3$, for $u_2$ and $u_3$, precisely one of them has hamming
weight $8$, and the other has hamming weight $4$. Up to column
permutation, let $\lVert u_2 \rVert = 8$ and $\lVert u_3 \rVert = 4$.

\[
U = \begin{bmatrix}
  1 & 1 & 1 & 1 & 1 & 1 & 1 & 1 \\
  1 & 1 & 1 & 1 & 1 & 1 & 1 & 1 \\
  1 & 1 & 1 & 1 & 0 & 0 & 0 & 0 \\
  1 & 1 & 1 & 1 & 0 & 0 & 0 & 0 \\
  1 & 0 & 1 & 0 &   &   &   &   \\
  1 & 0 & 1 & 0 &   &   &   &   \\
  1 & 0 & 1 & 0 &   &   &   &   \\
  1 & 0 & 1 & 0 &   &   &   &  
\end{bmatrix} \ \xrightarrow{\cref{lem:weight}} \  \begin{bmatrix}
  1 & 1 & 1 & 1 & 1 & 1 & 1 & 1 \\
  1 & 1 & 1 & 1 & 1 & 1 & 1 & 1 \\
  1 & 1 & 1 & 1 & 0 & 0 & 0 & 0 \\
  1 & 1 & 1 & 1 & 0 & 0 & 0 & 0 \\
  1 & 0 & 1 & 0 & \color{IMSGreen}1  &  \color{IMSGreen}1 &  \color{IMSGreen}0 & \color{IMSGreen}0  \\
  1 & 0 & 1 & 0 & \color{IMSGreen}1  &   &   &   \\
  1 & 0 & 1 & 0 & \color{IMSGreen}0  &   &   &   \\
  1 & 0 & 1 & 0 & \color{IMSGreen}0  &   &   & 
\end{bmatrix} \ \xrightarrow{\cref{lem:collision}} \  \begin{bmatrix}
  1 & 1 & 1 & 1 & 1 & 1 & 1 & 1 \\
  1 & 1 & 1 & 1 & 1 & 1 & 1 & 1 \\
  1 & 1 & 1 & 1 & 0 & 0 & 0 & 0 \\
  1 & 1 & 1 & 1 & 0 & 0 & 0 & 0 \\
  1 & 0 & 1 & 0 & 1 & 1 & 0 & 0  \\
  1 & 0 & 1 & 0 & 1 & \color{IMSGreen}1  &   &   \\
  1 & 0 & 1 & 0 & 0 &   &   &   \\
  1 & 0 & 1 & 0 & 0 &   &   & 
\end{bmatrix}
\]
\[
\xrightarrow{\cref{lem:weight}} \  \begin{bmatrix}
  1 & 1 & 1 & 1 & 1 & 1 & 1 & 1 \\
  1 & 1 & 1 & 1 & 1 & 1 & 1 & 1 \\
  1 & 1 & 1 & 1 & 0 & 0 & 0 & 0 \\
  1 & 1 & 1 & 1 & 0 & 0 & 0 & 0 \\
  1 & 0 & 1 & 0 & 1 & 1 & 0 & 0  \\
  1 & 0 & 1 & 0 & 1 & 1  & \color{IMSGreen}0  &  \color{IMSGreen}0 \\
  1 & 0 & 1 & 0 & 0 & \color{IMSGreen}0  &   &   \\
  1 & 0 & 1 & 0 & 0 & \color{IMSGreen}0  &   & 
\end{bmatrix} \ \xrightarrow{\cref{lem:weight}} \  \begin{bmatrix}
  1 & 1 & 1 & 1 & 1 & 1 & 1 & 1 \\
  1 & 1 & 1 & 1 & 1 & 1 & 1 & 1 \\
  1 & 1 & 1 & 1 & 0 & 0 & 0 & 0 \\
  1 & 1 & 1 & 1 & 0 & 0 & 0 & 0 \\
  1 & 0 & 1 & 0 & 1 & 1 & 0 & 0  \\
  1 & 0 & 1 & 0 & 1 & 1 & 0 & 0 \\
  1 & 0 & 1 & 0 & 0 & 0 & \color{IMSGreen}1  &  \color{IMSGreen}1 \\
  1 & 0 & 1 & 0 & 0 & 0 & \color{IMSGreen}1  & \color{IMSGreen}1
\end{bmatrix}=C.
\]
\subcase $\lVert u_0 \rVert = 4$. Let $(x,y)$ be a pair of the entries
in the $i$-th column of rows $2$ and $3$ as shown below, for $1 \leq i
< 8$. By \cref{lem:collision} with $u_0$, $x=y$. Hence $ u_2^\dagger =
u_3^\dagger$.

\[
U = \begin{bmatrix}
  1 & 1 & 1 & 1 & 1 & 1 & 1 & 1 \\
  1 & 1 & 1 & 1 & 1 & 1 & 1 & 1 \\
  1 & x &   &   &   &   &   &   \\
  1 & y &   &   &   &   &   &   \\
  0 &   &   &   &   &   &   &   \\
  0 &   &   &   &   &   &   &   \\
  0 &   &   &   &   &   &   &   \\
  0 &   &   &   &   &   &   & 
\end{bmatrix}.
\]
\setcounter{ssubcases}{0}
\ssubcase
$\lVert u_2^\dagger \rVert = 8$. By \cref{lem:weight}, $\lVert
u_4^\dagger \rVert = 4$ or $\lVert u_4^\dagger \rVert = 0$. Then we
have

\[
U = \begin{bmatrix}
  1 & 1 & 1 & 1 & 1 & 1 & 1 & 1 \\
  1 & 1 & 1 & 1 & 1 & 1 & 1 & 1 \\
  1 & 1 & 1 & 1 & 1 & 1 & 1 & 1 \\
  1 & 1 & 1 & 1 & 1 & 1 & 1 & 1 \\
  0 &   &   &   &   &   &   &   \\
  0 &   &   &   &   &   &   &   \\
  0 &   &   &   &   &   &   &   \\
  0 &   &   &   &   &   &   & 
\end{bmatrix}.
\]
\setcounter{sssubcases}{0}
\sssubcase
$\lVert u_4^\dagger \rVert = 4$, then
\[
U = \begin{bmatrix}
  1 & 1 & 1 & 1 & 1 & 1 & 1 & 1 \\
  1 & 1 & 1 & 1 & 1 & 1 & 1 & 1 \\
  1 & 1 & 1 & 1 & 1 & 1 & 1 & 1 \\
  1 & 1 & 1 & 1 & 1 & 1 & 1 & 1 \\
  0 & 1 & 1 & 1 & 1 & 0 & 0 & 0 \\
  0 &   &   &   &   &   &   &   \\
  0 &   &   &   &   &   &   &   \\
  0 &   &   &   &   &   &   & 
\end{bmatrix}  \ \xrightarrow{\cref{lem:weight}} \ \begin{bmatrix}
  1 & 1 & 1 & 1 & 1 & 1 & 1 & 1 \\
  1 & 1 & 1 & 1 & 1 & 1 & 1 & 1 \\
  1 & 1 & 1 & 1 & 1 & 1 & 1 & 1 \\
  1 & 1 & 1 & 1 & 1 & 1 & 1 & 1 \\
  0 & 1 & 1 & 1 & 1 & 0 & 0 & 0 \\
  0 & \color{IMSGreen}1 & \color{IMSGreen}1 & \color{IMSGreen}1 & \color{IMSGreen}1 & \color{IMSGreen}0 & \color{IMSGreen}0 & \color{IMSGreen}0 \\
  0 & \color{IMSGreen}1 & \color{IMSGreen}1 & \color{IMSGreen}1 & \color{IMSGreen}1 & \color{IMSGreen}0 & \color{IMSGreen}0 & \color{IMSGreen}0 \\
  0 & \color{IMSGreen}1 & \color{IMSGreen}1 & \color{IMSGreen}1 & \color{IMSGreen}1 & \color{IMSGreen}0 & \color{IMSGreen}0 & \color{IMSGreen}0
\end{bmatrix} = B.
\]
\sssubcase
$\lVert u_4^\dagger \rVert = 0$, then
\[
U = \begin{bmatrix}
  1 & 1 & 1 & 1 & 1 & 1 & 1 & 1 \\
  1 & 1 & 1 & 1 & 1 & 1 & 1 & 1 \\
  1 & 1 & 1 & 1 & 1 & 1 & 1 & 1 \\
  1 & 1 & 1 & 1 & 1 & 1 & 1 & 1 \\
  0 & 0 & 0 & 0 & 0 & 0 & 0 & 0 \\
  0 &   &   &   &   &   &   &   \\
  0 &   &   &   &   &   &   &   \\
  0 &   &   &   &   &   &   & 
\end{bmatrix}  \ \xrightarrow{\cref{lem:weight}} \ \begin{bmatrix}
  1 & 1 & 1 & 1 & 1 & 1 & 1 & 1 \\
  1 & 1 & 1 & 1 & 1 & 1 & 1 & 1 \\
  1 & 1 & 1 & 1 & 1 & 1 & 1 & 1 \\
  1 & 1 & 1 & 1 & 1 & 1 & 1 & 1 \\
  0 & 0 & 0 & 0 & 0 & 0 & 0 & 0 \\
  0 & \color{IMSGreen}0 & \color{IMSGreen}0 & \color{IMSGreen}0 & \color{IMSGreen}0 & \color{IMSGreen}0 & \color{IMSGreen}0 & \color{IMSGreen}0 \\
  0 & \color{IMSGreen}0 & \color{IMSGreen}0 & \color{IMSGreen}0 & \color{IMSGreen}0 & \color{IMSGreen}0 & \color{IMSGreen}0 & \color{IMSGreen}0 \\
  0 & \color{IMSGreen}0 & \color{IMSGreen}0 & \color{IMSGreen}0 & \color{IMSGreen}0 & \color{IMSGreen}0 & \color{IMSGreen}0 & \color{IMSGreen}0
\end{bmatrix} = K.
\]
\ssubcase
$\lVert u_2^\dagger \rVert = 4$, then the binary matrix is shown
below. By \cref{lem:collision}, there can be two cases: precisely two
of $\{u_1, u_2, u_3\}$ have hamming weight $8$, or all of $\{u_1, u_2,
u_3\}$ have hamming weight $4$.

\[
U = \begin{bmatrix}
  1 & 1 & 1 & 1 & 1 & 1 & 1 & 1 \\
  1 & 1 & 1 & 1 & 1 & 1 & 1 & 1 \\
  1 & 1 & 1 & 1 & 0 & 0 & 0 & 0 \\
  1 & 1 & 1 & 1 & 0 & 0 & 0 & 0 \\
  0 &   &   &   &   &   &   &   \\
  0 &   &   &   &   &   &   &   \\
  0 &   &   &   &   &   &   &   \\
  0 &   &   &   &   &   &   & 
\end{bmatrix}. 
\]

\setcounter{sssubcases}{0}
\sssubcase
 Up to column permutation, let $\lVert u_1 \rVert = \lVert u_2 \rVert = 8$ and $\lVert u_3 \rVert = 4$.

 \[
U = \begin{bmatrix}
  1 & 1 & 1 & 1 & 1 & 1 & 1 & 1 \\
  1 & 1 & 1 & 1 & 1 & 1 & 1 & 1 \\
  1 & 1 & 1 & 1 & 0 & 0 & 0 & 0 \\
  1 & 1 & 1 & 1 & 0 & 0 & 0 & 0 \\
  0 & 1 & 1 & 0 &   &   &   &   \\
  0 & 1 & 1 & 0 &   &   &   &   \\
  0 & 1 & 1 & 0 &   &   &   &   \\
  0 & 1 & 1 & 0 &   &   &   & 
\end{bmatrix} \ \xrightarrow{\cref{lem:weight}} \ \begin{bmatrix}
  1 & 1 & 1 & 1 & 1 & 1 & 1 & 1 \\
  1 & 1 & 1 & 1 & 1 & 1 & 1 & 1 \\
  1 & 1 & 1 & 1 & 0 & 0 & 0 & 0 \\
  1 & 1 & 1 & 1 & 0 & 0 & 0 & 0 \\
  0 & 1 & 1 & 0 & \color{IMSGreen}1 & \color{IMSGreen}1  & \color{IMSGreen}0  &  \color{IMSGreen}0 \\
  0 & 1 & 1 & 0 & \color{IMSGreen}1  &   &   &   \\
  0 & 1 & 1 & 0 & \color{IMSGreen}0 &   &   &   \\
  0 & 1 & 1 & 0 & \color{IMSGreen}0 &   &   & 
\end{bmatrix} \ \xrightarrow{\cref{lem:collision}} \  \begin{bmatrix}
  1 & 1 & 1 & 1 & 1 & 1 & 1 & 1 \\
  1 & 1 & 1 & 1 & 1 & 1 & 1 & 1 \\
  1 & 1 & 1 & 1 & 0 & 0 & 0 & 0 \\
  1 & 1 & 1 & 1 & 0 & 0 & 0 & 0 \\
  0 & 1 & 1 & 0 & 1 & 1 & 0 & 0 \\
  0 & 1 & 1 & 0 & 1  & \color{IMSGreen}1  &   &   \\
  0 & 1 & 1 & 0 & 0 &   &   &   \\
  0 & 1 & 1 & 0 & 0 &   &   & 
\end{bmatrix}
\]
\[
\xrightarrow{\cref{lem:weight}} \ \begin{bmatrix}
  1 & 1 & 1 & 1 & 1 & 1 & 1 & 1 \\
  1 & 1 & 1 & 1 & 1 & 1 & 1 & 1 \\
  1 & 1 & 1 & 1 & 0 & 0 & 0 & 0 \\
  1 & 1 & 1 & 1 & 0 & 0 & 0 & 0 \\
  0 & 1 & 1 & 0 & 1 & 1 & 0 & 0 \\
  0 & 1 & 1 & 0 & 1 & 1 & \color{IMSGreen}0 & \color{IMSGreen}0 \\
  0 & 1 & 1 & 0 & 0 & \color{IMSGreen}0 &   &   \\
  0 & 1 & 1 & 0 & 0 & \color{IMSGreen}0 &   & 
\end{bmatrix} \ \xrightarrow{\cref{lem:weight}} \ \begin{bmatrix}
  1 & 1 & 1 & 1 & 1 & 1 & 1 & 1 \\
  1 & 1 & 1 & 1 & 1 & 1 & 1 & 1 \\
  1 & 1 & 1 & 1 & 0 & 0 & 0 & 0 \\
  1 & 1 & 1 & 1 & 0 & 0 & 0 & 0 \\
  0 & 1 & 1 & 0 & 1 & 1 & 0 & 0 \\
  0 & 1 & 1 & 0 & 1 & 1 & 0 & 0 \\
  0 & 1 & 1 & 0 & 0 & 0 & \color{IMSGreen}1 & \color{IMSGreen}1 \\
  0 & 1 & 1 & 0 & 0 & 0 & \color{IMSGreen}1 & \color{IMSGreen}1
  \end{bmatrix}=C.
\]
\sssubcase$
\lVert u_1 \rVert = \lVert u_2 \rVert = \lVert u_3 \rVert = 4$.
\[
U = \begin{bmatrix}
  1 & 1 & 1 & 1 & 1 & 1 & 1 & 1 \\
  1 & 1 & 1 & 1 & 1 & 1 & 1 & 1 \\
  1 & 1 & 1 & 1 & 0 & 0 & 0 & 0 \\
  1 & 1 & 1 & 1 & 0 & 0 & 0 & 0 \\
  0 & 0 & 0 & 0 &   &   &   &   \\
  0 & 0 & 0 & 0 &   &   &   &   \\
  0 & 0 & 0 & 0 &   &   &   &   \\
  0 & 0 & 0 & 0 &   &   &   &   
\end{bmatrix} \ \xrightarrow{\cref{lem:weight}} \ \begin{bmatrix}
  1 & 1 & 1 & 1 & 1 & 1 & 1 & 1 \\
  1 & 1 & 1 & 1 & 1 & 1 & 1 & 1 \\
  1 & 1 & 1 & 1 & 0 & 0 & 0 & 0 \\
  1 & 1 & 1 & 1 & 0 & 0 & 0 & 0 \\
  0 & 0 & 0 & 0 & \color{IMSGreen}1  &   &   &   \\
  0 & 0 & 0 & 0 & \color{IMSGreen}1  &   &   &   \\
  0 & 0 & 0 & 0 & \color{IMSGreen}0  &   &   &   \\
  0 & 0 & 0 & 0 & \color{IMSGreen}0 &   &   &   
\end{bmatrix} \ \xrightarrow{\cref{lem:weight}} \ \begin{bmatrix}
  1 & 1 & 1 & 1 & 1 & 1 & 1 & 1 \\
  1 & 1 & 1 & 1 & 1 & 1 & 1 & 1 \\
  1 & 1 & 1 & 1 & 0 & 0 & 0 & 0 \\
  1 & 1 & 1 & 1 & 0 & 0 & 0 & 0 \\
  0 & 0 & 0 & 0 & 1 & \color{IMSGreen}1 & \color{IMSGreen}1 & \color{IMSGreen}1 \\
  0 & 0 & 0 & 0 & 1 & \color{IMSGreen}1 & \color{IMSGreen}1 & \color{IMSGreen}1 \\
  0 & 0 & 0 & 0 & 0  & \color{IMSGreen}0 & \color{IMSGreen}0 & \color{IMSGreen}0  \\
  0 & 0 & 0 & 0 & 0 & \color{IMSGreen}0 & \color{IMSGreen}0 & \color{IMSGreen}0  
\end{bmatrix} = E.
\]
\case
$\lVert u_0^\dagger \rVert = \lVert u_1^\dagger \rVert = 4$.
\subcase
$\lVert u_0 \rVert = 8$.
\ssubcase
$\lVert u_1 \rVert = 8$.
\sssubcase
$\lVert u_2 \rVert = 8$, then
\[
U = \begin{bmatrix}
  1 & 1 & 1 & 1 & 0 & 0 & 0 & 0 \\
  1 & 1 & 1 & 1 & 0 & 0 & 0 & 0 \\
  1 & 1 & 1 &   &   &   &   &   \\
  1 & 1 & 1 &   &   &   &   &   \\
  1 & 1 & 1 &   &   &   &   &   \\
  1 & 1 & 1 &   &   &   &   &   \\
  1 & 1 & 1 &   &   &   &   &   \\
  1 & 1 & 1 &   &   &   &   & 
\end{bmatrix}  \ \xrightarrow{\cref{lem:collision}} \ \begin{bmatrix}
  1 & 1 & 1 & 1 & 0 & 0 & 0 & 0 \\
  1 & 1 & 1 & 1 & 0 & 0 & 0 & 0 \\
  1 & 1 & 1 & \color{IMSGreen}1  &   &   &   &   \\
  1 & 1 & 1 & \color{IMSGreen}1  &   &   &   &   \\
  1 & 1 & 1 & \color{IMSGreen}1  &   &   &   &   \\
  1 & 1 & 1 & \color{IMSGreen}1  &   &   &   &   \\
  1 & 1 & 1 & \color{IMSGreen}1  &   &   &   &   \\
  1 & 1 & 1 & \color{IMSGreen}1  &   &   &   & 
\end{bmatrix}.
\]
\ssssubcase
$\lVert u_2^\dagger \rVert = 8$, then
\[
U = \begin{bmatrix}
  1 & 1 & 1 & 1 & 0 & 0 & 0 & 0 \\
  1 & 1 & 1 & 1 & 0 & 0 & 0 & 0 \\
  1 & 1 & 1 & 1 & 1 & 1 & 1 & 1 \\
  1 & 1 & 1 & 1 &   &   &   &   \\
  1 & 1 & 1 & 1 &   &   &   &   \\
  1 & 1 & 1 & 1 &   &   &   &   \\
  1 & 1 & 1 & 1 &   &   &   &   \\
  1 & 1 & 1 & 1 &   &   &   & 
\end{bmatrix}  \ \xrightarrow{\cref{lem:weight}} \  \begin{bmatrix}
  1 & 1 & 1 & 1 & 0 & 0 & 0 & 0 \\
  1 & 1 & 1 & 1 & 0 & 0 & 0 & 0 \\
  1 & 1 & 1 & 1 & 1 & 1 & 1 & 1 \\
  1 & 1 & 1 & 1 & \color{IMSGreen}1  &   &   &   \\
  1 & 1 & 1 & 1 & \color{IMSGreen}1 &   &   &   \\
  1 & 1 & 1 & 1 & \color{IMSGreen}1 &   &   &   \\
  1 & 1 & 1 & 1 & \color{IMSGreen}0 &   &   &   \\
  1 & 1 & 1 & 1 & \color{IMSGreen}0 &   &   & 
\end{bmatrix}  \ \xrightarrow{\cref{lem:weight}} \  \begin{bmatrix}
  1 & 1 & 1 & 1 & 0 & 0 & 0 & 0 \\
  1 & 1 & 1 & 1 & 0 & 0 & 0 & 0 \\
  1 & 1 & 1 & 1 & 1 & 1 & 1 & 1 \\
  1 & 1 & 1 & 1 & 1 & \color{IMSGreen}1 & \color{IMSGreen}1 & \color{IMSGreen}1 \\
  1 & 1 & 1 & 1 & 1 & \color{IMSGreen}1 & \color{IMSGreen}1 & \color{IMSGreen}1 \\
  1 & 1 & 1 & 1 & 1 & \color{IMSGreen}1 & \color{IMSGreen}1 & \color{IMSGreen}1 \\
  1 & 1 & 1 & 1 & 0 & \color{IMSGreen}0 & \color{IMSGreen}0 & \color{IMSGreen}0 \\
  1 & 1 & 1 & 1 & 0 & \color{IMSGreen}0 & \color{IMSGreen}0 & \color{IMSGreen}0
\end{bmatrix} = B.
\]
\ssssubcase
$\lVert u_2^\dagger \rVert = 4$. By \cref{lem:weight}, there can be
two cases: precisely four of
$\{u_3^\dagger,u_4^\dagger,u_5^\dagger,u_6^\dagger,u_7^\dagger\}$ have
hamming weight $8$, or all of
$\{u_3^\dagger,u_4^\dagger,u_5^\dagger,u_6^\dagger,u_7^\dagger\}$ have
hamming weight $4$.
\sssssubcase
 Up to row permutation, $\lVert u_3^\dagger \rVert = \lVert u_4^\dagger \rVert = \lVert u_5^\dagger \rVert = \lVert u_6^\dagger \rVert = 8$.
 \[
U = \begin{bmatrix}
  1 & 1 & 1 & 1 & 0 & 0 & 0 & 0 \\
  1 & 1 & 1 & 1 & 0 & 0 & 0 & 0 \\
  1 & 1 & 1 & 1 & 0 & 0 & 0 & 0 \\
  1 & 1 & 1 & 1 & 1 & 1 & 1 & 1 \\
  1 & 1 & 1 & 1 & 1 & 1 & 1 & 1 \\
  1 & 1 & 1 & 1 & 1 & 1 & 1 & 1 \\
  1 & 1 & 1 & 1 & 1 & 1 & 1 & 1 \\
  1 & 1 & 1 & 1 & 0 & 0 & 0 & 0
\end{bmatrix} = B.
\]
\sssssubcase
$\lVert u_3^\dagger \rVert = \lVert u_4^\dagger \rVert = \lVert u_5^\dagger \rVert = \lVert u_6^\dagger \rVert = \lVert u_7^\dagger \rVert=4$.

\[
U = \begin{bmatrix}
  1 & 1 & 1 & 1 & 0 & 0 & 0 & 0 \\
  1 & 1 & 1 & 1 & 0 & 0 & 0 & 0 \\
  1 & 1 & 1 & 1 & 0 & 0 & 0 & 0 \\
  1 & 1 & 1 & 1 & 0 & 0 & 0 & 0 \\
  1 & 1 & 1 & 1 & 0 & 0 & 0 & 0 \\
  1 & 1 & 1 & 1 & 0 & 0 & 0 & 0 \\
  1 & 1 & 1 & 1 & 0 & 0 & 0 & 0 \\
  1 & 1 & 1 & 1 & 0 & 0 & 0 & 0
\end{bmatrix} = K^\intercal.
\]
\sssubcase
$\lVert u_2 \rVert = 4$. Let $(x,y)$ be a pair of the entries in the
$i$-th column of rows $2$ and $3$ as shown below, for $3 \leq i <
8$. By \cref{lem:collision} with $u_3$, $x=y$. Hence $ u_2^\dagger =
u_3 ^\dagger$.
\[
U = \begin{bmatrix}
  1 & 1 & 1 & 1 & 0 & 0 & 0 & 0 \\
  1 & 1 & 1 & 1 & 0 & 0 & 0 & 0 \\
  1 & 1 & 1 &   &   &   &   &   \\
  1 & 1 & 1 &   &   &   &   &   \\
  1 & 1 & 0 &   &   &   &   &   \\
  1 & 1 & 0 &   &   &   &   &   \\
  1 & 1 & 0 &   &   &   &   &   \\
  1 & 1 & 0 &   &   &   &   & 
\end{bmatrix}  \ \xrightarrow{\cref{lem:collision}} \ \begin{bmatrix}
  1 & 1 & 1 & 1 & 0 & 0 & 0 & 0 \\
  1 & 1 & 1 & 1 & 0 & 0 & 0 & 0 \\
  1 & 1 & 1 & \color{IMSGreen}1 & x &   &   &   \\
  1 & 1 & 1 & \color{IMSGreen}1  & y &   &   &   \\
  1 & 1 & 0 & \color{IMSGreen}0 &   &   &   &   \\
  1 & 1 & 0 & \color{IMSGreen}0 &   &   &   &   \\
  1 & 1 & 0 & \color{IMSGreen}0 &   &   &   &   \\
  1 & 1 & 0 & \color{IMSGreen}0 &   &   &   & 
\end{bmatrix}.
\]
\setcounter{ssssubcases}{0}
\ssssubcase
$\lVert u_2^\dagger \rVert = 8$, then
\[
\begin{bmatrix}
  1 & 1 & 1 & 1 & 0 & 0 & 0 & 0 \\
  1 & 1 & 1 & 1 & 0 & 0 & 0 & 0 \\
  1 & 1 & 1 & 1 & 1 & 1 & 1 & 1 \\
  1 & 1 & 1 & 1 & 1 & 1 & 1 & 1 \\
  1 & 1 & 0 & 0 &   &   &   &   \\
  1 & 1 & 0 & 0 &   &   &   &   \\
  1 & 1 & 0 & 0 &   &   &   &   \\
  1 & 1 & 0 & 0 &   &   &   & 
\end{bmatrix}  \ \xrightarrow{\cref{lem:weight}} \ 
\begin{bmatrix}
  1 & 1 & 1 & 1 & 0 & 0 & 0 & 0 \\
  1 & 1 & 1 & 1 & 0 & 0 & 0 & 0 \\
  1 & 1 & 1 & 1 & 1 & 1 & 1 & 1 \\
  1 & 1 & 1 & 1 & 1 & 1 & 1 & 1 \\
  1 & 1 & 0 & 0 & \color{IMSGreen}1 & \color{IMSGreen}1 & \color{IMSGreen}0 & \color{IMSGreen}0 \\
  1 & 1 & 0 & 0 & \color{IMSGreen}1 &   &   &   \\
  1 & 1 & 0 & 0 & \color{IMSGreen}0 &   &   &   \\
  1 & 1 & 0 & 0 & \color{IMSGreen}0 &   &   & 
\end{bmatrix} \ \xrightarrow{\cref{lem:collision}} \ \begin{bmatrix}
  1 & 1 & 1 & 1 & 0 & 0 & 0 & 0 \\
  1 & 1 & 1 & 1 & 0 & 0 & 0 & 0 \\
  1 & 1 & 1 & 1 & 1 & 1 & 1 & 1 \\
  1 & 1 & 1 & 1 & 1 & 1 & 1 & 1 \\
  1 & 1 & 0 & 0 & 1 & 1 & 0 & 0 \\
  1 & 1 & 0 & 0 & 1 & \color{IMSGreen}1  &   &   \\
  1 & 1 & 0 & 0 & 0 &   &   &   \\
  1 & 1 & 0 & 0 & 0 &   &   & 
\end{bmatrix}
\]
\[
\xrightarrow{\cref{lem:weight}} \ \begin{bmatrix}
  1 & 1 & 1 & 1 & 0 & 0 & 0 & 0 \\
  1 & 1 & 1 & 1 & 0 & 0 & 0 & 0 \\
  1 & 1 & 1 & 1 & 1 & 1 & 1 & 1 \\
  1 & 1 & 1 & 1 & 1 & 1 & 1 & 1 \\
  1 & 1 & 0 & 0 & 1 & 1 & 0 & 0 \\
  1 & 1 & 0 & 0 & 1 & 1  & \color{IMSGreen}0 & \color{IMSGreen}0 \\
  1 & 1 & 0 & 0 & 0 & \color{IMSGreen}0 &   &   \\
  1 & 1 & 0 & 0 & 0 & \color{IMSGreen}0 &   & 
\end{bmatrix} \ \xrightarrow{\cref{lem:weight}} \ \begin{bmatrix}
  1 & 1 & 1 & 1 & 0 & 0 & 0 & 0 \\
  1 & 1 & 1 & 1 & 0 & 0 & 0 & 0 \\
  1 & 1 & 1 & 1 & 1 & 1 & 1 & 1 \\
  1 & 1 & 1 & 1 & 1 & 1 & 1 & 1 \\
  1 & 1 & 0 & 0 & 1 & 1 & 0 & 0 \\
  1 & 1 & 0 & 0 & 1 & 1 & 0 & 0 \\
  1 & 1 & 0 & 0 & 0 & 0 & \color{IMSGreen}1  & \color{IMSGreen}1  \\
  1 & 1 & 0 & 0 & 0 & 0 & \color{IMSGreen}1  & \color{IMSGreen}1 
\end{bmatrix} = C.
\]
\ssssubcase
$\lVert u_2^\dagger \rVert = 4$, then
\[
\begin{bmatrix}
  1 & 1 & 1 & 1 & 0 & 0 & 0 & 0 \\
  1 & 1 & 1 & 1 & 0 & 0 & 0 & 0 \\
  1 & 1 & 1 & 1 & 0 & 0 & 0 & 0 \\
  1 & 1 & 1 & 1 & 0 & 0 & 0 & 0 \\
  1 & 1 & 0 & 0 &   &   &   &   \\
  1 & 1 & 0 & 0 &   &   &   &   \\
  1 & 1 & 0 & 0 &   &   &   &   \\
  1 & 1 & 0 & 0 &   &   &   & 
\end{bmatrix} \ \xrightarrow{\cref{lem:weight}} \ \begin{bmatrix}
  1 & 1 & 1 & 1 & 0 & 0 & 0 & 0 \\
  1 & 1 & 1 & 1 & 0 & 0 & 0 & 0 \\
  1 & 1 & 1 & 1 & 0 & 0 & 0 & 0 \\
  1 & 1 & 1 & 1 & 0 & 0 & 0 & 0 \\
  1 & 1 & 0 & 0 & \color{IMSGreen}1 & \color{IMSGreen}1 & \color{IMSGreen}0 & \color{IMSGreen}0 \\
  1 & 1 & 0 & 0 & \color{IMSGreen}1 & \color{IMSGreen}1 & \color{IMSGreen}0 & \color{IMSGreen}0 \\
  1 & 1 & 0 & 0 & \color{IMSGreen}1 & \color{IMSGreen}1 & \color{IMSGreen}0 & \color{IMSGreen}0 \\
  1 & 1 & 0 & 0 & \color{IMSGreen}1 & \color{IMSGreen}1 & \color{IMSGreen}0 & \color{IMSGreen}0
\end{bmatrix} = D.
\]
\ssubcase
$\lVert u_1 \rVert = 4$. Let $x,y,z,w$ be the entries in $U$ as shown below. By \cref{lem:collision} with row $1$, $x=y$ and $z=w$. By \cref{lem:collision} with column $1$, $x=z$ and $y=w$. Hence $x=y=z=w$. Moreover, since $(x,z)$ can be any pair of the entries coming from any column $i$ of rows $2$ and $3$, for $2 \leq i < 8$, we have $u_2^\dagger = u_3^\dagger$.
\[
U = \begin{bmatrix}
  1 & 1 & 1 & 1 & 0 & 0 & 0 & 0 \\
  1 & 1 & 1 & 1 & 0 & 0 & 0 & 0 \\
  1 & 1 & x & y  &   &   &   &   \\
  1 & 1 & z & w  &   &   &   &   \\
  1 & 0 &  &   &   &   &   &   \\
  1 & 0 &  &   &   &   &   &   \\
  1 & 0 &  &   &   &   &   &   \\
  1 & 0 &  &   &   &   &   & 
\end{bmatrix}
\]
\setcounter{sssubcases}{0}
\sssubcase
$x=y=z=w=1$, then
\[
U = \begin{bmatrix}
  1 & 1 & 1 & 1 & 0 & 0 & 0 & 0 \\
  1 & 1 & 1 & 1 & 0 & 0 & 0 & 0 \\
  1 & 1 & 1 & 1 &   &   &   &   \\
  1 & 1 & 1 & 1 &   &   &   &   \\
  1 & 0 &  &   &   &   &   &   \\
  1 & 0 &  &   &   &   &   &   \\
  1 & 0 &  &   &   &   &   &   \\
  1 & 0 &  &   &   &   &   & 
\end{bmatrix} \ \xrightarrow{\cref{lem:collision}} \ \begin{bmatrix}
  1 & 1 & 1 & 1 & 0 & 0 & 0 & 0 \\
  1 & 1 & 1 & 1 & 0 & 0 & 0 & 0 \\
  1 & 1 & 1 & 1 &   &   &   &   \\
  1 & 1 & 1 & 1 &   &   &   &   \\
  1 & 0 & \color{IMSGreen}1 & \color{IMSGreen}0  &   &   &   &   \\
  1 & 0 &  &   &   &   &   &   \\
  1 & 0 &  &   &   &   &   &   \\
  1 & 0 &  &   &   &   &   & 
\end{bmatrix} \ \xrightarrow{\cref{lem:weight}} \ \begin{bmatrix}
  1 & 1 & 1 & 1 & 0 & 0 & 0 & 0 \\
  1 & 1 & 1 & 1 & 0 & 0 & 0 & 0 \\
  1 & 1 & 1 & 1 &   &   &   &   \\
  1 & 1 & 1 & 1 &   &   &   &   \\
  1 & 0 & 1 & 0  &   &   &   &   \\
  1 & 0 & \color{IMSGreen}1 & \color{IMSGreen}0  &   &   &   &   \\
  1 & 0 & \color{IMSGreen}1 & \color{IMSGreen}0  &   &   &   &   \\
  1 & 0 & \color{IMSGreen}1 & \color{IMSGreen}0  &   &   &   & 
\end{bmatrix}
\]
\setcounter{ssssubcases}{0}
\ssssubcase
$\lVert u_2^\dagger \rVert = 8$, then
\[
U=\begin{bmatrix}
  1 & 1 & 1 & 1 & 0 & 0 & 0 & 0 \\
  1 & 1 & 1 & 1 & 0 & 0 & 0 & 0 \\
  1 & 1 & 1 & 1 & 1 & 1 & 1 & 1 \\
  1 & 1 & 1 & 1 & 1 & 1 & 1 & 1 \\
  1 & 0 & 1 & 0 &   &   &   &   \\
  1 & 0 & 1 & 0 &   &   &   &   \\
  1 & 0 & 1 & 0 &   &   &   &   \\
  1 & 0 & 1 & 0 &   &   &   & 
\end{bmatrix} \ \xrightarrow{\cref{lem:weight}} \ \begin{bmatrix}
  1 & 1 & 1 & 1 & 0 & 0 & 0 & 0 \\
  1 & 1 & 1 & 1 & 0 & 0 & 0 & 0 \\
  1 & 1 & 1 & 1 & 1 & 1 & 1 & 1 \\
  1 & 1 & 1 & 1 & 1 & 1 & 1 & 1 \\
  1 & 0 & 1 & 0 & \color{IMSGreen}1  & \color{IMSGreen}1  & \color{IMSGreen}0 &  \color{IMSGreen}0 \\
  1 & 0 & 1 & 0 & \color{IMSGreen}1 &   &   &   \\
  1 & 0 & 1 & 0 & \color{IMSGreen}0 &   &   &   \\
  1 & 0 & 1 & 0 & \color{IMSGreen}0 &   &   & 
\end{bmatrix} \ \xrightarrow{\cref{lem:collision}} \ \begin{bmatrix}
  1 & 1 & 1 & 1 & 0 & 0 & 0 & 0 \\
  1 & 1 & 1 & 1 & 0 & 0 & 0 & 0 \\
  1 & 1 & 1 & 1 & 1 & 1 & 1 & 1 \\
  1 & 1 & 1 & 1 & 1 & 1 & 1 & 1 \\
  1 & 0 & 1 & 0 & 1 & 1 & 0 & 0 \\
  1 & 0 & 1 & 0 & 1 & \color{IMSGreen}1  &   &   \\
  1 & 0 & 1 & 0 & 0 &   &   &   \\
  1 & 0 & 1 & 0 & 0 &   &   & 
\end{bmatrix}
\]
\[
\xrightarrow{\cref{lem:weight}} \ \begin{bmatrix}
  1 & 1 & 1 & 1 & 0 & 0 & 0 & 0 \\
  1 & 1 & 1 & 1 & 0 & 0 & 0 & 0 \\
  1 & 1 & 1 & 1 & 1 & 1 & 1 & 1 \\
  1 & 1 & 1 & 1 & 1 & 1 & 1 & 1 \\
  1 & 0 & 1 & 0 & 1 & 1 & 0 & 0 \\
  1 & 0 & 1 & 0 & 1 & 1  & \color{IMSGreen}0  & \color{IMSGreen}0  \\
  1 & 0 & 1 & 0 & 0 & \color{IMSGreen}0 &   &   \\
  1 & 0 & 1 & 0 & 0 & \color{IMSGreen}0 &   & 
\end{bmatrix} \ \xrightarrow{\cref{lem:weight}} \ \begin{bmatrix}
  1 & 1 & 1 & 1 & 0 & 0 & 0 & 0 \\
  1 & 1 & 1 & 1 & 0 & 0 & 0 & 0 \\
  1 & 1 & 1 & 1 & 1 & 1 & 1 & 1 \\
  1 & 1 & 1 & 1 & 1 & 1 & 1 & 1 \\
  1 & 0 & 1 & 0 & 1 & 1 & 0 & 0 \\
  1 & 0 & 1 & 0 & 1 & 1 & 0 & 0  \\
  1 & 0 & 1 & 0 & 0 & 0 & \color{IMSGreen}1 & \color{IMSGreen}1 \\
  1 & 0 & 1 & 0 & 0 & 0 & \color{IMSGreen}1 & \color{IMSGreen}1
\end{bmatrix} = C.
\]
\ssssubcase
$\lVert u_2^\dagger \rVert = 4$, then
\[
U=\begin{bmatrix}
  1 & 1 & 1 & 1 & 0 & 0 & 0 & 0 \\
  1 & 1 & 1 & 1 & 0 & 0 & 0 & 0 \\
  1 & 1 & 1 & 1 & 0 & 0 & 0 & 0 \\
  1 & 1 & 1 & 1 & 0 & 0 & 0 & 0 \\
  1 & 0 & 1 & 0 &   &   &   &   \\
  1 & 0 & 1 & 0 &   &   &   &   \\
  1 & 0 & 1 & 0 &   &   &   &   \\
  1 & 0 & 1 & 0 &   &   &   & 
\end{bmatrix}  \ \xrightarrow{\cref{lem:weight}} \ \begin{bmatrix}
  1 & 1 & 1 & 1 & 0 & 0 & 0 & 0 \\
  1 & 1 & 1 & 1 & 0 & 0 & 0 & 0 \\
  1 & 1 & 1 & 1 & 0 & 0 & 0 & 0 \\
  1 & 1 & 1 & 1 & 0 & 0 & 0 & 0 \\
  1 & 0 & 1 & 0 & \color{IMSGreen}1  & \color{IMSGreen}1  & \color{IMSGreen}0 & \color{IMSGreen}0 \\
  1 & 0 & 1 & 0 & \color{IMSGreen}1  &   &   &   \\
  1 & 0 & 1 & 0 & \color{IMSGreen}1  &   &   &   \\
  1 & 0 & 1 & 0 & \color{IMSGreen}1  &   &   & 
\end{bmatrix}  \ \xrightarrow{\cref{lem:collision}} 
\]
\[
\begin{bmatrix}
  1 & 1 & 1 & 1 & 0 & 0 & 0 & 0 \\
  1 & 1 & 1 & 1 & 0 & 0 & 0 & 0 \\
  1 & 1 & 1 & 1 & 0 & 0 & 0 & 0 \\
  1 & 1 & 1 & 1 & 0 & 0 & 0 & 0 \\
  1 & 0 & 1 & 0 & 1 & 1 & 0 & 0 \\
  1 & 0 & 1 & 0 & 1 & \color{IMSGreen}1  &   &   \\
  1 & 0 & 1 & 0 & 1 & \color{IMSGreen}1  &   &   \\
  1 & 0 & 1 & 0 & 1 & \color{IMSGreen}1  &   & 
\end{bmatrix} \ \xrightarrow{\cref{lem:weight}} \ \begin{bmatrix}
  1 & 1 & 1 & 1 & 0 & 0 & 0 & 0 \\
  1 & 1 & 1 & 1 & 0 & 0 & 0 & 0 \\
  1 & 1 & 1 & 1 & 0 & 0 & 0 & 0 \\
  1 & 1 & 1 & 1 & 0 & 0 & 0 & 0 \\
  1 & 0 & 1 & 0 & 1 & 1 & 0 & 0 \\
  1 & 0 & 1 & 0 & 1 & 1 & \color{IMSGreen}0 & \color{IMSGreen}0 \\
  1 & 0 & 1 & 0 & 1 & 1 & \color{IMSGreen}0 & \color{IMSGreen}0 \\
  1 & 0 & 1 & 0 & 1 & 1 & \color{IMSGreen}0 & \color{IMSGreen}0
\end{bmatrix}=D.
\]
\sssubcase
$x=y=z=w=0$, then we have what follows.
\[
U = \begin{bmatrix}
  1 & 1 & 1 & 1 & 0 & 0 & 0 & 0 \\
  1 & 1 & 1 & 1 & 0 & 0 & 0 & 0 \\
  1 & 1 & 0 & 0 &  &   &   &   \\
  1 & 1 & 0 & 0 &  &   &   &   \\
  1 & 0 &  &   &   &   &   &   \\
  1 & 0 &  &   &   &   &   &   \\
  1 & 0 &  &   &   &   &   &   \\
  1 & 0 &  &   &   &   &   & 
\end{bmatrix}  \ \xrightarrow{\cref{lem:weight}} \ \begin{bmatrix}
  1 & 1 & 1 & 1 & 0 & 0 & 0 & 0 \\
  1 & 1 & 1 & 1 & 0 & 0 & 0 & 0 \\
  1 & 1 & 0 & 0 & \color{IMSGreen}1  & \color{IMSGreen}1  & \color{IMSGreen}0 & \color{IMSGreen}0 \\
  1 & 1 & 0 & 0 & \color{IMSGreen}1  & \color{IMSGreen}1  & \color{IMSGreen}0 & \color{IMSGreen}0 \\
  1 & 0 & \color{IMSGreen}1 &   &   &   &   &   \\
  1 & 0 & \color{IMSGreen}1 &   &   &   &   &   \\
  1 & 0 & \color{IMSGreen}0 &   &   &   &   &   \\
  1 & 0 & \color{IMSGreen}0 &   &   &   &   & 
\end{bmatrix}\ \xrightarrow{\cref{lem:collision}}
\]
Let $(x,y)$ be a pair of the entries in the $i$-th column of rows $4$ and $5$ as shown below, for $4 \leq i < 8$. By \cref{lem:collision} with $u_2$, $x=y$. Hence $ u_4^\dagger = u_5^\dagger$.
\[
\begin{bmatrix}
  1 & 1 & 1 & 1 & 0 & 0 & 0 & 0 \\
  1 & 1 & 1 & 1 & 0 & 0 & 0 & 0 \\
  1 & 1 & 0 & 0 & 1 & 1 & 0 & 0 \\
  1 & 1 & 0 & 0 & 1 & 1 & 0 & 0 \\
  1 & 0 & 1 & \color{IMSGreen}0 & x &   &   &   \\
  1 & 0 & 1 & \color{IMSGreen}0 & y &   &   &   \\
  1 & 0 & 0 & \color{IMSGreen}1 &   &   &   &   \\
  1 & 0 & 0 & \color{IMSGreen}1 &   &   &   & 
\end{bmatrix}\ \xrightarrow{\cref{lem:collision}} \ \begin{bmatrix}
  1 & 1 & 1 & 1 & 0 & 0 & 0 & 0 \\
  1 & 1 & 1 & 1 & 0 & 0 & 0 & 0 \\
  1 & 1 & 0 & 0 & 1 & 1 & 0 & 0 \\
  1 & 1 & 0 & 0 & 1 & 1 & 0 & 0 \\
  1 & 0 & 1 & 0 & \color{IMSGreen}1  & \color{IMSGreen}0  & \color{IMSGreen}1  & \color{IMSGreen}0  \\
  1 & 0 & 1 & 0 & \color{IMSGreen}1  & \color{IMSGreen}0  & \color{IMSGreen}1  & \color{IMSGreen}0  \\
  1 & 0 & 0 & 1 &   &   &   &   \\
  1 & 0 & 0 & 1 &   &   &   & 
\end{bmatrix} \ \xrightarrow{\cref{lem:weight}} \ \begin{bmatrix}
  1 & 1 & 1 & 1 & 0 & 0 & 0 & 0 \\
  1 & 1 & 1 & 1 & 0 & 0 & 0 & 0 \\
  1 & 1 & 0 & 0 & 1 & 1 & 0 & 0 \\
  1 & 1 & 0 & 0 & 1 & 1 & 0 & 0 \\
  1 & 0 & 1 & 0 & 1 & 0 & 1 & 0 \\
  1 & 0 & 1 & 0 & 1 & 0 & 1 & 0 \\
  1 & 0 & 0 & 1 & \color{IMSGreen}0 & \color{IMSGreen}1 & \color{IMSGreen}1 & \color{IMSGreen}0 \\
  1 & 0 & 0 & 1 & \color{IMSGreen}0 & \color{IMSGreen}1 & \color{IMSGreen}1 & \color{IMSGreen}0
\end{bmatrix}=F.
\]
\subcase
$\lVert u_0 \rVert = 4$. Let $(x,y)$ be a pair of the entries in the $i$-th column of rows $2$ and $3$ as shown below, for $1 \leq i < 8$. By \cref{lem:collision} with $u_0$, $x=y$. Hence $ u_2^\dagger = u_3^\dagger$. By \cref{lem:collision} with $u_0^\dagger$, there must be oddly many $1$'s in $\{x,z,w\}$. Up to column permutation, consider the following two cases: $x=z=w=1$ or $x=1$, $z=w=0$.
\[
U = \begin{bmatrix}
  1 & 1 & 1 & 1 & 0 & 0 & 0 & 0 \\
  1 & 1 & 1 & 1 & 0 & 0 & 0 & 0 \\
  1 & x & z & w &   &   &   &   \\
  1 & y &   &   &  &   &   &   \\
  0 &   &  &   &   &   &   &   \\
  0 &   &  &   &   &   &   &   \\
  0 &   &  &   &   &   &   &   \\
  0 &   &  &   &   &   &   &  
\end{bmatrix}.
\]
\setcounter{ssubcases}{0}
\ssubcase
$x=z=w=1$, we have

\setcounter{sssubcases}{0}
\sssubcase
$\lVert u_2^\dagger \rVert = 8$, then
\[
U = \begin{bmatrix}
  1 & 1 & 1 & 1 & 0 & 0 & 0 & 0 \\
  1 & 1 & 1 & 1 & 0 & 0 & 0 & 0 \\
  1 & 1 & 1 & 1 & 1 & 1 & 1 & 1 \\
  1 & 1 & 1 & 1 & 1 & 1 & 1 & 1 \\
  0 &   &  &   &   &   &   &   \\
  0 &   &  &   &   &   &   &   \\
  0 &   &  &   &   &   &   &   \\
  0 &   &  &   &   &   &   &  
\end{bmatrix}.
\]
By \cref{lem:collision} with $u_0^\dagger$, there must be evenly many columns among $\{u_0,u_1,u_2,u_3\}$ that have hamming weight $8$. Since $\lVert u_0 \rVert = 4$, up to column permutation, there can be two cases.
\setcounter{ssssubcases}{0}
\ssssubcase
$\lVert u_1 \rVert = \lVert u_2 \rVert = 8$ and $\lVert u_3 \rVert = 4$.
\[
U = \begin{bmatrix}
  1 & 1 & 1 & 1 & 0 & 0 & 0 & 0 \\
  1 & 1 & 1 & 1 & 0 & 0 & 0 & 0 \\
  1 & 1 & 1 & 1 & 1 & 1 & 1 & 1 \\
  1 & 1 & 1 & 1 & 1 & 1 & 1 & 1 \\
  0 & 1 & 1 & 0 &   &   &   &   \\
  0 & 1 & 1 & 0 &   &   &   &   \\
  0 & 1 & 1 & 0 &   &   &   &   \\
  0 & 1 & 1 & 0 &   &   &   &  
\end{bmatrix} \ \xrightarrow{\cref{lem:weight}} \ \begin{bmatrix}
  1 & 1 & 1 & 1 & 0 & 0 & 0 & 0 \\
  1 & 1 & 1 & 1 & 0 & 0 & 0 & 0 \\
  1 & 1 & 1 & 1 & 1 & 1 & 1 & 1 \\
  1 & 1 & 1 & 1 & 1 & 1 & 1 & 1 \\
  0 & 1 & 1 & 0 & \color{IMSGreen}1  & \color{IMSGreen}1 & \color{IMSGreen}0 & \color{IMSGreen}0  \\
  0 & 1 & 1 & 0 & \color{IMSGreen}1 &   &   &   \\
  0 & 1 & 1 & 0 & \color{IMSGreen}0 &   &   &   \\
  0 & 1 & 1 & 0 & \color{IMSGreen}0 &   &   &  
\end{bmatrix} \ \xrightarrow{\cref{lem:collision}} \ \begin{bmatrix}
  1 & 1 & 1 & 1 & 0 & 0 & 0 & 0 \\
  1 & 1 & 1 & 1 & 0 & 0 & 0 & 0 \\
  1 & 1 & 1 & 1 & 1 & 1 & 1 & 1 \\
  1 & 1 & 1 & 1 & 1 & 1 & 1 & 1 \\
  0 & 1 & 1 & 0 & 1 & 1 & 0 & 0  \\
  0 & 1 & 1 & 0 & 1 & \color{IMSGreen}1  &   &   \\
  0 & 1 & 1 & 0 & 0 &   &   &   \\
  0 & 1 & 1 & 0 & 0 &   &   &  
\end{bmatrix}
\]
\[
\xrightarrow{\cref{lem:weight}} \ \begin{bmatrix}
  1 & 1 & 1 & 1 & 0 & 0 & 0 & 0 \\
  1 & 1 & 1 & 1 & 0 & 0 & 0 & 0 \\
  1 & 1 & 1 & 1 & 1 & 1 & 1 & 1 \\
  1 & 1 & 1 & 1 & 1 & 1 & 1 & 1 \\
  0 & 1 & 1 & 0 & 1 & 1 & 0 & 0  \\
  0 & 1 & 1 & 0 & 1 & 1 & \color{IMSGreen}0 & \color{IMSGreen}0  \\
  0 & 1 & 1 & 0 & 0 & \color{IMSGreen}0 &   &   \\
  0 & 1 & 1 & 0 & 0 & \color{IMSGreen}0 &   &  
\end{bmatrix} \ \xrightarrow{\cref{lem:weight}} \ \begin{bmatrix}
  1 & 1 & 1 & 1 & 0 & 0 & 0 & 0 \\
  1 & 1 & 1 & 1 & 0 & 0 & 0 & 0 \\
  1 & 1 & 1 & 1 & 1 & 1 & 1 & 1 \\
  1 & 1 & 1 & 1 & 1 & 1 & 1 & 1 \\
  0 & 1 & 1 & 0 & 1 & 1 & 0 & 0  \\
  0 & 1 & 1 & 0 & 1 & 1 & 0 & 0  \\
  0 & 1 & 1 & 0 & 0 & 0 & \color{IMSGreen}1 & \color{IMSGreen}1 \\
  0 & 1 & 1 & 0 & 0 & 0 & \color{IMSGreen}1 & \color{IMSGreen}1
\end{bmatrix}=C.
\]
\ssssubcase
$\lVert u_1 \rVert = \lVert u_2 \rVert = \lVert u_3 \rVert = 4$.
\[
U = \begin{bmatrix}
  1 & 1 & 1 & 1 & 0 & 0 & 0 & 0 \\
  1 & 1 & 1 & 1 & 0 & 0 & 0 & 0 \\
  1 & 1 & 1 & 1 & 1 & 1 & 1 & 1 \\
  1 & 1 & 1 & 1 & 1 & 1 & 1 & 1 \\
  0 & 0 & 0 & 0 &   &   &   &   \\
  0 & 0 & 0 & 0 &   &   &   &   \\
  0 & 0 & 0 & 0 &   &   &   &   \\
  0 & 0 & 0 & 0 &   &   &   & 
\end{bmatrix} \ \xrightarrow{\cref{lem:weight}} \ \begin{bmatrix}
  1 & 1 & 1 & 1 & 0 & 0 & 0 & 0 \\
  1 & 1 & 1 & 1 & 0 & 0 & 0 & 0 \\
  1 & 1 & 1 & 1 & 1 & 1 & 1 & 1 \\
  1 & 1 & 1 & 1 & 1 & 1 & 1 & 1 \\
  0 & 0 & 0 & 0 & \color{IMSGreen}1 & \color{IMSGreen}1  &  \color{IMSGreen}1 & \color{IMSGreen}1  \\
  0 & 0 & 0 & 0 & \color{IMSGreen}1 & \color{IMSGreen}1  & \color{IMSGreen}1  &  \color{IMSGreen}1 \\
  0 & 0 & 0 & 0 & \color{IMSGreen}0 & \color{IMSGreen}0  & \color{IMSGreen}0  &  \color{IMSGreen}0 \\
  0 & 0 & 0 & 0 & \color{IMSGreen}0 &  \color{IMSGreen}0 &  \color{IMSGreen}0 & \color{IMSGreen}0
\end{bmatrix}=E.
\]
\sssubcase
$\lVert u_2^\dagger \rVert = 4$, then
\[
U = \begin{bmatrix}
  1 & 1 & 1 & 1 & 0 & 0 & 0 & 0 \\
  1 & 1 & 1 & 1 & 0 & 0 & 0 & 0 \\
  1 & 1 & 1 & 1 & 0 & 0 & 0 & 0 \\
  1 & 1 & 1 & 1 & 0 & 0 & 0 & 0 \\
  0 &   &  &   &   &   &   &   \\
  0 &   &  &   &   &   &   &   \\
  0 &   &  &   &   &   &   &   \\
  0 &   &  &   &   &   &   &  
\end{bmatrix}.
\]
By \cref{lem:collision} with $u_0^\dagger$, there must be evenly many columns among $\{u_0,u_1,u_2,u_3\}$ that have hamming weight $8$. Since $\lVert u_0 \rVert = 4$, up to column permutation, there can be two cases.

\setcounter{ssssubcases}{0}
\ssssubcase
$\lVert u_1 \rVert = \lVert u_2 \rVert = 8$ and $\lVert u_3 \rVert = 4$.
\[
U = \begin{bmatrix}
  1 & 1 & 1 & 1 & 0 & 0 & 0 & 0 \\
  1 & 1 & 1 & 1 & 0 & 0 & 0 & 0 \\
  1 & 1 & 1 & 1 & 0 & 0 & 0 & 0 \\
  1 & 1 & 1 & 1 & 0 & 0 & 0 & 0 \\
  0 & 1 & 1 & 0 &   &   &   &   \\
  0 & 1 & 1 & 0 &   &   &   &   \\
  0 & 1 & 1 & 0 &   &   &   &   \\
  0 & 1 & 1 & 0 &   &   &   &  
\end{bmatrix} \ \xrightarrow{\cref{lem:weight}} \ \begin{bmatrix}
  1 & 1 & 1 & 1 & 0 & 0 & 0 & 0 \\
  1 & 1 & 1 & 1 & 0 & 0 & 0 & 0 \\
  1 & 1 & 1 & 1 & 0 & 0 & 0 & 0 \\
  1 & 1 & 1 & 1 & 0 & 0 & 0 & 0 \\
  0 & 1 & 1 & 0 & \color{IMSGreen}1 & \color{IMSGreen}1 & \color{IMSGreen}0 & \color{IMSGreen}0 \\
  0 & 1 & 1 & 0 &   &   &   &   \\
  0 & 1 & 1 & 0 &   &   &   &   \\
  0 & 1 & 1 & 0 &   &   &   &  
\end{bmatrix} \ \xrightarrow{\cref{lem:weight}} \ \begin{bmatrix}
  1 & 1 & 1 & 1 & 0 & 0 & 0 & 0 \\
  1 & 1 & 1 & 1 & 0 & 0 & 0 & 0 \\
  1 & 1 & 1 & 1 & 0 & 0 & 0 & 0 \\
  1 & 1 & 1 & 1 & 0 & 0 & 0 & 0 \\
  0 & 1 & 1 & 0 & 1 & 1 & 0 & 0 \\
  0 & 1 & 1 & 0 & \color{IMSGreen}1 & \color{IMSGreen}1 & \color{IMSGreen}0 & \color{IMSGreen}0 \\
  0 & 1 & 1 & 0 & \color{IMSGreen}1 & \color{IMSGreen}1 & \color{IMSGreen}0 & \color{IMSGreen}0 \\
  0 & 1 & 1 & 0 & \color{IMSGreen}1 & \color{IMSGreen}1 & \color{IMSGreen}0 & \color{IMSGreen}0
\end{bmatrix}=D.
\]
\ssssubcase
$\lVert u_1 \rVert = \lVert u_2 \rVert = \lVert u_3 \rVert = 4$. Depending on the hamming weight of $u_4^\dagger$, consider what follows.
\setcounter{sssssubcases}{0}
\sssssubcase
$\lVert u_4^\dagger \rVert = 4$, then
\[
U = \begin{bmatrix}
  1 & 1 & 1 & 1 & 0 & 0 & 0 & 0 \\
  1 & 1 & 1 & 1 & 0 & 0 & 0 & 0 \\
  1 & 1 & 1 & 1 & 0 & 0 & 0 & 0 \\
  1 & 1 & 1 & 1 & 0 & 0 & 0 & 0 \\
  0 & 0 & 0 & 0 & 1 & 1 & 1 & 1 \\
  0 & 0 & 0 & 0 &   &   &   &   \\
  0 & 0 & 0 & 0 &   &   &   &   \\
  0 & 0 & 0 & 0 &   &   &   &   
\end{bmatrix} \ \xrightarrow{\cref{lem:weight}} \ \begin{bmatrix}
  1 & 1 & 1 & 1 & 0 & 0 & 0 & 0 \\
  1 & 1 & 1 & 1 & 0 & 0 & 0 & 0 \\
  1 & 1 & 1 & 1 & 0 & 0 & 0 & 0 \\
  1 & 1 & 1 & 1 & 0 & 0 & 0 & 0 \\
  0 & 0 & 0 & 0 & 1 & 1 & 1 & 1 \\
  0 & 0 & 0 & 0 & \color{IMSGreen}1 & \color{IMSGreen}1 & \color{IMSGreen}1 & \color{IMSGreen}1 \\
  0 & 0 & 0 & 0 & \color{IMSGreen}1 & \color{IMSGreen}1 & \color{IMSGreen}1 & \color{IMSGreen}1 \\
  0 & 0 & 0 & 0 & \color{IMSGreen}1 & \color{IMSGreen}1 & \color{IMSGreen}1 & \color{IMSGreen}1
\end{bmatrix}=I.
\]

\sssssubcase
$\lVert u_4^\dagger \rVert = 0$, then
\[
U = \begin{bmatrix}
  1 & 1 & 1 & 1 & 0 & 0 & 0 & 0 \\
  1 & 1 & 1 & 1 & 0 & 0 & 0 & 0 \\
  1 & 1 & 1 & 1 & 0 & 0 & 0 & 0 \\
  1 & 1 & 1 & 1 & 0 & 0 & 0 & 0 \\
  0 & 0 & 0 & 0 & 0 & 0 & 0 & 0 \\
  0 & 0 & 0 & 0 &   &   &   &   \\
  0 & 0 & 0 & 0 &   &   &   &   \\
  0 & 0 & 0 & 0 &   &   &   &   
\end{bmatrix} \ \xrightarrow{\cref{lem:weight}} \ \begin{bmatrix}
  1 & 1 & 1 & 1 & 0 & 0 & 0 & 0 \\
  1 & 1 & 1 & 1 & 0 & 0 & 0 & 0 \\
  1 & 1 & 1 & 1 & 0 & 0 & 0 & 0 \\
  1 & 1 & 1 & 1 & 0 & 0 & 0 & 0 \\
  0 & 0 & 0 & 0 & 0 & 0 & 0 & 0 \\
  0 & 0 & 0 & 0 & \color{IMSGreen}0 & \color{IMSGreen}0 & \color{IMSGreen}0 & \color{IMSGreen}0 \\
  0 & 0 & 0 & 0 & \color{IMSGreen}0 & \color{IMSGreen}0 & \color{IMSGreen}0 & \color{IMSGreen}0 \\
  0 & 0 & 0 & 0 & \color{IMSGreen}0 & \color{IMSGreen}0 & \color{IMSGreen}0 & \color{IMSGreen}0
\end{bmatrix}=J.
\]
\ssubcase
$x=1$, $z=w=0$, we have
\[
U = \begin{bmatrix}
  1 & 1 & 1 & 1 & 0 & 0 & 0 & 0 \\
  1 & 1 & 1 & 1 & 0 & 0 & 0 & 0 \\
  1 & 1 & 0 & 0 &   &   &   &   \\
  1 & 1 & 0 & 0 &  &   &   &   \\
  0 &   &  &   &   &   &   &   \\
  0 &   &  &   &   &   &   &   \\
  0 &   &  &   &   &   &   &   \\
  0 &   &  &   &   &   &   &  
\end{bmatrix} \ \xrightarrow{\cref{lem:weight}} \ \begin{bmatrix}
  1 & 1 & 1 & 1 & 0 & 0 & 0 & 0 \\
  1 & 1 & 1 & 1 & 0 & 0 & 0 & 0 \\
  1 & 1 & 0 & 0 & \color{IMSGreen}1 & \color{IMSGreen}1 & \color{IMSGreen}0 & \color{IMSGreen}0 \\
  1 & 1 & 0 & 0 & \color{IMSGreen}1 & \color{IMSGreen}1 & \color{IMSGreen}0 & \color{IMSGreen}0 \\
  0 &   &  &   &   &   &   &   \\
  0 &   &  &   &   &   &   &   \\
  0 &   &  &   &   &   &   &   \\
  0 &   &  &   &   &   &   &  
\end{bmatrix}.
\]
\setcounter{sssubcases}{0}
\sssubcase
$\lVert u_1 \rVert = 8$, then we have what follows.
\[
U=\begin{bmatrix}
  1 & 1 & 1 & 1 & 0 & 0 & 0 & 0 \\
  1 & 1 & 1 & 1 & 0 & 0 & 0 & 0 \\
  1 & 1 & 0 & 0 & 1 & 1 & 0 & 0 \\
  1 & 1 & 0 & 0 & 1 & 1 & 0 & 0 \\
  0 & 1 &  &   &   &   &   &   \\
  0 & 1 &  &   &   &   &   &   \\
  0 & 1 &  &   &   &   &   &   \\
  0 & 1 &  &   &   &   &   &  
\end{bmatrix} \ \xrightarrow{\cref{lem:weight}} \ \begin{bmatrix}
  1 & 1 & 1 & 1 & 0 & 0 & 0 & 0 \\
  1 & 1 & 1 & 1 & 0 & 0 & 0 & 0 \\
  1 & 1 & 0 & 0 & 1 & 1 & 0 & 0 \\
  1 & 1 & 0 & 0 & 1 & 1 & 0 & 0 \\
  0 & 1 & \color{IMSGreen}1 &   &   &   &   &   \\
  0 & 1 & \color{IMSGreen}1 &   &   &   &   &   \\
  0 & 1 & \color{IMSGreen}0 &   &   &   &   &   \\
  0 & 1 & \color{IMSGreen}0 &   &   &   &   &  
\end{bmatrix} \ \xrightarrow{\cref{lem:collision}} \ \begin{bmatrix}
  1 & 1 & 1 & 1 & 0 & 0 & 0 & 0 \\
  1 & 1 & 1 & 1 & 0 & 0 & 0 & 0 \\
  1 & 1 & 0 & 0 & 1 & 1 & 0 & 0 \\
  1 & 1 & 0 & 0 & 1 & 1 & 0 & 0 \\
  0 & 1 & 1 & \color{IMSGreen}0  & x  &   &   &   \\
  0 & 1 & 1 & \color{IMSGreen}0 &  y &   &   &   \\
  0 & 1 & 0 & \color{IMSGreen}1 &   &   &   &   \\
  0 & 1 & 0 & \color{IMSGreen}1 &   &   &   &  
\end{bmatrix}.
\]
Let $(x,y)$ be a pair of the entries in the $i$-th column of rows $4$ and $5$ as shown above, for $4 \leq i < 8$. By \cref{lem:collision} with $u_2$, $x=y$. Hence $ u_4^\dagger = u_5 ^\dagger$.
\[
\xrightarrow{\cref{lem:collision}} \ \begin{bmatrix}
  1 & 1 & 1 & 1 & 0 & 0 & 0 & 0 \\
  1 & 1 & 1 & 1 & 0 & 0 & 0 & 0 \\
  1 & 1 & 0 & 0 & 1 & 1 & 0 & 0 \\
  1 & 1 & 0 & 0 & 1 & 1 & 0 & 0 \\
  0 & 1 & 1 & 0 & \color{IMSGreen}1 & \color{IMSGreen}0 & \color{IMSGreen}1 & \color{IMSGreen}0 \\
  0 & 1 & 1 & 0 & \color{IMSGreen}1 & \color{IMSGreen}0 & \color{IMSGreen}1 & \color{IMSGreen}0 \\
  0 & 1 & 0 & 1 &   &   &   &   \\
  0 & 1 & 0 & 1 &   &   &   &  
\end{bmatrix} \ \xrightarrow{\cref{lem:weight}} \ \begin{bmatrix}
  1 & 1 & 1 & 1 & 0 & 0 & 0 & 0 \\
  1 & 1 & 1 & 1 & 0 & 0 & 0 & 0 \\
  1 & 1 & 0 & 0 & 1 & 1 & 0 & 0 \\
  1 & 1 & 0 & 0 & 1 & 1 & 0 & 0 \\
  0 & 1 & 1 & 0 & 1 & 0 & 1 & 0 \\
  0 & 1 & 1 & 0 & 1 & 0 & 1 & 0 \\
  0 & 1 & 0 & 1 & \color{IMSGreen}0 & \color{IMSGreen}1 & \color{IMSGreen}1 & \color{IMSGreen}0 \\
  0 & 1 & 0 & 1 & \color{IMSGreen}0 & \color{IMSGreen}1  & \color{IMSGreen}1 & \color{IMSGreen}0
\end{bmatrix}=F.
\]
\sssubcase
$\lVert u_1 \rVert = 4$, then
\[
U=\begin{bmatrix}
  1 & 1 & 1 & 1 & 0 & 0 & 0 & 0 \\
  1 & 1 & 1 & 1 & 0 & 0 & 0 & 0 \\
  1 & 1 & 0 & 0 & 1 & 1 & 0 & 0 \\
  1 & 1 & 0 & 0 & 1 & 1 & 0 & 0 \\
  0 & 0 &  &   &   &   &   &   \\
  0 & 0 &  &   &   &   &   &   \\
  0 & 0 &  &   &   &   &   &   \\
  0 & 0 &  &   &   &   &   &  
\end{bmatrix} \ \xrightarrow{\cref{lem:weight}} \ \begin{bmatrix}
  1 & 1 & 1 & 1 & 0 & 0 & 0 & 0 \\
  1 & 1 & 1 & 1 & 0 & 0 & 0 & 0 \\
  1 & 1 & 0 & 0 & 1 & 1 & 0 & 0 \\
  1 & 1 & 0 & 0 & 1 & 1 & 0 & 0 \\
  0 & 0 & \color{IMSGreen}1 & x &   &   &   &   \\
  0 & 0 & \color{IMSGreen}1 & y &   &   &   &   \\
  0 & 0 & \color{IMSGreen}0 &   &   &   &   &   \\
  0 & 0 & \color{IMSGreen}0 &   &   &   &   &  
\end{bmatrix}.
\]
Let $(x,y)$ be a pair of the entries in the $i$-th column of rows $4$ and $5$ as shown above, for $3 \leq i < 8$. By \cref{lem:collision} with $u_2$, $x=y$. Hence $ u_4^\dagger = u_5^\dagger$. Moreover, by \cref{lem:collision} with $u_0^\dagger$, we have
\[ 
U=\begin{bmatrix}
  1 & 1 & 1 & 1 & 0 & 0 & 0 & 0 \\
  1 & 1 & 1 & 1 & 0 & 0 & 0 & 0 \\
  1 & 1 & 0 & 0 & 1 & 1 & 0 & 0 \\
  1 & 1 & 0 & 0 & 1 & 1 & 0 & 0 \\
  0 & 0 & 1 & 1 &   &   &   &   \\
  0 & 0 & 1 & 1 &   &   &   &   \\
  0 & 0 & 0 &   &   &   &   &   \\
  0 & 0 & 0 &   &   &   &   &  
\end{bmatrix} \ \xrightarrow{\cref{lem:weight}} \ \begin{bmatrix}
  1 & 1 & 1 & 1 & 0 & 0 & 0 & 0 \\
  1 & 1 & 1 & 1 & 0 & 0 & 0 & 0 \\
  1 & 1 & 0 & 0 & 1 & 1 & 0 & 0 \\
  1 & 1 & 0 & 0 & 1 & 1 & 0 & 0 \\
  0 & 0 & 1 & 1 & x & z &   &   \\
  0 & 0 & 1 & 1 & y & w &   &   \\
  0 & 0 & 0 & \color{IMSGreen}0 &   &   &   &   \\
  0 & 0 & 0 & \color{IMSGreen}0 &   &   &   &  
\end{bmatrix}. 
\]
By \cref{lem:collision} with $u_2^\dagger$, $x=z$ and $y=w$. Since $x=y$ and $z=w$, $x=y=z=w$.
\setcounter{ssssubcases}{0}
\ssssubcase
$x=y=z=w=1$
\[
U=\begin{bmatrix}
  1 & 1 & 1 & 1 & 0 & 0 & 0 & 0 \\
  1 & 1 & 1 & 1 & 0 & 0 & 0 & 0 \\
  1 & 1 & 0 & 0 & 1 & 1 & 0 & 0 \\
  1 & 1 & 0 & 0 & 1 & 1 & 0 & 0 \\
  0 & 0 & 1 & 1 & 1 & 1 &   &   \\
  0 & 0 & 1 & 1 & 1 & 1 &   &   \\
  0 & 0 & 0 & 0 &   &   &   &   \\
  0 & 0 & 0 & 0 &   &   &   &  
\end{bmatrix} \ \xrightarrow{\cref{lem:weight}} \ \begin{bmatrix}
  1 & 1 & 1 & 1 & 0 & 0 & 0 & 0 \\
  1 & 1 & 1 & 1 & 0 & 0 & 0 & 0 \\
  1 & 1 & 0 & 0 & 1 & 1 & 0 & 0 \\
  1 & 1 & 0 & 0 & 1 & 1 & 0 & 0 \\
  0 & 0 & 1 & 1 & 1 & 1 & \color{IMSGreen}0  & \color{IMSGreen}0 \\
  0 & 0 & 1 & 1 & 1 & 1 &  \color{IMSGreen}0 &  \color{IMSGreen}0 \\
  0 & 0 & 0 & 0 & \color{IMSGreen}0 & \color{IMSGreen}0 &   &   \\
  0 & 0 & 0 & 0 & \color{IMSGreen}0 & \color{IMSGreen}0 &   &  
\end{bmatrix} \ \xrightarrow{\cref{lem:weight}} \ \begin{bmatrix}
  1 & 1 & 1 & 1 & 0 & 0 & 0 & 0 \\
  1 & 1 & 1 & 1 & 0 & 0 & 0 & 0 \\
  1 & 1 & 0 & 0 & 1 & 1 & 0 & 0 \\
  1 & 1 & 0 & 0 & 1 & 1 & 0 & 0 \\
  0 & 0 & 1 & 1 & 1 & 1 & 0 & 0 \\
  0 & 0 & 1 & 1 & 1 & 1 & 0 & 0 \\
  0 & 0 & 0 & 0 & 0 & 0 & \color{IMSGreen}0 & \color{IMSGreen}0 \\
  0 & 0 & 0 & 0 & 0 & 0 & \color{IMSGreen}0 & \color{IMSGreen}0
\end{bmatrix}=G.
\]
\ssssubcase
$x=y=z=w=0$.
\[
U=\begin{bmatrix}
  1 & 1 & 1 & 1 & 0 & 0 & 0 & 0 \\
  1 & 1 & 1 & 1 & 0 & 0 & 0 & 0 \\
  1 & 1 & 0 & 0 & 1 & 1 & 0 & 0 \\
  1 & 1 & 0 & 0 & 1 & 1 & 0 & 0 \\
  0 & 0 & 1 & 1 & 0 & 0 &   &   \\
  0 & 0 & 1 & 1 & 0 & 0 &   &   \\
  0 & 0 & 0 & 0 &   &   &   &   \\
  0 & 0 & 0 & 0 &   &   &   &  
\end{bmatrix} \ \xrightarrow{\cref{lem:weight}} \ \begin{bmatrix}
  1 & 1 & 1 & 1 & 0 & 0 & 0 & 0 \\
  1 & 1 & 1 & 1 & 0 & 0 & 0 & 0 \\
  1 & 1 & 0 & 0 & 1 & 1 & 0 & 0 \\
  1 & 1 & 0 & 0 & 1 & 1 & 0 & 0 \\
  0 & 0 & 1 & 1 & 0 & 0 & \color{IMSGreen}1 & \color{IMSGreen}1 \\
  0 & 0 & 1 & 1 & 0 & 0 & \color{IMSGreen}1 & \color{IMSGreen}1 \\
  0 & 0 & 0 & 0 & \color{IMSGreen}1 & \color{IMSGreen}1 &   &   \\
  0 & 0 & 0 & 0 & \color{IMSGreen}1 & \color{IMSGreen}1 &   &  
\end{bmatrix} \ \xrightarrow{\cref{lem:weight}} \ \begin{bmatrix}
  1 & 1 & 1 & 1 & 0 & 0 & 0 & 0 \\
  1 & 1 & 1 & 1 & 0 & 0 & 0 & 0 \\
  1 & 1 & 0 & 0 & 1 & 1 & 0 & 0 \\
  1 & 1 & 0 & 0 & 1 & 1 & 0 & 0 \\
  0 & 0 & 1 & 1 & 0 & 0 & 1 & 1 \\
  0 & 0 & 1 & 1 & 0 & 0 & 1 & 1 \\
  0 & 0 & 0 & 0 & 1 & 1 & \color{IMSGreen}1 & \color{IMSGreen}1 \\
  0 & 0 & 0 & 0 & 1 & 1 & \color{IMSGreen}1 & \color{IMSGreen}1
\end{bmatrix}=H.
\]
\end{mycases}

\end{proof}

\subsection{Binary Patterns that are neither Row-paired nor Column-paired}

\begin{definition}
We define the set $\mathcal{B}_1$ of binary matrices as $\mathcal{B}_1
= \{L,M,N\}$, where
{\scriptsize
\[
L=\begin{bmatrix}
    \ 1 \ & \ 1 \ & \ 1 \ & \ 1 \ & \ 1 \ & \ 1 \ & \ 1 \ & \ 1 \ \\
    1 & 1 & 1 & 1 & 0 & 0 & 0 & 0\\
    1 & 1 & 0 & 0 & 1 & 1 & 0 & 0\\
    1 & 1 & 0 & 0 & 0 & 0 & 1 & 1\\
    1 & 0 & 1 & 0 & 1 & 0 & 1 & 0\\
    1 & 0 & 1 & 0 & 0 & 1 & 0 & 1\\
    1 & 0 & 0 & 1 & 1 & 0 & 0 & 1\\
    1 & 0 & 0 & 1 & 0 & 1 & 1 & 0
  \end{bmatrix},\quad M=\begin{bmatrix}
    \ 1 \ & \ 1 \ & \ 1 \ & \ 1 \ & \ 0 \ & \ 0 \ & \ 0 \ & \ 0\ \\
    1 & 1 & 0 & 0 & 1 & 1 & 0 & 0\\
    1 & 0 & 1 & 0 & 1 & 0 & 1 & 0\\
    1 & 0 & 0 & 1 & 0 & 1 & 1 & 0\\
    0 & 1 & 1 & 0 & 1 & 0 & 0 & 1\\
    0 & 1 & 0 & 1 & 0 & 1 & 0 & 1\\
    0 & 0 & 1 & 1 & 0 & 0 & 1 & 1\\
    0 & 0 & 0 & 0 & 1 & 1 & 1 & 1\\
  \end{bmatrix}, \quad N=\begin{bmatrix}
    \ 1 \ & \ 1 \ & \ 1 \ & \ 1 \ & \ 0 \ & \ 0 \ & \ 0 \ & \ 0\ \\
    1 & 1 & 0 & 0 & 1 & 1 & 0 & 0\\
    1 & 0 & 1 & 0 & 1 & 0 & 1 & 0\\
    1 & 0 & 0 & 1 & 0 & 1 & 1 & 0\\
    0 & 1 & 1 & 0 & 0 & 1 & 1 & 0\\
    0 & 1 & 0 & 1 & 1 & 0 & 1 & 0\\
    0 & 0 & 1 & 1 & 1 & 1 & 0 & 0\\
    0 & 0 & 0 & 0 & 0 & 0 & 0 & 0\\
  \end{bmatrix}.
\]  
}
\end{definition}

\begin{proposition}
\label{prop:notnice}
Let $U \in \Z_2^{8 \times 8}$. Suppose $U$ satisfies \cref{lem:collision,lem:weight}. If there are no identical rows nor columns in $U$, then $U \in \mathcal{B}_1$ up to permutation and transposition.
\end{proposition}

\begin{proof}
Let $u_i$ denote the $i$-th column of $U$, and $u_i^\dagger$ denote the the $i$-th row of $U$, $0 \leq i < 8$. Let $\lVert v \rVert$ denote the hamming weight of $v$, where $v$ is a string of binary bits. Since there are no identical rows nor columns in $U$, we proceed by case distinction on the maximum hamming weight of a row vector in $U$.
\begin{mycases}
\case
There is a row with hamming weight $8$. Up to row permutation, $\lVert u_0^\dagger \rVert = 8$.
\subcase
There is a row with hamming weight $0$. Up to row permutation, $\lVert u_1^\dagger \rVert= 0$.
\[
U=\begin{bmatrix}
    1 & 1 & 1 & 1 & 1 & 1 & 1 & 1\\
    0 & 0 & 0 & 0 & 0 & 0 & 0 & 0\\
      &  &  &  &  &  &  & \\
      &  &  &  &  &  &  & \\
      &  &  &  &  &  &  & \\
      &  &  &  &  &  &  & \\
      &  &  &  &  &  &  & \\
      &  &  &  &  &  &  & 
  \end{bmatrix}  \ \xrightarrow{\text{Case 1}} \ \begin{bmatrix}
    1 & 1 & 1 & 1 & 1 & 1 & 1 & 1\\
    0 & 0 & 0 & 0 & 0 & 0 & 0 & 0\\
    \color{IMSGreen}1 & \color{IMSGreen}1 & \color{IMSGreen}1 & \color{IMSGreen}1 & \color{IMSGreen}0 & \color{IMSGreen}0 & \color{IMSGreen}0 & \color{IMSGreen}0\\
      &  &  &  &  &  &  & \\
      &  &  &  &  &  &  & \\
      &  &  &  &  &  &  & \\
      &  &  &  &  &  &  & \\
      &  &  &  &  &  &  & 
  \end{bmatrix} \ \xrightarrow{\cref{lem:weight}} \ \begin{bmatrix}
    1 & 1 & 1 & 1 & 1 & 1 & 1 & 1\\
    0 & 0 & 0 & 0 & 0 & 0 & 0 & 0\\
    1 & 1 & 1 & 1 & 0 & 0 & 0 & 0\\
    \color{IMSGreen}1 &  &  &  &  &  &  & \\
    \color{IMSGreen}1 &  &  &  &  &  &  & \\
    \color{IMSGreen}0 &  &  &  &  &  &  & \\
    \color{IMSGreen}0 &  &  &  &  &  &  & \\
    \color{IMSGreen}0 &  &  &  &  &  &  & 
  \end{bmatrix}
\]
\[
\xrightarrow{\cref{lem:collision}} \ \begin{bmatrix}
    1 & 1 & 1 & 1 & 1 & 1 & 1 & 1\\
    0 & 0 & 0 & 0 & 0 & 0 & 0 & 0\\
    1 & 1 & 1 & 1 & 0 & 0 & 0 & 0\\
    1 & \color{IMSGreen}1 & \color{IMSGreen}0 & \color{IMSGreen}0 & \color{IMSGreen}1 & \color{IMSGreen}1 & \color{IMSGreen}0 & \color{IMSGreen}0\\
    1 &  &  &  &  &  &  & \\
    0 &  &  &  &  &  &  & \\
    0 &  &  &  &  &  &  & \\
    0 &  &  &  &  &  &  & 
  \end{bmatrix} \ \xrightarrow{\cref{lem:collision}} \ \begin{bmatrix}
    1 & 1 & 1 & 1 & 1 & 1 & 1 & 1\\
    0 & 0 & 0 & 0 & 0 & 0 & 0 & 0\\
    1 & 1 & 1 & 1 & 0 & 0 & 0 & 0\\
    1 & 1 & 0 & 0 & 1 & 1 & 0 & 0\\
    1 & \color{IMSGreen}1 &  &  &  &  &  & \\
    0 &  &  &  &  &  &  & \\
    0 &  &  &  &  &  &  & \\
    0 &  &  &  &  &  &  & 
  \end{bmatrix}\ \xrightarrow{\cref{lem:weight}} \ \begin{bmatrix}
    1 & 1 & 1 & 1 & 1 & 1 & 1 & 1\\
    0 & 0 & 0 & 0 & 0 & 0 & 0 & 0\\
    1 & 1 & 1 & 1 & 0 & 0 & 0 & 0\\
    1 & 1 & 0 & 0 & 1 & 1 & 0 & 0\\
    1 & 1 &  &  &  &  &  & \\
    0 & \color{IMSGreen}0 &  &  &  &  &  & \\
    0 & \color{IMSGreen}0 &  &  &  &  &  & \\
    0 & \color{IMSGreen}0 &  &  &  &  &  & 
  \end{bmatrix}.
\]
Note that $u_0 = u_1$, but it contradicts our assumption that there are no identical columns in $U$. Thus this case is not possible.
\subcase
There is no row with hamming weight $0$. Up to row and column permutation, consider
\[
U=\begin{bmatrix}
    1 & 1 & 1 & 1 & 1 & 1 & 1 & 1\\
    1 & 1 & 1 & 1 & 0 & 0 & 0 & 0\\
      &  &  &  &  &  &  & \\
      &  &  &  &  &  &  & \\
      &  &  &  &  &  &  & \\
      &  &  &  &  &  &  & \\
      &  &  &  &  &  &  & \\
      &  &  &  &  &  &  & 
  \end{bmatrix}.
  \]
  \ssubcase
  There is a column with hamming weight $8$, consider
  \[
U=\begin{bmatrix}
    1 & 1 & 1 & 1 & 1 & 1 & 1 & 1\\
    1 & 1 & 1 & 1 & 0 & 0 & 0 & 0\\
    1 &  &  &  &  &  &  & \\
    1 &  &  &  &  &  &  & \\
    1 &  &  &  &  &  &  & \\
    1 &  &  &  &  &  &  & \\
    1 &  &  &  &  &  &  & \\
    1 &  &  &  &  &  &  & 
  \end{bmatrix} \ \xrightarrow{\cref{lem:weight}} \ \begin{bmatrix}
    1 & 1 & 1 & 1 & 1 & 1 & 1 & 1\\
    1 & 1 & 1 & 1 & 0 & 0 & 0 & 0\\
    1 & \color{IMSGreen}1 & \color{IMSGreen}0 & \color{IMSGreen}0 & \color{IMSGreen}1 & \color{IMSGreen}1 & \color{IMSGreen}0 & \color{IMSGreen}0 \\
    1 & \color{IMSGreen}1 &  &  &  &  &  & \\
    1 & \color{IMSGreen}0 &  &  &  &  &  & \\
    1 & \color{IMSGreen}0 &  &  &  &  &  & \\
    1 & \color{IMSGreen}0 &  &  &  &  &  & \\
    1 & \color{IMSGreen}0 &  &  &  &  &  & 
  \end{bmatrix} \ \xrightarrow{\cref{lem:collision}} \ \begin{bmatrix}
    1 & 1 & 1 & 1 & 1 & 1 & 1 & 1\\
    1 & 1 & 1 & 1 & 0 & 0 & 0 & 0\\
    1 & 1 & 0 & 0 & 1 & 1 & 0 & 0 \\
    1 & 1 & \color{IMSGreen}0 & \color{IMSGreen}0 & \color{IMSGreen}0 & \color{IMSGreen}0 & \color{IMSGreen}1 & \color{IMSGreen}1 \\
    1 & 0 &  &  &  &  &  & \\
    1 & 0 &  &  &  &  &  & \\
    1 & 0 &  &  &  &  &  & \\
    1 & 0 &  &  &  &  &  & 
  \end{bmatrix}
  \]
\[
\xrightarrow{\cref{lem:weight}} \ \begin{bmatrix}
    1 & 1 & 1 & 1 & 1 & 1 & 1 & 1\\
    1 & 1 & 1 & 1 & 0 & 0 & 0 & 0\\
    1 & 1 & 0 & 0 & 1 & 1 & 0 & 0 \\
    1 & 1 & 0 & 0 & 0 & 0 & 1 & 1 \\
    1 & 0 & \color{IMSGreen}1 &  &  &  &  & \\
    1 & 0 & \color{IMSGreen}1 &  &  &  &  & \\
    1 & 0 & \color{IMSGreen}0 &  &  &  &  & \\
    1 & 0 & \color{IMSGreen}0 &  &  &  &  & 
  \end{bmatrix} \ \xrightarrow{\cref{lem:collision}} \ \begin{bmatrix}
    1 & 1 & 1 & 1 & 1 & 1 & 1 & 1\\
    1 & 1 & 1 & 1 & 0 & 0 & 0 & 0\\
    1 & 1 & 0 & 0 & 1 & 1 & 0 & 0 \\
    1 & 1 & 0 & 0 & 0 & 0 & 1 & 1 \\
    1 & 0 & 1 & \color{IMSGreen}0 &  &  &  & \\
    1 & 0 & 1 & \color{IMSGreen}0 &  &  &  & \\
    1 & 0 & 0 & \color{IMSGreen}1 &  &  &  & \\
    1 & 0 & 0 & \color{IMSGreen}1 &  &  &  & 
  \end{bmatrix} \ \xrightarrow{\cref{lem:weight}} \ \begin{bmatrix}
    1 & 1 & 1 & 1 & 1 & 1 & 1 & 1\\
    1 & 1 & 1 & 1 & 0 & 0 & 0 & 0\\
    1 & 1 & 0 & 0 & 1 & 1 & 0 & 0 \\
    1 & 1 & 0 & 0 & 0 & 0 & 1 & 1 \\
    1 & 0 & 1 & 0 & \color{IMSGreen}1 & \color{IMSGreen}0 & \color{IMSGreen}1 & \color{IMSGreen}0 \\
    1 & 0 & 1 & 0 &  &  &  & \\
    1 & 0 & 0 & 1 &  &  &  & \\
    1 & 0 & 0 & 1 &  &  &  & 
  \end{bmatrix}
\]
\[
\xrightarrow{\cref{lem:collision}} \ \begin{bmatrix}
    1 & 1 & 1 & 1 & 1 & 1 & 1 & 1\\
    1 & 1 & 1 & 1 & 0 & 0 & 0 & 0\\
    1 & 1 & 0 & 0 & 1 & 1 & 0 & 0 \\
    1 & 1 & 0 & 0 & 0 & 0 & 1 & 1 \\
    1 & 0 & 1 & 0 & 1 & 0 & 1 & 0 \\
    1 & 0 & 1 & 0 & \color{IMSGreen}0 &  &  & \\
    1 & 0 & 0 & 1 & \color{IMSGreen}1 &  &  & \\
    1 & 0 & 0 & 1 & \color{IMSGreen}0 &  &  &
  \end{bmatrix} \ \xrightarrow{\cref{lem:collision}} \ \begin{bmatrix}
    1 & 1 & 1 & 1 & 1 & 1 & 1 & 1\\
    1 & 1 & 1 & 1 & 0 & 0 & 0 & 0\\
    1 & 1 & 0 & 0 & 1 & 1 & 0 & 0 \\
    1 & 1 & 0 & 0 & 0 & 0 & 1 & 1 \\
    1 & 0 & 1 & 0 & 1 & 0 & 1 & 0 \\
    1 & 0 & 1 & 0 & 0 & \color{IMSGreen}1 &  & \\
    1 & 0 & 0 & 1 & 1 & \color{IMSGreen}0 &  & \\
    1 & 0 & 0 & 1 & 0 & \color{IMSGreen}1 &  &
  \end{bmatrix} \ \xrightarrow{\cref{lem:collision}} \ \begin{bmatrix}
    1 & 1 & 1 & 1 & 1 & 1 & 1 & 1\\
    1 & 1 & 1 & 1 & 0 & 0 & 0 & 0\\
    1 & 1 & 0 & 0 & 1 & 1 & 0 & 0 \\
    1 & 1 & 0 & 0 & 0 & 0 & 1 & 1 \\
    1 & 0 & 1 & 0 & 1 & 0 & 1 & 0 \\
    1 & 0 & 1 & 0 & 0 & 1 & \color{IMSGreen}0 & \color{IMSGreen}1\\
    1 & 0 & 0 & 1 & 1 & 0 &  & \\
    1 & 0 & 0 & 1 & 0 & 1 &  &
  \end{bmatrix}
\]
\[
\xrightarrow{\cref{lem:collision}} \ \begin{bmatrix}
    1 & 1 & 1 & 1 & 1 & 1 & 1 & 1\\
    1 & 1 & 1 & 1 & 0 & 0 & 0 & 0\\
    1 & 1 & 0 & 0 & 1 & 1 & 0 & 0 \\
    1 & 1 & 0 & 0 & 0 & 0 & 1 & 1 \\
    1 & 0 & 1 & 0 & 1 & 0 & 1 & 0 \\
    1 & 0 & 1 & 0 & 0 & 1 & 0 & 1\\
    1 & 0 & 0 & 1 & 1 & 0 & \color{IMSGreen}0 & \\
    1 & 0 & 0 & 1 & 0 & 1 & \color{IMSGreen}1 &
  \end{bmatrix} \ \xrightarrow{\cref{lem:weight}} \ \begin{bmatrix}
    1 & 1 & 1 & 1 & 1 & 1 & 1 & 1\\
    1 & 1 & 1 & 1 & 0 & 0 & 0 & 0\\
    1 & 1 & 0 & 0 & 1 & 1 & 0 & 0 \\
    1 & 1 & 0 & 0 & 0 & 0 & 1 & 1 \\
    1 & 0 & 1 & 0 & 1 & 0 & 1 & 0 \\
    1 & 0 & 1 & 0 & 0 & 1 & 0 & 1\\
    1 & 0 & 0 & 1 & 1 & 0 & 0 & \color{IMSGreen}1\\
    1 & 0 & 0 & 1 & 0 & 1 & 1 & \color{IMSGreen}0
  \end{bmatrix}=L.
\]
\ssubcase
There is no column whose hamming weight is $8$. Up to row and column permutation, consider
\[
U=\begin{bmatrix}
    1 & 1 & 1 & 1 & 1 & 1 & 1 & 1\\
    1 & 1 & 1 & 1 & 0 & 0 & 0 & 0\\
    1 &  &  &  &  &  &  & \\
    1 &  &  &  &  &  &  & \\
    0 &  &  &  &  &  &  & \\
    0 &  &  &  &  &  &  & \\
    0 &  &  &  &  &  &  & \\
    0 &  &  &  &  &  &  & 
  \end{bmatrix} \ \xrightarrow[\cref{lem:collision}]{\cref{lem:weight}} \ \begin{bmatrix}
    1 & 1 & 1 & 1 & 1 & 1 & 1 & 1\\
    1 & 1 & 1 & 1 & 0 & 0 & 0 & 0\\
    1 & \color{IMSGreen}1 & \color{IMSGreen}0 & \color{IMSGreen}0 & \color{IMSGreen}1 & \color{IMSGreen}1 & \color{IMSGreen}0 & \color{IMSGreen}0 \\
    1 &  &  &  &  &  &  & \\
    0 &  &  &  &  &  &  & \\
    0 &  &  &  &  &  &  & \\
    0 &  &  &  &  &  &  & \\
    0 &  &  &  &  &  &  & 
  \end{bmatrix} \ \xrightarrow{\cref{lem:collision}} \ \begin{bmatrix}
    1 & 1 & 1 & 1 & 1 & 1 & 1 & 1\\
    1 & 1 & 1 & 1 & 0 & 0 & 0 & 0\\
    1 & 1 & 0 & 0 & 1 & 1 & 0 & 0 \\
    1 & \color{IMSGreen}1 &  &  &  &  &  & \\
    0 & \color{IMSGreen}0 &  &  &  &  &  & \\
    0 & \color{IMSGreen}0 &  &  &  &  &  & \\
    0 & \color{IMSGreen}0 &  &  &  &  &  & \\
    0 & \color{IMSGreen}0 &  &  &  &  &  & 
  \end{bmatrix}.
  \]
  Note that $u_0 = u_1$, but it contradicts our assumption that there are no identical columns in $U$. Thus this case is not possible.
  \case
  There is no row with hamming weight $8$. Up to row and column permutation, $\lVert u_0^\dagger \rVert = 4$.
\setcounter{subcases}{0}
\subcase
There is a column with hamming weight $8$, consider
\[
U=\begin{bmatrix}
    1 & 1 & 1 & 1 & 0 & 0 & 0 & 0\\
    1 &  &  &  &  &  &  & \\
    1 &  &  &  &  &  &  & \\
    1 &  &  &  &  &  &  & \\
    1 &  &  &  &  &  &  & \\
    1 &  &  &  &  &  &  & \\
    1 &  &  &  &  &  &  & \\
    1 &  &  &  &  &  &  & 
  \end{bmatrix} \ \xrightarrow[\cref{lem:collision}]{\cref{lem:weight}} \ \begin{bmatrix}
    1 & 1 & 1 & 1 & 0 & 0 & 0 & 0\\
    1 & \color{IMSGreen}1 & \color{IMSGreen}0 & \color{IMSGreen}0 & \color{IMSGreen}1 & \color{IMSGreen}1 & \color{IMSGreen}0 & \color{IMSGreen}0\\
    1 & \color{IMSGreen}1 &  &  &  &  &  & \\
    1 & \color{IMSGreen}1 &  &  &  &  &  & \\
    1 & \color{IMSGreen}0 &  &  &  &  &  & \\
    1 & \color{IMSGreen}0 &  &  &  &  &  & \\
    1 & \color{IMSGreen}0 &  &  &  &  &  & \\
    1 & \color{IMSGreen}0 &  &  &  &  &  & 
  \end{bmatrix} \ \xrightarrow[\cref{lem:collision}]{\cref{lem:weight}} 
 \]
 \[
 \begin{bmatrix}
    1 & 1 & 1 & 1 & 0 & 0 & 0 & 0\\
    1 & 1 & 0 & 0 & 1 & 1 & 0 & 0\\
    1 & 1 & \color{IMSGreen}0 & \color{IMSGreen}0 & \color{IMSGreen}0 & \color{IMSGreen}0 & \color{IMSGreen}1 & \color{IMSGreen}1 \\
    1 & 1 & \color{IMSGreen}1 &  &  &  &  & \\
    1 & 0 &  &  &  &  &  & \\
    1 & 0 &  &  &  &  &  & \\
    1 & 0 &  &  &  &  &  & \\
    1 & 0 &  &  &  &  &  & 
  \end{bmatrix} \ \xrightarrow{\cref{lem:collision}}\ \begin{bmatrix}
    1 & 1 & 1 & 1 & 0 & 0 & 0 & 0\\
    1 & 1 & 0 & 0 & 1 & 1 & 0 & 0\\
    1 & 1 & 0 & 0 & 0 & 0 & 1 & 1 \\
    1 & 1 & 1 & \color{IMSGreen}1 & \color{IMSGreen}0 & \color{IMSGreen}0 & \color{IMSGreen}0 & \color{IMSGreen}0\\
    1 & 0 &  &  &  &  &  & \\
    1 & 0 &  &  &  &  &  & \\
    1 & 0 &  &  &  &  &  & \\
    1 & 0 &  &  &  &  &  & 
  \end{bmatrix}.
\]
Note that $u_0^\dagger = u_3^\dagger$, but it contradicts our assumption that there are no identical rows in $U$. Thus this case is not possible.
\subcase
There is no column with hamming weight $8$, consider
\[
U=\begin{bmatrix}
    1 & 1 & 1 & 1 & 0 & 0 & 0 & 0\\
    1 &  &  &  &  &  &  & \\
    1 &  &  &  &  &  &  & \\
    1 &  &  &  &  &  &  & \\
    0 &  &  &  &  &  &  & \\
    0 &  &  &  &  &  &  & \\
    0 &  &  &  &  &  &  & \\
    0 &  &  &  &  &  &  & 
  \end{bmatrix} \ \xrightarrow[\cref{lem:collision}]{\cref{lem:weight}} \ \begin{bmatrix}
    1 & 1 & 1 & 1 & 0 & 0 & 0 & 0\\
    1 & \color{IMSGreen}1 & \color{IMSGreen}0 & \color{IMSGreen}0 & \color{IMSGreen}1 & \color{IMSGreen}1 & \color{IMSGreen}0 & \color{IMSGreen}0\\
    1 & \color{IMSGreen}0 &  &  &  &  &  & \\
    1 & \color{IMSGreen}0 &  &  &  &  &  & \\
    0 & \color{IMSGreen}1 &  &  &  &  &  & \\
    0 & \color{IMSGreen}1 &  &  &  &  &  & \\
    0 & \color{IMSGreen}0 &  &  &  &  &  & \\
    0 & \color{IMSGreen}0 &  &  &  &  &  & 
  \end{bmatrix} \ \xrightarrow[\cref{lem:collision}]{\cref{lem:weight}} \ \begin{bmatrix}
    1 & 1 & 1 & 1 & 0 & 0 & 0 & 0\\
    1 & 1 & 0 & 0 & 1 & 1 & 0 & 0\\
    1 & 0 & \color{IMSGreen}1 & \color{IMSGreen}0 & \color{IMSGreen}1 & \color{IMSGreen}0 & \color{IMSGreen}1 & \color{IMSGreen}0\\
    1 & 0 & \color{IMSGreen}0 &  &  &  &  & \\
    0 & 1 & \color{IMSGreen}1 &  &  &  &  & \\
    0 & 1 & \color{IMSGreen}0 &  &  &  &  & \\
    0 & 0 & \color{IMSGreen}1 &  &  &  &  & \\
    0 & 0 & \color{IMSGreen}0 &  &  &  &  & 
  \end{bmatrix}
 \]
 \[
 \xrightarrow[\cref{lem:collision}]{\cref{lem:weight}} \ \begin{bmatrix}
    1 & 1 & 1 & 1 & 0 & 0 & 0 & 0\\
    1 & 1 & 0 & 0 & 1 & 1 & 0 & 0\\
    1 & 0 & 1 & 0 & 1 & 0 & 1 & 0\\
    1 & 0 & 0 & \color{IMSGreen}1 &  &  &  & \\
    0 & 1 & 1 & \color{IMSGreen}0 &  &  &  & \\
    0 & 1 & 0 & \color{IMSGreen}1 &  &  &  & \\
    0 & 0 & 1 & \color{IMSGreen}1 &  &  &  & \\
    0 & 0 & 0 & \color{IMSGreen}0 &  &  &  & 
  \end{bmatrix} \ \xrightarrow{\cref{lem:collision}} \ \begin{bmatrix}
    1 & 1 & 1 & 1 & 0 & 0 & 0 & 0\\
    1 & 1 & 0 & 0 & 1 & 1 & 0 & 0\\
    1 & 0 & 1 & 0 & 1 & 0 & 1 & 0\\
    1 & 0 & 0 & 1 & \color{IMSGreen}0 & \color{IMSGreen}1 & \color{IMSGreen}1 & \color{IMSGreen}0\\
    0 & 1 & 1 & 0 & x &  &  & \\
    0 & 1 & 0 & 1 &  &  &  & \\
    0 & 0 & 1 & 1 &  &  &  & \\
    0 & 0 & 0 & 0 &  &  &  & 
  \end{bmatrix}.
 \]
 \setcounter{ssubcases}{0}
 \ssubcase
 $x = 1$
 \[
U=\begin{bmatrix}
    1 & 1 & 1 & 1 & 0 & 0 & 0 & 0\\
    1 & 1 & 0 & 0 & 1 & 1 & 0 & 0\\
    1 & 0 & 1 & 0 & 1 & 0 & 1 & 0\\
    1 & 0 & 0 & 1 & 0 & 1 & 1 & 0\\
    0 & 1 & 1 & 0 & 1 &  &  & \\
    0 & 1 & 0 & 1 &  &  &  & \\
    0 & 0 & 1 & 1 &  &  &  & \\
    0 & 0 & 0 & 0 &  &  &  & 
  \end{bmatrix} \ \xrightarrow{\cref{lem:collision}} \ \begin{bmatrix}
    1 & 1 & 1 & 1 & 0 & 0 & 0 & 0\\
    1 & 1 & 0 & 0 & 1 & 1 & 0 & 0\\
    1 & 0 & 1 & 0 & 1 & 0 & 1 & 0\\
    1 & 0 & 0 & 1 & 0 & 1 & 1 & 0\\
    0 & 1 & 1 & 0 & 1 & \color{IMSGreen}0 & \color{IMSGreen}0 & \color{IMSGreen}1\\
    0 & 1 & 0 & 1 & \color{IMSGreen}0 &  &  & \\
    0 & 0 & 1 & 1 & \color{IMSGreen}0 &  &  & \\
    0 & 0 & 0 & 0 & \color{IMSGreen}1 &  &  &
  \end{bmatrix} \ \xrightarrow{\cref{lem:collision}} 
 \]
 \[
 \begin{bmatrix}
    1 & 1 & 1 & 1 & 0 & 0 & 0 & 0\\
    1 & 1 & 0 & 0 & 1 & 1 & 0 & 0\\
    1 & 0 & 1 & 0 & 1 & 0 & 1 & 0\\
    1 & 0 & 0 & 1 & 0 & 1 & 1 & 0\\
    0 & 1 & 1 & 0 & 1 & 0 & 0 & 1\\
    0 & 1 & 0 & 1 & 0 & \color{IMSGreen}1 & \color{IMSGreen}0 & \color{IMSGreen}1\\
    0 & 0 & 1 & 1 & 0 & \color{IMSGreen}0 &  & \\
    0 & 0 & 0 & 0 & 1 & \color{IMSGreen}1 &  &
  \end{bmatrix} \ \xrightarrow{\cref{lem:weight}} \ \begin{bmatrix}
    1 & 1 & 1 & 1 & 0 & 0 & 0 & 0\\
    1 & 1 & 0 & 0 & 1 & 1 & 0 & 0\\
    1 & 0 & 1 & 0 & 1 & 0 & 1 & 0\\
    1 & 0 & 0 & 1 & 0 & 1 & 1 & 0\\
    0 & 1 & 1 & 0 & 1 & 0 & 0 & 1\\
    0 & 1 & 0 & 1 & 0 & 1 & 0 & 1\\
    0 & 0 & 1 & 1 & 0 & 0 & \color{IMSGreen}1 & \color{IMSGreen}1\\
    0 & 0 & 0 & 0 & 1 & 1 & \color{IMSGreen}1 & \color{IMSGreen}1
  \end{bmatrix}=M.
 \]
 \ssubcase
 $x = 0$
 \[
U=\begin{bmatrix}
    1 & 1 & 1 & 1 & 0 & 0 & 0 & 0\\
    1 & 1 & 0 & 0 & 1 & 1 & 0 & 0\\
    1 & 0 & 1 & 0 & 1 & 0 & 1 & 0\\
    1 & 0 & 0 & 1 & 0 & 1 & 1 & 0\\
    0 & 1 & 1 & 0 & 0 &  &  & \\
    0 & 1 & 0 & 1 &  &  &  & \\
    0 & 0 & 1 & 1 &  &  &  & \\
    0 & 0 & 0 & 0 &  &  &  & 
  \end{bmatrix} \ \xrightarrow{\cref{lem:collision}} \ \begin{bmatrix}
    1 & 1 & 1 & 1 & 0 & 0 & 0 & 0\\
    1 & 1 & 0 & 0 & 1 & 1 & 0 & 0\\
    1 & 0 & 1 & 0 & 1 & 0 & 1 & 0\\
    1 & 0 & 0 & 1 & 0 & 1 & 1 & 0\\
    0 & 1 & 1 & 0 & 0 & \color{IMSGreen}1 & \color{IMSGreen}1 & \color{IMSGreen}0\\
    0 & 1 & 0 & 1 & \color{IMSGreen}1 &  &  & \\
    0 & 0 & 1 & 1 & \color{IMSGreen}1 &  &  & \\
    0 & 0 & 0 & 0 & \color{IMSGreen}0 &  &  & 
  \end{bmatrix} \ \xrightarrow[\cref{lem:collision}]{\cref{lem:weight}} \
 \]
 \[
 \begin{bmatrix}
    1 & 1 & 1 & 1 & 0 & 0 & 0 & 0\\
    1 & 1 & 0 & 0 & 1 & 1 & 0 & 0\\
    1 & 0 & 1 & 0 & 1 & 0 & 1 & 0\\
    1 & 0 & 0 & 1 & 0 & 1 & 1 & 0\\
    0 & 1 & 1 & 0 & 0 & 1 & 1 & 0\\
    0 & 1 & 0 & 1 & 1 & \color{IMSGreen}0 & \color{IMSGreen}1 & \color{IMSGreen}0\\
    0 & 0 & 1 & 1 & 1 & \color{IMSGreen}1 &  & \\
    0 & 0 & 0 & 0 & 0 & \color{IMSGreen}0 &  & 
  \end{bmatrix}  \ \xrightarrow{\cref{lem:weight}} \ \begin{bmatrix}
    1 & 1 & 1 & 1 & 0 & 0 & 0 & 0\\
    1 & 1 & 0 & 0 & 1 & 1 & 0 & 0\\
    1 & 0 & 1 & 0 & 1 & 0 & 1 & 0\\
    1 & 0 & 0 & 1 & 0 & 1 & 1 & 0\\
    0 & 1 & 1 & 0 & 0 & 1 & 1 & 0\\
    0 & 1 & 0 & 1 & 1 & 0 & 1 & 0\\
    0 & 0 & 1 & 1 & 1 & 1 & \color{IMSGreen}0 & \color{IMSGreen}0\\
    0 & 0 & 0 & 0 & 0 & 0 & \color{IMSGreen}0 & \color{IMSGreen}0
  \end{bmatrix} = N.
 \]
 Hence, when there are no identical rows nor columns in $U$, $U \in \mathcal{B}_1$ up to permutation.
\end{mycases}
\end{proof}

\begin{figure}[!ht]
  \begin{tikzpicture}[scale=0.8]

    \node[anchor=west] (case1) at (0,0) {$\bullet \ \lVert u_0^\dagger \rVert = \lVert u_1^\dagger \rVert = 8$};
      \node[anchor=west] (case1_1) at (3,-1) {$\lVert u_0 \rVert = 8$};
        \node[anchor=west] (case1_1_1) at (5,-2) {$\lVert u_1 \rVert = 8$};
          \node[anchor=west] (case1_1_1_1) at (7,-3) {$\lVert u_2 \rVert = 8$};
            \node[anchor=west] (case1_1_1_1_1) at (9,-4) {$\lVert u_3 \rVert = 8$};
              \node[anchor=west] (case1_1_1_1_1_1) at (11,-5) {$\lVert u_4 \rVert = 8 \color{IMSGreen}\Rightarrow A\color{black}$};
              \node[anchor=west] (case1_1_1_1_1_2) at (11,-6) {$\lVert u_4 \rVert = 4 \color{IMSGreen}\Rightarrow B \color{black}$};
            \node[anchor=west] (case1_1_1_1_2) at (9,-7) {$\lVert u_3 \rVert = 4$};
              \node[anchor=west] (case1_1_1_1_2_1) at (11,-8) {$\lVert u_2^\dagger \rVert = 8 \color{IMSGreen}\Rightarrow B\color{black}$};
              \node[anchor=west] (case1_1_1_1_2_2) at (11,-9) {$\lVert u_2^\dagger \rVert = 4 \color{IMSGreen}\Rightarrow \emptyset \color{black}$};
          \node[anchor=west] (u3_2) at (7,-10) {$\lVert u_2 \rVert = 4$};
              \node[anchor=west] (case1_1_1_2_1) at (9,-11) {$\lVert u_2^\dagger \rVert = 8 \color{IMSGreen}\Rightarrow B\color{black}$};
              \node[anchor=west] (case1_1_1_2_2) at (9,-12) {$\lVert u_2^\dagger \rVert = 4 \color{IMSGreen}\Rightarrow C \color{black}$};
        \node[anchor=west] (case1_1_2) at (5,-13) {$\lVert u_1 \rVert = 4$};
          \node[anchor=west] (case1_1_2_1) at (7,-14) {$\lVert u_2^\dagger \rVert = 8 \color{IMSGreen}\Rightarrow B \color{black}$};
          \node[anchor=west] (case1_1_2_2) at (7,-15) {$\lVert u_2^\dagger \rVert = 4$};
             \node[anchor=west] (case1_1_2_2_1) at (9,-16) {$\lVert u_2^\dagger \rVert = 8 \color{IMSGreen}\Rightarrow C\color{black}$};
             \node[anchor=west] (case1_1_2_2_2) at (9,-17) {$\lVert u_2^\dagger \rVert = 4 \color{IMSGreen}\Rightarrow C \color{black}$};
      \node[anchor=west] (case1_2) at (3,-18) {$\lVert u_0 \rVert = 4$};
        \node[anchor=west] (case1_2_1) at (5,-19) {$\lVert u_2^\dagger \rVert = 8$};
          \node[anchor=west] (case1_2_1_1) at (7,-20) {$\lVert u_4^\dagger \rVert = 4  \color{IMSGreen}\Rightarrow B \color{black}$};
          \node[anchor=west] (case1_2_1_2) at (7,-21) {$\lVert u_4^\dagger \rVert = 0  \color{IMSGreen}\Rightarrow K \color{black}$};
        \node[anchor=west] (case1_2_2) at (5,-22) {$\lVert u_2^\dagger \rVert = 4$};
          \node[anchor=west] (case1_2_2_1) at (7,-23) {$\lVert u_1 \rVert = \lVert u_2 \rVert = 8 \ , \ \lVert u_3 \rVert = 4 \color{IMSGreen}\Rightarrow C \color{black}$};
          \node[anchor=west] (case1_2_2_2) at (7,-24) {$\lVert u_1 \rVert = \lVert u_2 \rVert = \lVert u_3 \rVert = 4  \color{IMSGreen}\Rightarrow E \color{black}$};

  \draw[->] (case1) |- (case1_1);
  \draw[->] (case1) |- (case1_2);
  \draw[->] (case1_1) |- (case1_1_1);
  \draw[->] (case1_1) |- (case1_1_2);
  \draw[->] (case1_1_1) |- (case1_1_1_1);
  \draw[->] (case1_1_1) |- (u3_2);
  \draw[->] (case1_1_1_1) |- (case1_1_1_1_1);
  \draw[->] (case1_1_1_1) |- (case1_1_1_1_2);
  \draw[->] (case1_1_1_1_1) |- (case1_1_1_1_1_1);
  \draw[->] (case1_1_1_1_1) |- (case1_1_1_1_1_2);
  \draw[->] (case1_1_1_1_2) |- (case1_1_1_1_2_1);
  \draw[->] (case1_1_1_1_2) |- (case1_1_1_1_2_2);
  \draw[->] (u3_2) |- (case1_1_1_2_1);
  \draw[->] (u3_2) |- (case1_1_1_2_2);
  \draw[->] (case1_1_2) |- (case1_1_2_1);
  \draw[->] (case1_1_2) |- (case1_1_2_2);
  \draw[->] (case1_1_2_2) |- (case1_1_2_2_1);
  \draw[->] (case1_1_2_2) |- (case1_1_2_2_2);
  \draw[->] (case1_2) |- (case1_2_1);
  \draw[->] (case1_2) |- (case1_2_2);
  \draw[->] (case1_2_1) |- (case1_2_1_1);
  \draw[->] (case1_2_1) |- (case1_2_1_2);
  \draw[->] (case1_2_2) |- (case1_2_2_1);
  \draw[->] (case1_2_2) |- (case1_2_2_2);

\end{tikzpicture}
  \caption{Case distinction for $\lVert u_0^\dagger \rVert = \lVert u_1^\dagger \rVert = 8$.}
  \label{fig:casedist1}
\end{figure}

\begin{figure}[!ht]
  \begin{tikzpicture}[scale=0.75]

    \node[anchor=west] (case2) at (0,0) {$\bullet \ \lVert u_0^\dagger \rVert = \lVert u_1^\dagger \rVert = 4$};
      \node[anchor=west] (case2_1) at (3,-1) {$\lVert u_0 \rVert = 8$};
        \node[anchor=west] (case2_1_1) at (5,-2) {$\lVert u_1 \rVert = 8$};
          \node[anchor=west] (case2_1_1_1) at (7,-3) {$\lVert u_2 \rVert = 8$};
            \node[anchor=west] (case2_1_1_1_1) at (9,-4) {$\lVert u_2^\dagger \rVert = 8 \color{IMSGreen}\Rightarrow B \color{black}$};
            \node[anchor=west] (case2_1_1_1_2) at (9,-5) {$\lVert u_2^\dagger \rVert = 4$};
              \node[anchor=west] (case2_1_1_1_2_1) at (11,-6) {$\lVert u_3^\dagger \rVert = 8 \color{IMSGreen}\Rightarrow B \color{black}$};
              \node[anchor=west] (case2_1_1_1_2_2) at (11,-7) {$\lVert u_3^\dagger \rVert = 4 \color{IMSGreen}\Rightarrow K^\intercal \color{black}$};
           \node[anchor=west] (case2_1_1_2) at (7,-8) {$\lVert u_2 \rVert = 4$};
             \node[anchor=west] (case2_1_1_2_1) at (9,-9) {$\lVert u_2^\dagger \rVert = 8 \color{IMSGreen}\Rightarrow C \color{black}$};
            \node[anchor=west] (case2_1_1_2_2) at (9,-10) {$\lVert u_2^\dagger \rVert = 4 \color{IMSGreen}\Rightarrow D \color{black}$};
       \node[anchor=west] (case2_1_2) at (5,-11) {$\lVert u_1 \rVert = 4$};
         \node[anchor=west] (case2_1_2_1) at (7,-12) {$x=y=z=w=1$};
           \node[anchor=west] (case2_1_2_1_1) at (10,-13) {$\lVert u_2^\dagger \rVert = 8 \color{IMSGreen}\Rightarrow C \color{black}$};
           \node[anchor=west] (case2_1_2_1_2) at (10,-14) {$\lVert u_2^\dagger \rVert = 4 \color{IMSGreen}\Rightarrow D \color{black}$};
         \node[anchor=west] (case2_1_2_2) at (7,-15) {$x=y=z=w=0 \color{IMSGreen}\Rightarrow F \color{black}$};
    \node[anchor=west] (case2_2) at (3,-15) {$\lVert u_0 \rVert = 4$};
      \node[anchor=west] (case2_2_1) at (5,-16) {$x=z=w=1$};
        \node[anchor=west] (case2_2_1_1) at (8,-17) {$\lVert u_2^\dagger \rVert = 8$};
          \node[anchor=west] (case2_2_1_1_1) at (10,-18) {$\lVert u_1 \rVert = \lVert u_2 \rVert= 8, \ \lVert u_3 \rVert = 4 \color{IMSGreen}\Rightarrow C \color{black}$};
          \node[anchor=west] (case2_2_1_1_2) at (10,-19) {$\lVert u_1 \rVert = \lVert u_2 \rVert= \lVert u_3 \rVert=4 \color{IMSGreen}\Rightarrow E \color{black}$};
        \node[anchor=west] (case2_2_1_2) at (8,-20) {$\lVert u_2^\dagger \rVert = 4$};
          \node[anchor=west] (case2_2_1_2_1) at (10,-21) {$\lVert u_1 \rVert = \lVert u_2 \rVert= 8, \ \lVert u_3 \rVert = 4 \color{IMSGreen}\Rightarrow D \color{black}$};
          \node[anchor=west] (case2_2_1_2_2) at (10,-22) {$\lVert u_1 \rVert = \lVert u_2 \rVert= \lVert u_3 \rVert=4$};
            \node[anchor=west] (case2_2_1_2_2_1) at (14,-23) {$\lVert u_4^\dagger \rVert = 4 \color{IMSGreen}\Rightarrow I \color{black}$};
            \node[anchor=west] (case2_2_1_2_2_2) at (14,-24) {$\lVert u_4^\dagger \rVert= 0 \color{IMSGreen}\Rightarrow J \color{black}$};
      \node[anchor=west] (case2_2_2) at (5,-23) {$x=1,z=w=0$};
        \node[anchor=west] (case2_2_2_1) at (8,-24) {$\lVert u_1 \rVert= 8 \color{IMSGreen}\Rightarrow F \color{black}$};
        \node[anchor=west] (case2_2_2_2) at (8,-25) {$\lVert u_1 \rVert= 4$};
          \node[anchor=west] (case2_2_2_2_1) at (10,-26) {$x=y=z=w=1 \color{IMSGreen}\Rightarrow G \color{black}$};
          \node[anchor=west] (case2_2_2_2_2) at (10,-27) {$x=y=z=w=0 \color{IMSGreen}\Rightarrow H \color{black}$};

  \draw[->] (case2) |- (case2_1);
  \draw[->] (case2_1) |- (case2_1_1);
  \draw[->] (case2_1_1) |- (case2_1_1_1);
  \draw[->] (case2_1_1_1) |- (case2_1_1_1_1);
  \draw[->] (case2_1_1_1) |- (case2_1_1_1_2);
  \draw[->] (case2_1_1_1_2) |- (case2_1_1_1_2_1);
  \draw[->] (case2_1_1_1_2) |- (case2_1_1_1_2_2);
  \draw[->] (case2_1_1) |- (case2_1_1_2);
  \draw[->] (case2_1_1_2) |- (case2_1_1_2_1);
  \draw[->] (case2_1_1_2) |- (case2_1_1_2_2);
  \draw[->] (case2_1) |- (case2_1_2);
  \draw[->] (case2_1_2) |- (case2_1_2_1);
  \draw[->] (case2_1_2) |- (case2_1_2_2);
  \draw[->] (case2_1_2_1) |- (case2_1_2_1_1);
  \draw[->] (case2_1_2_1) |- (case2_1_2_1_2);
  \draw[->] (case2) |- (case2_2);
  \draw[->] (case2_2) |- (case2_2_1);
  \draw[->] (case2_2_1) |- (case2_2_1_1);
  \draw[->] (case2_2_1_1) |- (case2_2_1_1_1);
  \draw[->] (case2_2_1_1) |- (case2_2_1_1_2);
  \draw[->] (case2_2_1) |- (case2_2_1_2);
  \draw[->] (case2_2_1_2) |- (case2_2_1_2_1);
  \draw[->] (case2_2_1_2) |- (case2_2_1_2_2);
  \draw[->] (case2_2_1_2_2) |- (case2_2_1_2_2_1);
  \draw[->] (case2_2_1_2_2) |- (case2_2_1_2_2_2);
  \draw[->] (case2_2) |- (case2_2_2);
  \draw[->] (case2_2_2) |- (case2_2_2_1);
  \draw[->] (case2_2_2) |- (case2_2_2_2);
  \draw[->] (case2_2_2_2) |- (case2_2_2_2_1);
  \draw[->] (case2_2_2_2) |- (case2_2_2_2_2);

\end{tikzpicture}
  \caption{Case distinction for $\lVert u_0^\dagger \rVert = \lVert u_1^\dagger \rVert = 4$.}
  \label{fig:casedist2}
\end{figure}
\end{document}